\pdfoutput=1
\documentclass[12pt]{article}
\newcommand{\blind}{0}
\def\withLineNumbering{0}
\def\withFloats{1}

\if\withFloats0
	\usepackage[figuresonly,nolists,nomarkers]{endfloat}
	
\fi

\usepackage[left=1in, top=1in, right=1in, bottom=1in]{geometry}

\usepackage{amsfonts}
\usepackage{graphicx}
\usepackage{amsmath}
\usepackage{amsthm}
\usepackage{amssymb} 
\usepackage[outdir=./]{epstopdf}
\usepackage{subfigure} 
\usepackage{algorithm} 
\usepackage{algpseudocode} 
\usepackage{algorithmicx}
\usepackage{enumerate}
\usepackage{natbib}
\usepackage{bm}
\usepackage{dsfont} 
\usepackage[title]{appendix} 
\usepackage{xcolor}
\usepackage{placeins}
\usepackage[shortcuts]{extdash}
\usepackage{url}
\usepackage{tabulary}
\usepackage{booktabs}
\usepackage{multirow}
\usepackage[norefs, nocites]{refcheck}
\usepackage{multibib}
\newcites{supplement}{References for Supplement}

\usepackage{etoolbox}
\makeatletter
\patchcmd{\@makecaption}
{\parbox}
{\advance\@tempdima-\fontdimen2} 
{}{}
\makeatother

\newtheorem{theorem}{Theorem}[section]
\newtheorem{lemma}[theorem]{Lemma}
\newtheorem{proposition}[theorem]{Proposition}
\newtheorem{rule-of-thumb}[theorem]{Rule of Thumb}
\newtheorem{principle-behind}[theorem]{Principle Behind Prior-preconditioning}

\newcommand{\argmin}{\text{argmin}}
\newcommand{\transpose}{\text{\raisebox{.5ex}{$\intercal$}}}
\newcommand{\eigen}{\nu}
\newcommand{\ritz}{\widehat{\eigen}}

\newcommand{\polyspace}{\mathcal{P}}
\newcommand{\nunshrunk}{q}
\newcommand{\nlargest}{r}
\newcommand{\nsmallest}{s}
\newcommand{\nRitzIter}{k}
\newcommand{\nSubseqIter}{\ell}
\newcommand{\secmom}{\xi} 
\newcommand{\y}{\bm{y}}

\newcommand{\x}{\bm{x}}
\newcommand{\X}{\bm{X}}
\newcommand{\tilX}{\bm{\tilde{X}}}
\newcommand{\I}{\bm{I}}
\newcommand{\bv}{\bm{v}}
\newcommand{\br}{\bm{r}}
\newcommand{\bw}{\bm{w}}
\newcommand{\bb}{\bm{b}}
\newcommand{\bbTilde}{\tilde{\bb}}
\newcommand{\M}{\bm{M}}
\newcommand{\bA}{\bm{A}}
\newcommand{\bB}{\bm{B}}
\newcommand{\bF}{\bm{F}}
\newcommand{\localPrior}{\pi_{\rm loc}}
\newcommand{\globalPrior}{\pi_{\rm glo}}
\newcommand{\lshrink}{\lambda}
\newcommand{\Lshrink}{\Lambda}
\newcommand{\bbeta}{\bm{\beta}}

\newcommand{\balpha}{\bm{\alpha}}
\newcommand{\bgamma}{\bm{\gamma}}
\newcommand{\btilbeta}{\bm{\tilde{\beta}}}
\newcommand{\bomega}{\bm{\omega}}

\newcommand{\blshrink}{\bm{\lshrink}}
\newcommand{\bOmega}{\bm{\Omega}}
\newcommand{\bLshrink}{\bm{\Lshrink}}
\newcommand{\bPhi}{\bm{\Phi}}
\newcommand{\btilPhi}{\bm{\tilde{\Phi}}}

\newcommand{\normal}{\mathcal{N}}
\newcommand{\given}{\, | \,}

\newcommand{\blas}[1][]{\textsc{#1{b}}\textsc{las}}
\newcommand{\openblas}[1][]{\textsc{#1{o}}{\small pen}\blas}
\newcommand{\spgemm}[1][]{\textsc{#1{s}}{\small p}\textsc{gemm}}
\newcommand{\mkl}[1][]{\textsc{#1{m}}\textsc{kl}}
\newcommand{\cpu}{\textsc{cpu}}
\newcommand{\gb}{\textsc{gb}}

\newcommand{\nSignal}{s}
\newcommand{\counter}{N}
\newcommand{\nNonzero}{\counter_{\textrm{nonzero}}}
\newcommand{\nCg}{\counter_{\textrm{cg}}}
\newcommand{\nMatvec}{\counter_{\textrm{matvec}}}
\newcommand{\nOverlap}{\counter_{\textrm{overlap}}}
\newcommand{\nMatmat}{\counter_{\textrm{matmat}}}
\newcommand{\nCholesky}{\counter_{\textrm{cholesky}}}
\newcommand{\Xsynthetic}{\X_{\textrm{syn}}}
\newcommand{\ind}{\mathds{1}}

\renewcommand{\theequation}{\arabic{section}.\arabic{equation}}
\renewcommand{\thefigure}{\arabic{section}.\arabic{figure}}
\addtocounter{section}{0}
\if\withLineNumbering1
	\usepackage{lineno}

	\let\oldequation\equation
	\let\oldendequation\endequation
	\renewenvironment{equation}
	 {\begin{linenomath*}\oldequation}
	 {\oldendequation\end{linenomath*}\ignorespacesafterend}

	\makeatletter
	\renewenvironment{equation*}{%
	\begin{linenomath*}\
	\mathdisplay@push
	\st@rredtrue \global\@eqnswfalse
	\mathdisplay{equation*}%
	}{%
	\endmathdisplay{equation*}%
	\mathdisplay@pop
	\end{linenomath*}
	\ignorespacesafterend
	}
	\makeatother

	\makeatletter
	\patchcmd{\@startsection}{\@ifstar}{\nolinenumbers\@ifstar}{}{}
	\patchcmd{\@xsect}{\ignorespaces}{\linenumbers\ignorespaces}{}{}
	\makeatother

\else

	\newcommand{\nolinenumbers}{}
	\newcommand{\linenumbers}{}

\fi

\begin{document}

\def\spacingset#1{\renewcommand{\baselinestretch}%
{#1}\small\normalsize} \spacingset{1}


\newcommand{\titleString}{Prior-preconditioned conjugate gradient method for accelerated Gibbs sampling in `large $n$ \& large $p$' Bayesian sparse regression}

\if0\blind
{
  \title{\bf \titleString}
  \author{
    Akihiko Nishimura \\
    Department of Biostatistics, Johns Hopkins University \\
    and \\
    Marc A.\ Suchard \\
    Department of Biomathematics, Biostatistics, and Human Genetics, \\
    University of California - Los Angeles}
    \date{}
  \maketitle
  \vspace*{-1\baselineskip} 
} \fi

\if1\blind
{
  \bigskip
  \bigskip
  \bigskip
  \begin{center}
  \spacingset{1.6} 
    {\LARGE\bf \titleString}
  \end{center}
  \medskip
} \fi

\begin{abstract}
	In a modern observational study based on healthcare databases, the number of observations and of predictors typically range in the order of $10^5 \sim 10^6$ and of $10^4 \sim 10^5$.
	Despite the large sample size, data rarely provide sufficient information to reliably estimate such a large number of parameters.
	Sparse regression techniques provide potential solutions, one notable approach being the Bayesian method based on shrinkage priors.
	In the ``large $n$ \& large $p$'' setting, however, the required posterior computation encounters a bottleneck at repeated sampling from a high-dimensional Gaussian distribution, whose precision matrix $\bPhi$ is expensive to compute and factorize.
	In this article, we present a novel algorithm to speed up this bottleneck based on the following observation: we can cheaply generate a random vector $\bm{b}$ such that the solution to the linear system $\bPhi \bbeta = \bm{b}$ has the desired Gaussian distribution.
	We can then solve the linear system by the conjugate gradient (CG) algorithm through matrix-vector multiplications by $\bPhi$;
	this involves no explicit factorization or calculation of $\bPhi$ itself.
	Rapid convergence of CG in this context is guaranteed by the theory of \textit{prior-preconditioning} we develop.
	We apply our algorithm to a clinically relevant large-scale observational study with $n = 72{,}489$ patients and $p = 22{,}175$ clinical covariates, designed to assess the relative risk of adverse events from two alternative blood anti-coagulants. 
	Our algorithm demonstrates an order of magnitude speed-up in posterior inference, in our case cutting the computation time from two weeks to less than a day.
\end{abstract}

\noindent%
{\it Keywords:} Big Data, Conjugate gradient, Markov chain Monte Carlo, numerical linear algebra, sparse matrix, variable selection
\vfill

\linenumbers 

\newpage
\spacingset{1.5} 

\section{Introduction}
\label{sec:introduction}
	Given an outcome of interest $y_i$ and a large number of features $x_{i1}, \ldots, x_{ip}$ for $i = 1, \ldots, n$, the goal of sparse regression  is to find a small subset of these features that captures the principal relationship between the outcome and features.  Such a sparsity assumption is mathematical necessity when $p$ exceeds the sample size $n$. Even when $n > p$, however, the assumption often remains critical in improving the interpretability and stable estimation of regression coefficients $\bbeta$.
	This is especially true under the following conditions, either of which reduces the amount of information the data provides on the regression coefficients: 1) the design matrix $\X$ is sparse i.e.\ only a small fraction of the design matrix contains non-zero entries due to infrequent binary features, and/or 2) the binary outcome $\y$ is rare i.e.\ $y_i = 0$ for most of $i$'s.
	Sparse design matrices are extremely common in modern observational studies based on healthcare databases; while a large number of potential pre-existing conditions and available treatments exist, only a small subset of these applies to each patient \citep{schuemie2018empirical-calibration}.
	Rare binary outcomes are also common as many diseases of interest have low incidence rates among the population.

	A particular application considered in this manuscript is a comparative study of two blood anti-coagulants \textit{dabigatran} and \textit{warfarin}, using observational data from Truven Health MarketScan Medicare Supplemental and Coordination of Benefits Database. The anti-coagulants help prevent blood clot formation among patients with atrial fibrilation but come with risks of serious side effects. The goal of the study is to quantify which of the two drugs has a lower risk of gastrointestinal bleeding. The data set consists of $n = 72{,}489$ patients and $p = 22{,}175$ clinical covariates of potential relevance.

	To induce sparsity in the estimate of regression coefficient $\bbeta$, an increasingly common approach is the Bayesian method based on \textit{shrinkage priors}.
	This class of prior is often represented as a scale-mixture of Gaussians:
	\begin{equation*}
	\beta_j \given \lshrink_j, \tau
	\sim \normal(0, \tau^2 \lshrink_j^2), \
	\lshrink_j \sim \localPrior(\cdot), \
	\tau \sim \globalPrior(\cdot),
	\end{equation*}
	where $\tau$ and $\lshrink_j$ are unknown \textit{global} and \textit{local scale} parameters with priors $\localPrior(\cdot)$ and $\globalPrior(\cdot)$  \citep{carvalho2010horseshoe,polson2014bayes_bridge,bhattacharya2015dirichlet_laplace,bhadra2017lasso-meets-horseshoe}.
	Compared to more traditional ``spike-and-slab'' discrete-mixture priors, continuous shrinkage priors are typically more computationally efficient while maintaining highly desirable statistical properties \citep{bhattacharya2015dirichlet_laplace, pal2014ergodicity_of_shrinkage, datta2013horseshoe_property}. Despite the relative computational advantage, however, posterior inference under these priors still faces a serious scalability issue.
	In the blood anti-coagulant safety study, for instance, it takes over 200 hours on a modern high-end commodity desktop to run 10,000 iterations of the current state-of-the-art Gibbs sampler, even with optimized implementation (Section~\ref{sec:application}).

	We focus on sparse logistic regression in this article, but our Gibbs sampler acceleration technique applies whenever the likelihood function can be expressed as a Gaussian mixture.
	The data augmentation scheme of \cite{polson2013polya_gamma} makes a posterior under the logistic model amenable to Gibbs sampling as follows.
	Conditioning on a Polya-Gamma auxiliary parameter $\bomega$, the likelihood of a binary outcome $\y$ becomes
	\begin{equation}
	\label{eq:pg_augmented_likelihood}
	\tilde{y}_i \given \X, \bbeta, \bomega
		\sim \normal(\x_i^\transpose \bbeta, \omega_i^{-1})
		\ \text{ for } \
		\tilde{y}_i := \omega_i^{-1} \left(y_i - 1/2 \right).
	\end{equation}
	Correspondingly, the full conditional distribution of $\bbeta$ is given by
	\begin{equation}
	\label{eq:beta_conditional}
	\bbeta \given \bomega, \blshrink, \tau, \y, \X
		\sim \normal( \bPhi^{-1}\X^\transpose \bOmega \tilde{\y}, \bPhi^{-1})
		\ \text{ for } \
		\bPhi = \X^\transpose \bOmega \X + \tau^{-2} \bLshrink^{-2},
	\end{equation}
	where $\bOmega = \text{diag}(\bomega)$, a diagonal matrix with entries $\Omega_{ii} = \omega_i$, and $\bLshrink = \text{diag}(\blshrink)$.
	(See Supplement~\ref{sec:gibbs_sampler_details} for a complete description of the conditional updates within the Gibbs sampler.)

	The main computational bottleneck of the Gibbs sampler is the need to repeatedly sample from high-dimensional Gaussians of the form \eqref{eq:beta_conditional}. The standard algorithm requires $O(n p^2 + p^3)$ operations: $O(n p^2)$ for computing the term $\X^\transpose \bOmega \X$ and $O(p^3)$ for Cholesky factorization of $\bPhi$. These operations remain significant burden even with sparsity in $\X$ because computing times of sparse linear algebra operations are dominated not by the number of arithmetic operations but by latency in irregular data access \citep{dongarra2016hpc_connjugate_gradient, duff2017direct-sparse}.

	The ``large $n$ \& large $p$'' logistic regression problem considered in this article remains unsolved despite the recent computational advances.
	For $n \ll p$ cases, \cite{bhattacharya15fast_sampling} propose an algorithm to sample from \eqref{eq:beta_conditional} with only $O(n^2 p + n^3)$ operations.
	\cite{johndrow2018scalable_mcmc} reduce the $O(n^2 p)$ cost by replacing the matrix $\X \bLshrink^2 \X^\transpose$ with an approximation that can be computed with $O(n^2 k)$ operations for $k < p$. These techniques offer no reduction in computational cost for $n > p$ cases, however.
	\cite{hahn2018elliptical_slice_for_linear_model} propose a sampling approach for linear regression based on an extensive pre-processing of the matrix $\X^\transpose \X$ --- a trick limited in scope strictly to the Gaussian likelihood model.

	Proposed in this article is a novel algorithm to rapidly sample from a high-dimensional Gaussian distribution of the form \eqref{eq:beta_conditional} through the conjugate gradient (CG) method, using only a small number of matrix-vector multiplications $\bv \to \bPhi \bv$.
	Our algorithm requires no explicit formation of the matrix $\bPhi$ because we can compute $\bPhi \bv$ via operations $\bv \to \X \bv$ and $\bw \to \X^\transpose \bw$, along with element-wise vector multiplications.
	This is an important feature not only for computational efficiency but also for memory efficiency when dealing with a large and sparse design matrix $\X$.
	The matrix $\X^\transpose \bOmega \X$ and hence $\bPhi$ typically contain a much larger proportion of non-zero entries than $\X$, making it far more memory intensive to handle $\bPhi$ directly.
	For example, when $p = 10^5$, it would require 74.5 \gb{} of memory to store a $p \times p$ dense matrix $\bPhi$ in double-precision numbers.
	On the other hand, our algorithm can exploit a sparsity structure in $\X$ for both computational and memory efficiency.


	Practical utility of CG depends critically on effective \textit{preconditioning}, whose purpose is to speed up the algorithm by relating the given linear system to a modified one.
	Finding an effective preconditioner is a highly problem-specific task and is often viewed as ``a combination of art and science'' \citep{saad2003iterative-methods}.
	Exploiting fundamental features of sparse regression posteriors, we develop the \textit{prior-preconditioning} strategy tailored towards the linear systems in our specific context.
	We study its theoretical properties and demonstrate its superiority over general-purpose preconditioners in Bayesian sparse regression applications.

	 The rest of the paper is organized as follows. Section~\ref{sec:cg_sampler} begins by describing how to recast the problem of sampling from the distribution \eqref{eq:beta_conditional} as that of solving a linear system $\bPhi \bbeta = \bm{b}$.
	 The remainder of the section explains how to apply CG to rapidly solve the linear system, developing necessary theories along the way.
	 In Section~\ref{sec:cg_sampler_demonstration}, we use simulated data to study the effectiveness of our CG sampler in the sparse regression context. Also studied is how the behavior of CG depends on different preconditioning strategies.
	 In Section~\ref{sec:application}, we apply our algorithm to the blood anti-coagulant safety study, demonstrating an order of magnitude speed-up in the posterior computation.
	 Among the 22,175 predictors, the sparse regression posterior identifies age groups as significant source of treatment effect heterogeneity.

	 Our CG-accelerated Gibbs sampler is implemented as the \textit{bayesbridge} package available from Python Package Index (\url{pypi.org}). The source code is available at a GitHub repository \url{https://github.com/ohdsi/bayes-bridge}.

\section{Conjugate gradient sampler}
\label{sec:cg_sampler}

\subsection{Generating Gaussian vector as solution of linear system}
\label{sec:main_idea}
	The standard algorithm for sampling a multivariate-Gaussian requires the Cholesky factorization $\bPhi = \bm{L} \bm{L}^\transpose$ of its precision (or covariance) matrix \citep{rue2005gmrf}. When the precision matrix $\bm{\Phi}$ has a specific structure as in \eqref{eq:beta_conditional}, however, it turns out we can recast the problem of sampling from the distribution \eqref{eq:beta_conditional} to that of solving a linear system. This in particular obviates the need to compute and factorize $\bPhi$.
	\begin{proposition}
	\label{prop:gaussian_as_linear_system_solution}
		The following procedure generates a sample $\bbeta$ from the distribution \eqref{eq:beta_conditional}:
		\begin{enumerate}
			\item Generate $\bm{b}  \sim \normal\big( \X^\transpose \bOmega \tilde{\y}, \bm{\Phi} \big)$ by sampling independent Gaussian vectors $\bm{\eta} \sim \normal(\bm{0}, \bm{I}_n)$ and $\bm{\delta} \sim \normal(\bm{0}, \bm{I}_p)$ and then setting
			\begin{equation}
			\label{eq:cg_sampler_target_vector}
			\bm{b} = \X^\transpose \bOmega \tilde{\y} + \X^\transpose \bm{\Omega}^{1/2} \bm{\eta} + \tau^{-1} \bLshrink^{-1} \bm{\delta}.
			\end{equation}
			\item Solve the following linear system for $\bbeta$:
			\begin{equation}
			\label{eq:linear_system_for_cg_sampler}
			\bPhi \bbeta = \bm{b}
			\quad \, \text{where } \,
			\bPhi = \X^\transpose \bOmega \X + \tau^{-2} \bLshrink^{-2}.
			\end{equation}
		\end{enumerate}
	\end{proposition}
	\noindent The result follows immediately from basic properties of multivariate Gaussians.
	The Gaussian vector $\bm{b}$ has $\text{Var}(\bm{b}) = \bm{\Phi}$ and is generated with a computational cost negligible compared to computing and factorizing $\bPhi$.
	The solution to \eqref{eq:linear_system_for_cg_sampler} has the required covariance structure because $\textrm{Var}(\bPhi^{-1} \bm{b}) = \bPhi^{-1} \textrm{Var}(\bm{b}) (\bPhi^{-1})^{\transpose}$.

	\cite{bhattacharya15fast_sampling} propose a related algorithm which reduces the task of sampling a multivariate Gaussian to solving a $n \times n$ linear system.
	On the other hand, our algorithm reduces the task to solving a $p \times p$ system, which is smaller in size when $p < n$ and, more importantly, amenable to a fast solution via CG as we will show.

\subsection{Iterative method for solving linear system}
\label{sec:iterative_method_for_solving_linear_system}
	Proposition~\ref{prop:gaussian_as_linear_system_solution} is useful because solving the linear system \eqref{eq:linear_system_for_cg_sampler} can be significantly faster than the standard algorithm for sampling a Gaussian vector. We achieve this speed-up by applying the CG method \citep{hestenes1952cg,lanczos1952iterative}. CG belongs to a family of \textit{iterative methods} for solving a linear system. Compared to traditional direct methods, iterative methods are more memory efficient and, if the matrix $\bPhi$ has certain structures (Section~\ref{sec:cg_convergence_behavior}), can be significantly faster.

	Iterative methods have found applications in Gaussian process models, where optimizing the hyper-parameters of covariance functions requires solving linear systems involving large covariance matrices \citep{gibbs1997effcient}. Significant research has gone into how best to apply iterative methods in this specific context; see  \cite{stein2012difference-preconditioning}, \cite{sun2016estimating-eq-for-spatial}, and \cite{stroud2017gp-hyper-param-maximization} for example. Outside the Gaussian process literature, \cite{zhou2017iterative-method-for-variable-selection} use an iterative method to address the bottleneck of having to solve large linear systems when computing Bayes factors in a model selection problem.

	A novel feature of our work is the use of CG as a computational tool for Monte Carlo simulation.
	A related work is \cite{zhang2019practical-scalable-computing}, brought to our attention while we were preparing the first draft of our manuscript. They use the same idea as in Proposition~\ref{prop:gaussian_as_linear_system_solution} to generate a posterior sample from a Gaussian process model. However, they fail to investigate when and how CG delivers practical computational gains.
	Our work is distinguished by the development --- supported by both theoretical analysis and systematic empirical evaluations --- of a novel preconditioning technique tailored toward Bayesian sparse regression problems (Section~\ref{sec:prior_preconditioning} and \ref{sec:theory_of_prior_preconditioning}).
	In the process, we also compile a summary of the most practically useful of theoretical results regarding CG (Appendix~\ref{sec:theories_behind_cg}), which has previously been scattered across the literature, to facilitate potential applications of CG to a broader range of statistical problems.


	The CG method solves a linear system $\bPhi \bbeta = \bm{b}$ involving a positive definite matrix $\bPhi$ as follows. Given an initial guess $\bbeta_0$, which may be taken as $\bbeta_0 = \bm{0}$ for example, CG generates a sequence $\{\bbeta_k\}_{k = 1, 2, \ldots}$ of increasingly accurate approximations to the solution. The convergence of the CG iterates $\bbeta_k$'s is intimately tied to the \textit{Krylov subspace}
	\begin{equation*}
	\mathcal{K}(\bPhi, \br_0, k) = \textrm{span}
	\left\{ \bm{r}_0, \bPhi \bm{r}_0, \ldots, \bPhi^{k - 1} \br_0 \right\},
	\end{equation*}
	generated from the initial residual $\bm{r}_0 = \bPhi \bbeta_0 - \bm{b}$. With $\bbeta_0 + \mathcal{K}(\bPhi, \br_0, k)$ denoting an affine space $\{ \bbeta_0 + \bv : \bv \in \mathcal{K}(\bPhi, \br_0, k) \}$, the approximate solution $\bbeta_k$ satisfies the following optimality property in terms of a weighted $l^2$ norm $\| \cdot \|_{\bPhi}$, often referred to as the \textit{$\bPhi$-norm}:
	\begin{equation}
	\label{eq:cg_optimality_property}
	\bbeta_k = \argmin \left\{
		\left\| \bbeta' - \bbeta \right\|_{\bPhi} : \bbeta' \in \bbeta_0 + \mathcal{K}(\bPhi, \br_0, k)
	\right\}
	\, \text{ where }
	\left\| \br \right\|^2_{\bPhi} := \br^\transpose \bPhi \br.
	\end{equation}

	The optimality property \eqref{eq:cg_optimality_property} in particular implies that CG yields the exact solution after $p$ iterations.
	As evident from the pseudo-code in Supplement Section~\ref{sec:cg_pseudo_code}, the main computational cost of each update $\bbeta_k \to \bbeta_{k + 1}$ is a matrix-vector operation $\bv \to \bPhi \bv$.
	Consequently, the required number of arithmetic operations to run $p$ iterations of the CG update is comparable to that of a direct linear algebra method.
	For a typical precision matrix $\bPhi$ in the conditional distribution \eqref{eq:beta_conditional}, however, we can induce rapid convergence of CG through the preconditioning strategy described in the next section.
	In our numerical results, we indeed find that the distribution of $\bbeta_k$ even for $k \ll p$ is indistinguishable from \eqref{eq:beta_conditional} for all practical purposes.

\subsection{Convergence of CG and its relation to eigenvalue distribution}
\label{sec:cg_convergence_behavior}
	The iterative solution $\{\bbeta_k\}_{k = 0, 1, 2, \ldots}$ often displays slow convergence when CG is directly applied to a given linear system. Section~\ref{sec:prior_preconditioning} covers the topic of how to induce more rapid CG convergence for the system \eqref{eq:linear_system_for_cg_sampler}. In preparation, here we describe how the convergence behavior of CG is related to the structure of the positive definite matrix $\bPhi$.

	CG convergence behavior is partially explained by the following well-known error bound in terms of the \textit{condition number} $\kappa(\bPhi)$, the ratio of the largest to smallest eigenvalue of $\bPhi$.
	\begin{theorem}
	\label{thm:cg_bound_by_sqrt_cond_num}
		Given a positive definite system $\bPhi \bbeta = \bm{b}$ and a starting vector $\bbeta_0$, the $k$-th CG iterate $\bbeta_k$ satisfies the following bound in its $\bPhi$-norm distance to the solution $\bbeta$:
		\begin{equation}
		\label{eq:cg_bound_by_sqrt_cond_num}
		\frac{
			\left\| \bbeta_k - \bbeta \right\|_{\bPhi}
		}{
			\left\| \bbeta_0 - \bbeta \right\|_{\bPhi}
		}
		\leq 2
		\left( \frac{
			\sqrt{\kappa(\bPhi)} - 1
		}{
			\sqrt{\kappa(\bPhi)} + 1
		} \right)^k.
		\end{equation}
	\end{theorem}
	\noindent See \cite{trefethen1997numerical_linalg} for a proof. Theorem~\ref{thm:cg_bound_by_sqrt_cond_num} guarantees fast convergence of the CG iterates when the condition number is small. On the other hand, a large condition number does not always prevent rapid convergence. This is because CG converges quickly also when the eigenvalues of $\bPhi$ are ``clustered.'' The following theorem quantifies this phenomenon, albeit in an idealized situation in which $\bPhi$ has exactly $k < p$ distinct eigenvalues.
	\begin{theorem}
	\label{thm:cg_convergence_for_low_rank_perturbation}
		If the positive definite matrix $\bPhi$ has only $k + 1$ distinct eigenvalues, then CG yields an exact solution within $k + 1$ iterations. In particular, the result holds if $\bPhi$ is a rank-$k$ perturbation of an identity i.e.\ $\bPhi = \bF \bF^\transpose + \I$ for $\bF \in \mathbb{R}^{p \times k}$.
	\end{theorem}
	\noindent See \cite{golub2012matrix} for a proof.

	Theorem~\ref{thm:cg_bound_by_sqrt_cond_num} and \ref{thm:cg_convergence_for_low_rank_perturbation} are arguably the most famous results on the convergence property of CG, perhaps because their conclusions are clear-cut and easy to understand. These results, however, fall short of capturing the most important aspects of CG convergence behavior in practice.
	To address this problem, we bring together the most useful of the known results scattered around the numerical linear algebra literature and summarize them as the following rule of thumb. All the statements below are made mathematically precise in Appendix~\ref{sec:theories_behind_cg}.
	\begin{rule-of-thumb}
	\label{quasi-thm:cg_convergence}
	Suppose that the eigenvalues $\eigen_p(\bPhi) \leq \ldots \leq \eigen_1(\bPhi)$ of $\bPhi$ are clustered in the interval $[\eigen_{p - \nsmallest}, \eigen_\nlargest]$ except for a small fraction of them. Then CG effectively ``removes'' the outlying eigenvalues exponentially quickly. Its convergence rate subsequently accelerates as if the condition number in Eq~\ref{eq:cg_bound_by_sqrt_cond_num} is replaced by the effective value $\eigen_\nlargest / \eigen_{p - \nsmallest}$. The $r$ largest eigenvalues are removed within $r$ iterations, while the same number of smallest eigenvalues tends to delay convergence longer.
	\end{rule-of-thumb}

\subsection{Preconditioning linear system to accelerate CG convergence}
\label{sec:prior_preconditioning}
	A \textit{preconditioner} is a positive definite matrix $\M$ chosen so that the \textit{preconditioned system}
	\begin{equation}
	\label{eq:preconditioned_system}
	\btilPhi \btilbeta = \bbTilde
		\quad \text{ for } \,
		\bm{\tilde{\Phi}} = \M^{-1/2} \bPhi \M^{-1/2}
		\text{ and }
		\bbTilde = \M^{-1/2} \bb
	\end{equation}
	leads to faster convergence of the CG iterates.
	In practice, the algorithm can be implemented so that only the operation $\bv \to \M^{-1} \bv$, and not $\M^{-1/2}$, is required to solve the preconditioned system \eqref{eq:preconditioned_system} via CG \citep{golub2012matrix}.
	This \textit{preconditioned CG} algorithm still returns a solution $\bbeta_k = \M^{-1/2} \btilbeta_k$ in terms of the original system.

	In light of Rule of Thumb~\ref{quasi-thm:cg_convergence}, an effective preconditioner should modify the eigenvalue structure of $\bPhi$ so that the preconditioned matrix $\bm{\tilde{\Phi}}$ has more tightly clustered eigenvalues except for a small number of outlying ones. Larger outlying eigenvalues are preferable over smaller ones, as smaller ones cause a more significant delay in CG convergence. Additionally, a choice of a preconditioner must take into consideration 1) the one-time cost of computing the preconditioner $\M$ and 2) the cost of operation $\bv \to \M^{-1} \bv$ during each CG iteration.

	In the contexts of Bayesian sparse regression, the linear system \eqref{eq:linear_system_for_cg_sampler} admits a deceptively simple yet highly effective preconditioner.
	As it turns out, the choice
	\begin{equation*}
	\M = \tau^{-2} \bLshrink^{-2}
	\end{equation*}
	yields a modified system \eqref{eq:preconditioned_system} with an eigenvalue structure ideally suited to CG.
	With a slight abuse of terminology, we call it the \textit{prior preconditioner} since it corresponds to the precision of $\bbeta \given \tau, \blshrink, \bomega \, ( \, \overset{d}{=} \bbeta \given \tau, \blshrink)$ before observing $\y$ and $\X$.
	Most existing preconditioners require explicit access to the elements of $\bPhi$ for their constructions \citep{golub2012matrix} and are thus useless when computing $\bPhi$ itself is a bottleneck.
	Arguably the only reasonable alternative here is the Jacobi preconditioner $\M = \textrm{diag}(\Phi_{11}, \ldots, \Phi_{pp})$, known as one of the most effective for $\bPhi$ with large diagonals.
	Our numerical results clearly show superior performances of the prior preconditioner, however (Section~\ref{sec:cg_convergence_on_simulated_data} and \ref{sec:speed_up_from_cg_acceleration_on_ohdsi_data}).

	Noting that the prior-preconditioned matrix is given by
	\begin{equation}
	\label{eq:prior_preconditioned_matrix}
	\btilPhi = \tau^{2} \bLshrink \X^\transpose \bOmega \X \bLshrink + \I_p,
	\end{equation}
	we can heuristically motivate the preconditioner as follows.
	When employing the shrinkage prior, we expect posterior draws of $\tau \blshrink$ to satisfy $\tau \lshrink_j \approx 0$ except for a relatively small subset $\{j_1, \ldots, j_k\}$ of $j = 1, \ldots, p$.
	The $(i, j)$-th entry of the matrix $\tau^{2} \bLshrink \X^\transpose \bOmega \X \bLshrink$ is given by
	\begin{equation*}
	\left( \tau^{2} \bLshrink \X^\transpose \bOmega \X \bLshrink \right)_{i, j}
		= (\tau \lshrink_i) (\tau \lshrink_j) \left( \X^\transpose \bOmega \X \right)_{ij},
	\end{equation*}
	which is small when either $\tau \lshrink_i \approx 0$ or $\tau \lshrink_j \approx 0$. Hence the entries of $\tau^{2} \bLshrink \X^\transpose \bOmega \X \bLshrink$ are small away from the $k \times k$ block corresponding to the indices $\{j_1, \ldots, j_k\}$. In general, smaller entries of a matrix have less contributions to the eigenvalue structures of the entire matrix \citep{golub2012matrix}. This means that the prior-preconditioned matrix \eqref{eq:prior_preconditioned_matrix} can be thought of as a perturbation of the identity with a matrix of approximate low-rank structure.\footnote{%
		It is too naive, however, to deduce that we obtain a good approximation to $\btilPhi$ by zeroing out $\tau \lshrink_j$'s below some threshold.
		We show in Supplement Section~\ref{supp:block_preconditioner} that such approximation is typically of a poor quality.
	}
	As such, $\btilPhi$ can be expected to have eigenvalues clustered around 1, except for a small number of larger ones.

	Alternatively, we can also motivate the prior-preconditioner as follows.
	Bayesian sparse regression achieves posterior sparsity because the shrinkage prior dominates the likelihood for all but a small number of coefficients.
	In other words, the posterior looks a lot like the prior except in a small number of directions.
	As explained in Supplement Section~\ref{sec:principle_behind_prior_preconditioning}, this phenomenon translates into the eigenvalues of the prior-preconditioned matrix $\btilPhi$ clustering around 1.
	Since this heuristics is based on expected behavior of a posterior under a strongly informative prior in general, it suggests that prior-preconditioning may be applicable beyond the sparse regression context, e.g.\ to a Gaussian process model like that of \cite{zhang2019practical-scalable-computing}.

%

\subsection{Theory of prior-preconditioning and role of posterior sparsity}
\label{sec:theory_of_prior_preconditioning}
	We now formally quantify the eigenvalue structure of the matrix \eqref{eq:prior_preconditioned_matrix}.
	\begin{theorem}
	\label{thm:clustering_of_eigenvalues}
		Let $\lshrink_{(k)} = \lshrink_{j_k}$ denote the $k$-th largest element of $\{\lshrink_1, \ldots, \lshrink_p\}$. The eigenvalues of the prior-preconditioned matrix \eqref{eq:prior_preconditioned_matrix} satisfies
		\begin{equation*}
		1
			\leq \eigen_{k}(\btilPhi)
			\leq 1 + \tau^2 \lshrink_{(k)}^2 \, \eigen_1 \! \left( \X^\transpose \bOmega \X  \right)
		\end{equation*}
		for $k = 1, \ldots, p$. In fact, the following more general bounds hold. Let $\bA_{(-k)}$ denote the $(p - k) \times (p - k)$ submatrix of a given matrix $\bA$ corresponding to the row and column indices $j_{k + 1}, \ldots, j_p$. With this notation, we have
		\begin{equation}
		\label{eq:prior_preconditioned_eigenvalue_bounds}
		1
			\leq \eigen_{k + \ell}(\btilPhi)
			\leq 1 + \tau^2 \lshrink_{(k)}^2 \, \eigen_{\ell + 1} \! \left( (\X^\transpose \bOmega \X )_{(-k)} \right)
			\leq 1 + \tau^2 \lshrink_{(k)}^2 \, \eigen_{\ell + 1} \! \left( \X^\transpose \bOmega \X  \right)
		\end{equation}
		for any $k \geq 1$ and $\ell \geq 0$ such that $1 \leq k + \ell \leq p$.
	\end{theorem}
	\noindent Theorem~\ref{thm:clustering_of_eigenvalues} guarantees tight clustering of the eigenvalues of the prior\-/preconditioned matrix --- and hence rapid convergence of CG --- when most of $\tau \lshrink_j$'s are close to zero.
	We can also relate the prior-preconditioned CG approximation error directly to the decay rate in $\tau \lshrink_{(k)}$'s:
	\begin{theorem}
	\label{thm:cg_error_bound_under_prior_preconditioner}
		The prior-preconditioned CG applied to \eqref{eq:linear_system_for_cg_sampler} yields iterates satisfying the following bound for any $m, m' \geq 0$:
		\begin{equation}
		\label{eq:cg_error_bound_under_prior_preconditioner}
		\begin{aligned}
		\frac{
			\left\| \bbeta_{m + m'} - \bbeta \right\|_{\bPhi}
		}{
			\left\| \bbeta_0 - \bbeta \right\|_{\bPhi}
		}
			\leq 2 \left( \frac{
				\widetilde{\kappa}_{m}^{1/2} - 1
			}{
				\widetilde{\kappa}_{m}^{1/2} + 1
			} \right)^{m'}
			\text{ where }
			\widetilde{\kappa}_{m} =
			1 + \displaystyle{\min_{k + \ell = m}} \tau^2 \lshrink_{(k+1)}^2
			\eigen_{\ell + 1} \! \left( (\X^\transpose \bOmega \X )_{(-k)} \right).
		\end{aligned}
	\end{equation}
	\end{theorem}
	\noindent See Appendix~\ref{sec:proofs} for proofs of Theorem~\ref{thm:clustering_of_eigenvalues} and \ref{thm:cg_error_bound_under_prior_preconditioner}.

	To illustrate the implication of Theorem~\ref{thm:cg_error_bound_under_prior_preconditioner} in concrete terms, suppose that a posterior draw $\tau, \blshrink, \bomega$ satisfies $\tau^2 \lshrink_{(m+1)}^2 \eigen_{1} \! \left( \X^\transpose \bOmega \X \right) \leq 100$ for some $m$. In this case, we have $\log_{10} \big(\widetilde{\kappa}_{m}^{1/2} - 1 \big) \big/ \big(\widetilde{\kappa}_{m}^{1/2} + 1 \big) \leq - 0.086$. So the bound of Theorem~\ref{thm:cg_error_bound_under_prior_preconditioner} implies
	\begin{equation*}
	\frac{
		\left\| \bbeta_{m + m'} - \bbeta \right\|_{\bPhi}
	}{
		\left\| \bbeta_0 - \bbeta \right\|_{\bPhi}
	}
		\leq 2 \cdot 10^{- 0.086 m'}.
	\end{equation*}
	After $m + 100$ iterations, therefore, the CG approximation error in the $\bPhi$-norm is guaranteed to be reduced by a factor of $2 \cdot 10^{-8.6} \approx 10^{-8.3}$ relative to the initial error.

	We have so far stated our theoretical results in purely linear algebraic languages.
	We now summarize our discussions in a more statistical language, providing a practical guideline on the CG sampler performance in the sparse regression context.
	\begin{rule-of-thumb}
	\label{quasi-thm:sparsity_and_cg_convergence}
	The prior-preconditioned CG applied to the linear system \eqref{eq:linear_system_for_cg_sampler} converges rapidly when the posterior of $\bbeta$ concentrates on sparse vectors. 
	As the sparsity of $\bbeta$ increases, the convergence rate of the CG sampler also increases.
	\end{rule-of-thumb}
	The statements above are born out by illustrative examples of Section~\ref{sec:cg_sampler_demonstration} using synthetic sparse regression posteriors. As we have seen, the statements can be made more precise in terms of the decay rate in the ordered statistics $\tau \lshrink_{(k)}$ of a posterior sample $\tau \blshrink$ (Rule of Thumb~\ref{quasi-thm:cg_convergence}, Theorem~\ref{thm:clustering_of_eigenvalues}, and Theorem~\ref{thm:cg_error_bound_under_prior_preconditioner}).
	We also note that, while our theoretical results hold for any values of $\bomega$, $\tau$, $\blshrink$, and $\bb$, these quantities are random within a sparse regression Gibbs sampler.
	Even with substantial variation in these random quantities, however, we consistently observe fast decay in all $\tau \lshrink_{(k)}$ and rapid CG convergence at every iteration.
	In fact, we rarely observe a deviation of more than $5 \sim 10\%$ from the average number of CG iterations at stationarity -- see Supplement Section~\ref{supp:cg_acceleration_mechanism_details}.


\subsection{Computational complexity of prior-preconditioned CG}
\label{sec:complexity_analysis}
	Based on the discussion of Section~\ref{sec:theory_of_prior_preconditioning}, we may crudely quantify the number of prior-preconditioned CG iterations required for updating $\bbeta$ within a sparse regression Gibbs sampler as $O(s)$, where $s$ is the number of $\tau \lshrink_j$'s --- and hence of $\beta_j$'s --- significantly away from $0$.
	As the cost of each CG iteration is dominated by the operations $\bv \to \X \bv$ and $\bw \to \X^\transpose \bw$, both of which require $O(np)$ floating point operations, the $O(s)$ CG iterations translate to the overall computational complexity of $O(nps)$.
	The cost of prior-preconditioned CG thus can be far smaller than the $O(np^2 + p^3)$ cost of the standard method as $s \ll p$ in many applications.\footnote{%
		While this is a useful qualitative comparison, we also note that the number of floating point operations is an imperfect proxy for the actual computing time on modern hardware.
		See Supplement Section~\ref{sec:flop_vs_actual_time}.
	}

\subsection{Practical details on deploying CG for sparse regression}
\label{sec:cg_sampler_practical_details}
	While prior preconditioning is undoubtedly the most essential ingredient, there remain a few more important details in applying the CG sampler to sparse regression posterior computation. These are 1) a choice of the initial CG vector $\bbeta_0$, 2) a termination criterion for CG, and 3) handling of regression coefficients with uninformative priors. We discuss them briefly here and defer more thorough discussions to Supplement Section~\ref{supp:cg_sampler_practical_details}.

	A choice of the initial vector has little effect on the eventual exponential convergence rate of CG and, while not to be neglected, is nowhere as consequential as that of the preconditioner \citep{meurant2006lancoz-and-cg}. In fact, we find that any reasonable choice such as $\bbeta_0 = \bm{0}$ works fine in our numerical results, with more elaborate choices resulting in $\lesssim 10\%$ improvement in performance (Section~\ref{supp:cg_init_vector}).

	In its typical applications, CG is terminated when the $\ell^2$-norm of the residual $\br_k = \bPhi \bbeta_k - \bb$ falls below some prespecified threshold. Utility of $\| \br_k \|$ as an error metric is dubious for the purpose of the CG sampler, however.
	We instead propose the prior-preconditioned residual $\bm{\tilde{r}}_k = \btilPhi \bm{\tilde{\beta}}_k - \bm{\tilde{b}}$ as a more tailored alternative, its squared norm being an approximate upper bound to $\sum_j \secmom_j^{-2} \left( \bbeta_k - \bbeta \right)_j^2$ with $\xi_j^2 = \mathbb{E} \big[ \, \beta_j^2 \given \bomega, \blshrink, \tau, \y, \X \big]$ (Section~\ref{supp:termination_criteria_for_cg_sampler}).
	Specifically, we use and validate the termination criterion $p^{-1/2} \| \bm{\tilde{r}}_k  \|_{2} \leq 10^{-6}$ in our numerical studies.

	When preconditioning CG, regression coefficients with uninformative priors, such as the intercept, must be handled differently from those under shrinkage.
	We can accommodate such coefficients by augmenting the prior-preconditioner with another diagonal matrix.
	We analyze the eigenvalues of the resulting preconditioned matrix and show that, by virtue of CG's ability to quickly remove the outlying eigenvalues (Rule of Thumb~\ref{quasi-thm:cg_convergence}), the convergence rate remains fast and is robust to the precise choice of the diagonal matrix (Section~\ref{supp:preconditioning_unshrunken_coef}).

\section{Simulation study of CG sampler performance}
\label{sec:cg_sampler_demonstration}
	  We study the CG sampler performance when applied to actual posterior conditional distributions of the form \eqref{eq:beta_conditional}.
	  We specifically focus on the prior-preconditioned CG's performance in solving the linear system \eqref{eq:linear_system_for_cg_sampler} since this directly translates into the performance of the CG-accelerated Gibbs sampler.\footnote{
	  	We confirm in Supplement Section~\ref{sec:cg_sampler_accuracy_and_speed_for_synthetic_data} that samples generated by the CG sampler is statistically indistinguishable from those generated by the direct linear algebra method.
	  	Also in Section~\ref{sec:cg_sampler_accuracy_and_speed_for_synthetic_data}, we show how the CG sampler's performance demonstrated here translates into actual gains in terms of computing time.
	  }
	  We simulate data with varying numbers of non-zero coefficients and confirm how sparsity in regression coefficients translates into faster CG convergence as predicted by Theorem~\ref{thm:clustering_of_eigenvalues} and Rule of Thumb~\ref{quasi-thm:sparsity_and_cg_convergence}. We also illustrate how the convergence rates are affected by different preconditioning strategies and by corresponding eigenvalue distributions of the preconditioned matrices.

\subsection{Choice of shrinkage prior: Bayesian bridge}
\label{sec:bayesian_bridge}
	Among existing global-local shrinkage priors, we adopt the Bayesian bridge prior of \cite{polson2014bayes_bridge} as the corresponding Gibbs sampler allows for collapsed updates of $\tau$ to improve mixing \citep{polson2014bayes_bridge}.
	The Bayesian bridge Gibbs sampler is in fact uniformly ergodic when the prior tails are properly modified \citep{nishimura2019regularized_shrinkage}.

	Under the Bayesian bridge, the local scale $\lambda_j$'s are given a prior $\pi(\lambda_j) \propto \lambda_j^{-2} \pi_{\rm st}(\lambda_j^{-2} / 2)$ where $\pi_{\rm st}(\cdot)$ is an alpha-stable distribution with index of stability $\alpha / 2$.
	The corresponding prior on $\beta_j \given \tau$, when $\lambda_j$ is marginalized out, is
	\begin{equation*}
	\pi(\beta_j \given \tau)
		\propto \tau^{-1} \exp \left( - | \beta_j / \tau |^\alpha \right).
	\end{equation*}
	The distribution of $\beta_j \given \tau$ becomes ``spikier'' as $\alpha \to 0$, placing greater mass around $0$ while inducing heavier tails.
	In typical applications, the data favors the values $\alpha < 1$ but only weakly identifies $\alpha$  \citep{polson2014bayes_bridge}, so in this article we simply fix $\alpha = 1/2$ except when a smaller value seems warranted; see Section~\ref{sec:prior_and_computation}.

\subsection{Experimental set-up}
\label{sec:simulation_setup}
	We generate synthetic data of sample size $n = 25{,}000$ with the number of predictors $p = 10{,}000$. In constructing a design matrix $\X$, we emulate a model from factor analysis \citep{jolliffe2002principal}. We first sample a set of $m = 99$ orthonormal vectors $\bm{u}_1, \ldots, \bm{u}_{m} \in \mathbb{R}^p$ uniformly from a Stiefel manifold. We then set the predictor $\x_i$ for the $i$-th observation as
	\begin{equation}
	\label{eq:factor_model_design}
	\x_i = \sum_{\ell = 1}^{99} f_{i, \ell} \bm{u}_\ell + \bm{\epsilon}_i
		\ \text{ for }
		f_{i, \ell} \sim \normal \left(0, (100 - \ell + 1)^2 - 1 \right)
		\text{ and }
		\bm{\epsilon}_i \sim \normal \left( \bm{0}, \I_p \right).
	\end{equation}
	This is equivalent to sampling $\x_i \sim \normal \! \left(\bm{0}, \bm{U} \bm{D} \bm{U}^\transpose \right)$ for a diagonal matrix $\bm{D}$ with $\sqrt{D_{\ell \ell}} = \max \{100 - \ell + 1, 1\}$ and and orthonormal matrix $\bm{U}$ sampled uniformly from the space of orthonormal matrices. We then center and standardize the predictors as is commonly done before applying sparse regression \citep{hastie2009statistical-learning}.

	The above process yields a design matrix $\X$ with moderate correlations among the $p$ predictors --- the distribution of pairwise correlations is approximately Gaussian centered around 0 with the standard deviation of 0.13. Based on this design matrix $\X$, we simulate three different binary outcome vectors by varying the number of non-zero regression coefficients. More specifically, we consider a sparse regression coefficient $\bbeta_{\rm true}$ with $\beta_{\textrm{true}, \, j} = \mathds{1}\{ j \leq \nSignal\}$ with varying numbers of signals $\nSignal = 10, 20$, and $50$. In all three scenarios, the binary outcome $\y$ is generated from the logistic model as $y_i \given \bbeta_{\rm true}, \x_i \sim \textrm{Bernoulli}(p_i)$ for $\textrm{logit}(p_i) = \x_i^\transpose \bbeta_{\rm true}$.

	For each synthetic data set, we obtain a posterior sample of $\bomega, \tau, \blshrink \given \y, \X$ by running the Polya-Gamma augmented Gibbs sampler with the brute-force direct linear algebra to sample $\bbeta$ from its conditional distribution \eqref{eq:beta_conditional}. We confirm the convergence of the Markov chain by examining the traceplot of the posterior log-density of $\bbeta, \tau \given \y, \X$. 
	Having obtained a posterior sample ($\bomega, \tau, \blshrink$), we sample the vector $\bm{b}$ as in \eqref{eq:cg_sampler_target_vector} and apply CG to the linear system \eqref{eq:linear_system_for_cg_sampler}. We compare the CG iterates $\{\bbeta_k \}_{k \geq 0}$ to the exact solution $\bbeta_{\textrm{direct}}$ obtained by solving the same system with the Cholesky-based direct method. We repeat this process for eight random replications of the right-hand vector $\bm{b}$.

\subsection{Results}
\label{sec:cg_convergence_on_simulated_data}

\subsubsection*{Convergence rates and eigenvalue distributions}
\label{sec:convergence_rate_and_eig_dist_on_simulated_data}
	Figure~\ref{fig:cg_convergence_plot_for_simulated_data} shows the CG approximation error as a function of the number of CG iterations, whose cost is dominated by matrix-vector multiplications $\bv \to \bPhi \bv$.
	 We characterize the approximation error as the relative error $| (\bbeta_k - \bbeta_{\textrm{direct}})_j / (\bbeta_{\textrm{direct}})_j |$ averaged across all the coefficients.
	 Each line on the plot shows the geometric average of this error metric over the eight random replications of \nolinebreak $\bm{b}$.
	 The CG convergence behavior observed here remains qualitatively similar regardless of the error metric choice and varies little across the different right-hand vectors;
	 see Supplement Section~\ref{supp:alt_cg_convergence_plot_for_simulated_data}.
	 We also observe there that, while the error $| (\bbeta_k - \bbeta_{\textrm{direct}})_j / (\bbeta_{\textrm{direct}})_j |$ varies substantially across the index $j$, the coefficient-specific errors all decay at roughly uniform rates as a function of the number of CG iterations.
	\begin{figure}
		\centering
		\includegraphics[width=.7\linewidth]{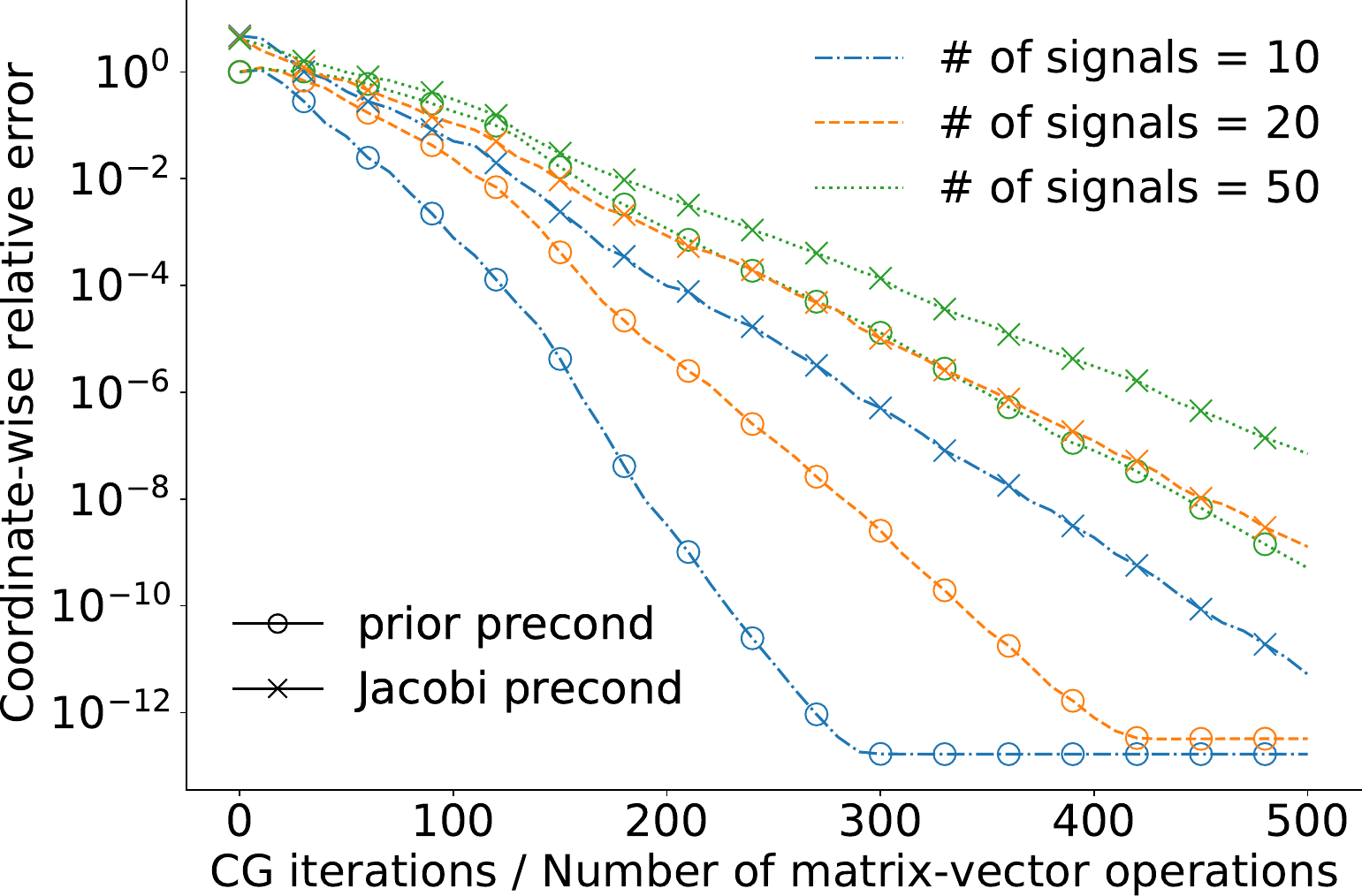}
		\caption{Plot of the CG approximation error vs.\ the number of CG iterations. The CG sampler is applied to the posterior conditionals based on synthetic data. The different line styles correspond to the different numbers of true signals in underlying data. The circle and cross markers denote the uses of the prior and Jacobi preconditioners. }
		\label{fig:cg_convergence_plot_for_simulated_data}
	\end{figure}

	We first focus on the approximation errors under the prior preconditioner, indicated by the lines with circles. After $k \ll p = 10{,}000$ matrix-vector operations, the distance between $\bbeta_k$ and $\bbeta_{\textrm{direct}}$ is already orders of magnitudes smaller than typical Monte Carlo errors.
	With additional CG iterations, the distance reaches the machine precision level; notice the eventual ``plateaus'' achieved under the prior preconditioner in the $\nSignal = 10 \text{ and } \nSignal = 20$ cases.

	Figure~\ref{fig:cg_convergence_plot_for_simulated_data} also shows the approximation errors under the Jacobi preconditioner $\M = \textrm{diag}(\Phi_{11}, \ldots, \Phi_{pp})$ which, as discussed in Section~\ref{sec:prior_preconditioning}, is the only reasonable alternative when using the CG sampler for the applications considered in this article.
	The prior preconditioner is clearly superior, with the difference in convergence speed more pronounced when true regression coefficients are sparser. Studying the eigenvalue distributions of the respective preconditioned matrices provides further insight into the observed convergence behaviors. Figure~\ref{fig:eigval_distributions_for_data_simulated_with_10_correlated_signal} (a) \& (b) show the eigenvalue distributions of the preconditioned matrices based on a posterior sample from the synthetic data with $\nSignal = 10$. The trimmed version of the histograms highlight the tails of the distributions. The prior preconditioner induces the distribution with a tight cluster around 1 (or 0 in the $\log_{10}$ scale) with a relatively small number of large ones, confirming the theory developed in Section~\ref{sec:theory_of_prior_preconditioning}. On the other hand, the Jacobi preconditioner induces a more spread-out distribution, problematically introducing quite a few small eigenvalues that delay the CG convergence (Rule of Thumb~\ref{quasi-thm:cg_convergence}).

	\begin{figure}
		\begin{minipage}{.6\linewidth}
		\subfigure[
			Based on synthetic data with 10 non-zero coefficients.
		]{
			\includegraphics[width=.95\linewidth]{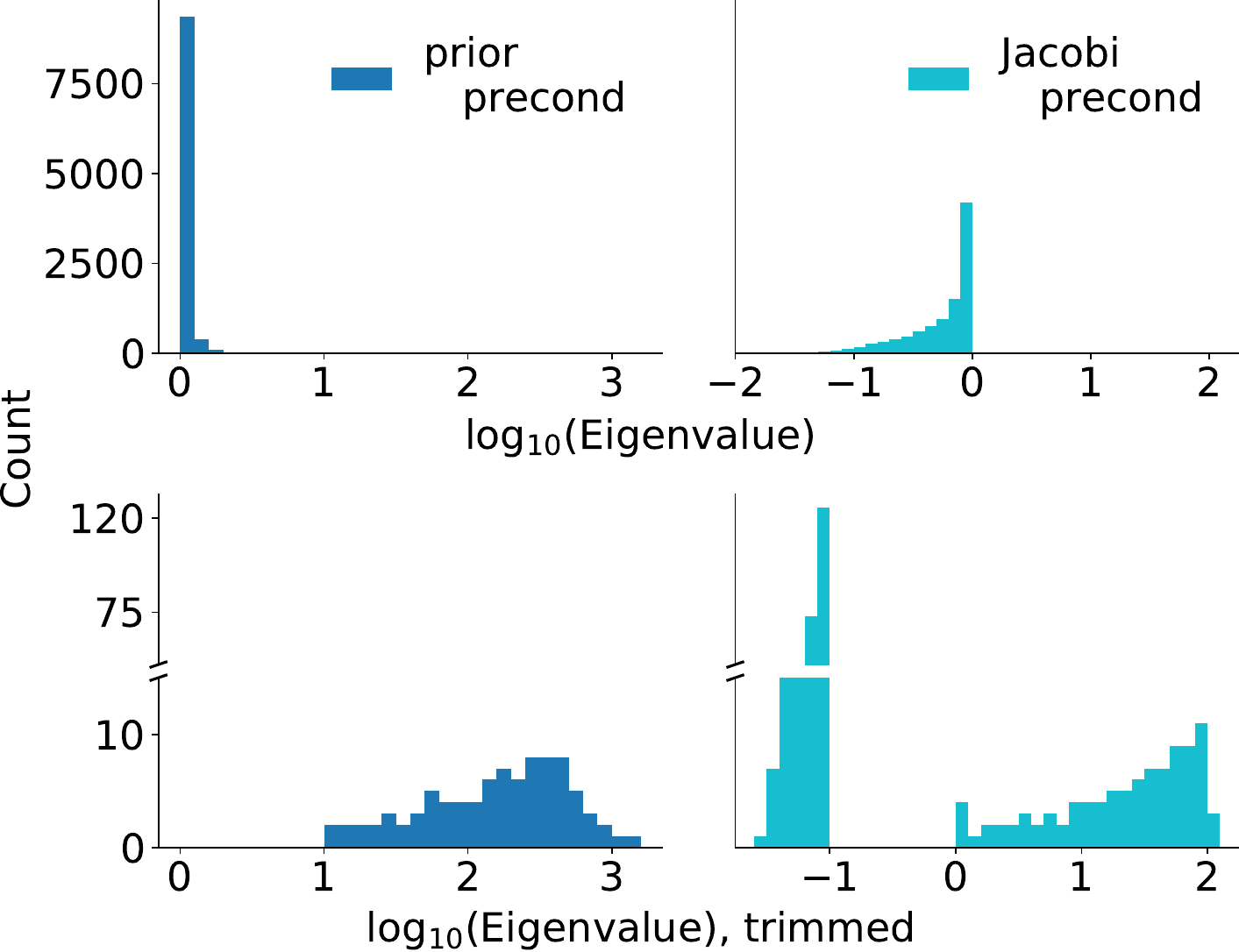}
		}

		\subfigure[Based on synthetic data with 50 non-zero coefficients.]{
			\includegraphics[width=.95\linewidth]{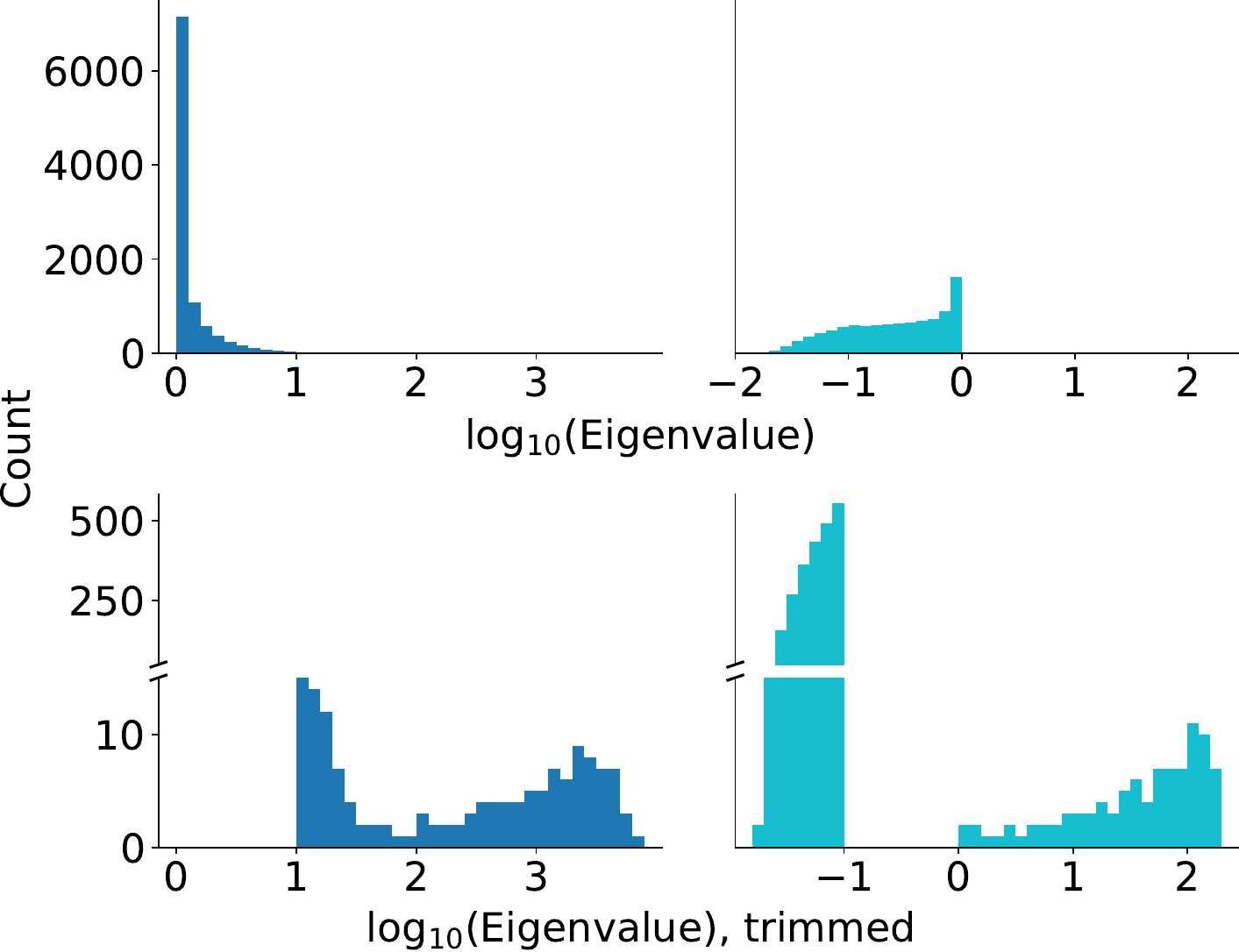}
		}
		\end{minipage}
		\begin{minipage}{.38\linewidth}
		\caption{Histograms of the eigenvalues of the preconditioned matrices.
		The eigenvalues under the prior preconditioner are shown on the left and those under the \mbox{Jacobi} on the right. Shown on the lower rows are the trimmed versions of the histograms, in which we remove the eigenvalues in the range $[0, 1]$ in the $\log_{10}$ scale for the prior preconditioner and those in the range $[-1, 0]$ for the \mbox{Jacobi}. The $y$-axes for the trimmed histograms have intermediate values removed to make small counts more visible. The width of the bins are kept constant throughout so that the $y$-axis values of the bars are proportional to probability densities and thus can be compared meaningfully across the plots.}
		\end{minipage}
		\label{fig:eigval_distributions_for_data_simulated_with_10_correlated_signal}
	\end{figure}

\subsubsection*{Relationship between convergence rate and posterior sparsity}
	Finally, we turn our attention to the relationship, as seen in Figure~\ref{fig:cg_convergence_plot_for_simulated_data}, between CG convergence rate and sparsity in the underlying true regression coefficients. The convergence is clearly quicker when the true regression coefficients are sparser. To understand this relationship, it is informative to look at the values of $\tau \lshrink_j = \textrm{var}(\beta_j \given \tau, \blshrink)^{1/2}$ drawn from the respective posterior distributions. Figure~\ref{fig:prior_sd_plot_with_varying_number_of_signals} plots the values of $\tau \lshrink_j$ for $j = 1, \ldots, 250$ corresponding to the first 250 coefficients. We use two different $y$-scales for $\nSignal = 10$  and $\nSignal = 50$, shown on the left and right respectively, to facilitate qualitative comparison between the two cases. As expected, the posterior sample from the synthetic data with a larger number of signals has a larger number of $\tau \lshrink_j$'s away from zero. These relatively large $\tau \lshrink_j$'s contribute to the delayed convergence of CG (Theorem~\ref{thm:clustering_of_eigenvalues} and Rule of Thumb~\ref{quasi-thm:cg_convergence}).

	\begin{figure}
		\hspace*{-.01\linewidth}
		\begin{minipage}{.55\linewidth}
			\includegraphics[width=\linewidth]{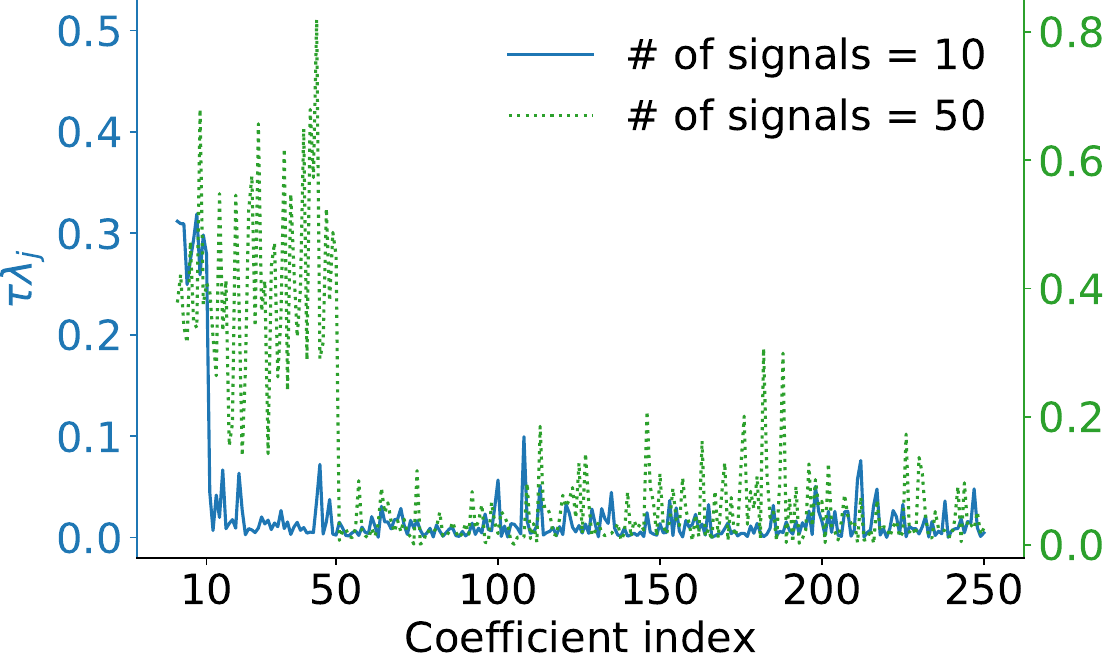}
		\end{minipage}
		~~
		\begin{minipage}{.4\linewidth}
			\vspace*{-.5\baselineskip}
			\caption{Plot of the posterior samples of $\tau \lshrink_j$'s for $j = 1, \ldots, 250$. The solid blue line and dashed green line correspond to the data sets simulated with $\beta_{\mathrm{true}, j} = 1$ for $j \leq 10$ and for $j \leq 50$ respectively.}
			\label{fig:prior_sd_plot_with_varying_number_of_signals}
		\end{minipage}
	\end{figure}

	A more significant cause of the delay, however, is the fact that the shrinkage prior yields weaker shrinkage on the zero coefficients when there are a larger number of signals. With a close look at Figure~\ref{fig:prior_sd_plot_with_varying_number_of_signals}, one can see that $\tau \lshrink_1, \ldots, \tau \lshrink_\nSignal$ corresponding to the true signals are not as well separated from the rest of $\tau \lshrink_j$'s when $\nSignal = 50$. In fact, the histograms on the left of Figure~\ref{fig:prior_sd_hist_with_varying_number_of_signals} shows that the distribution of $\tau \lshrink_j$'s for $\nSignal = 50$ are shifted toward larger values compared to that for $\nSignal = 10$. This is mostly due to the posterior distribution of $\tau$ concentrating around a larger value --- the value of the posterior sample is $\tau \approx 2.0 \times 10^{-3}$ for the $\nSignal = 10$ case while $\tau \approx 6.7 \times 10^{-3}$  for the $\nSignal = 50$ case.

	It is also worth taking a closer look at the tail of the distribution of $\tau \lshrink_j$'s. The histograms on the right of Figure~\ref{fig:prior_sd_hist_with_varying_number_of_signals} show the distribution of the 250 largest $\tau \lshrink_j$'s. The figure makes it clear that $\tau \lshrink_j$'s corresponding to the true signals are much more well separated from the rest when $\nSignal = 10$. Overall, the slower decay in the largest values of $\tau \lshrink_{j}$'s results in the eigenvalues of the preconditioned matrices having a less tight cluster around 1; compare the eigenvalue distributions of Figure~\ref{fig:eigval_distributions_for_data_simulated_with_10_correlated_signal}.(a) \& (b) to those of Figure~\ref{fig:eigval_distributions_for_data_simulated_with_10_correlated_signal}.(c) \& (d).

	\begin{figure}
		\hspace*{-.01\linewidth}
		\subfigure{
			\includegraphics[height=.21\textheight]{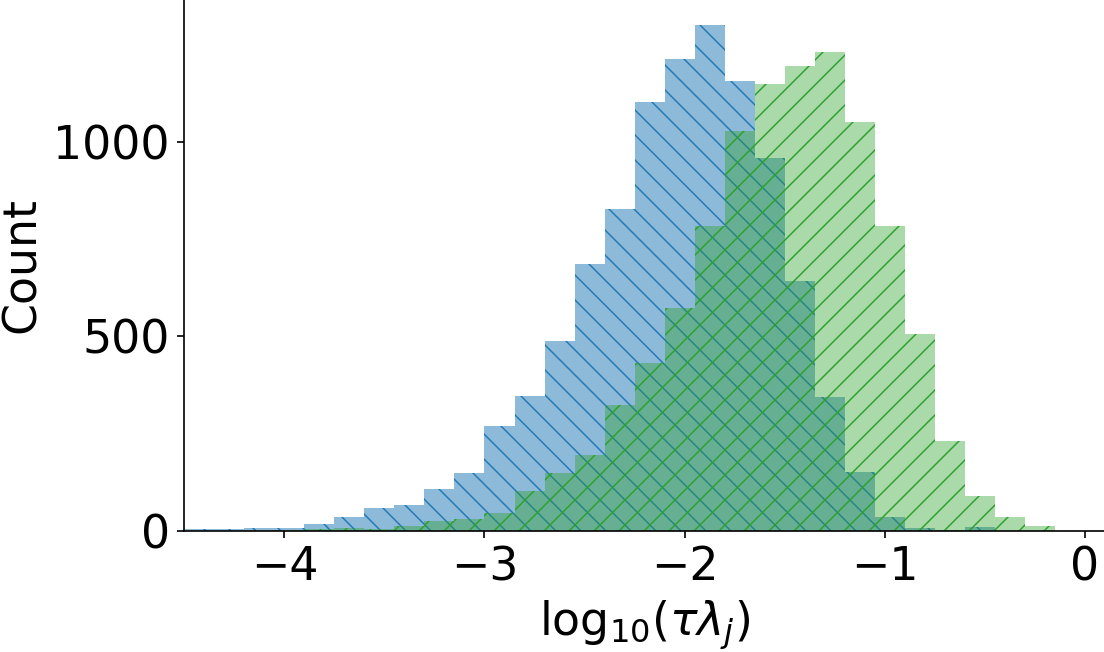}
		}
		\subfigure{
			\includegraphics[height=.21\textheight]{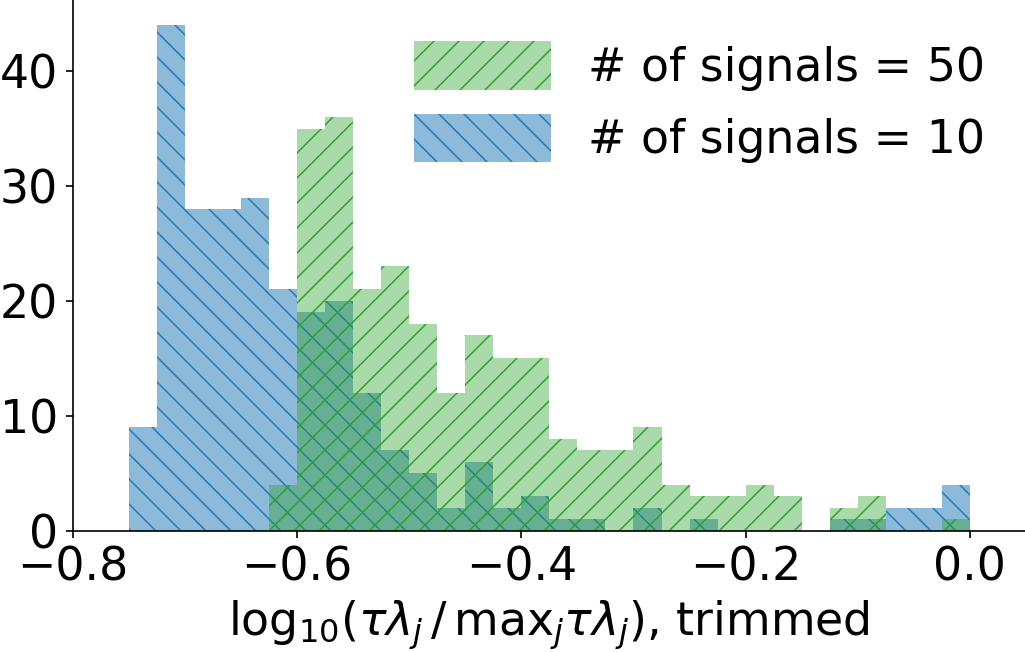}
		}
		\vspace*{-.5\baselineskip}
		\caption{%
			Histograms of the posterior samples of $\tau \lshrink_j$'s.
			The two different colors indicate the distinct posteriors with 10 and 50 non-zero regression coefficients.
			To better expose the relative tail behaviors in the two distributions of $\tau \lshrink_j$'s, the histograms on the right take the 250 largest values and plot their magnitudes relative to $\max_j \tau \lshrink_j$.
		}
		\label{fig:prior_sd_hist_with_varying_number_of_signals}
	\end{figure}

\subsubsection*{Comments on generalizability of conclusions from simulation study}
	We conclude by noting that the convergence rate of the CG sampler is also a function of signal strengths and correlation among the predictors, both of which affect the posterior sparsity in regression coefficients.
	For example in the propensity score model application of Section~\ref{sec:application}, despite 82 regression coefficients having posterior means of substantial magnitudes, the prior-preconditioned CG converges after $107 \sim 120$ iterations in 95\% of the cases.
	We also confirm that, when using a synthetic design matrix with independent columns, the CG sampler demonstrates much faster rates of convergence for the same numbers of signals (Supplement Section~\ref{sec:effects_of_correlation}).
	On the other hand, a synthetic design matrix with correlation structure more extreme than \eqref{eq:factor_model_design} leads to slower convergence rates for the same numbers of signals (Supplement Section~\ref{sec:effects_of_num_factors}).
	Finally, the posterior sparsity structure, and hence the CG sampler's performance, is also affected by a choice of shrinkage prior. 
	How this choice affects the posterior sparsity is difficult to quantify.
	The additional simulation studies using different priors (Supplement Section~\ref{supp:shrinkage_prior_choice_and_cg_performance} and \ref{supp:shrinkage_prior_choice_and_cg_performance_on_ohdsi_data}) indicate, however, that the main takeaway holds regardless: the sparser the posterior, the faster the CG sampler's convergence.



\newcommand{\ohdsi}{\textsc{ohdsi}}

\section{Application: comparison of alternative treatments}
\label{sec:application}
	In this section, we demonstrate the magnitude of speed-up delivered by CG-acceleration in modern large-scale applications.
	We apply Bayesian sparse logistic regression to conduct a comparative study of two blood anti-coagulants \textit{dabigatran} and \textit{warfarin}.
	The goal of the study is to quantify which of the two drugs have a lower risk of a potential side effect, gastrointestinal bleeding. This question has previously been investigated by \cite{graham2014dabigatran_vs_warfarin} and our analysis yields clinical findings consistent with theirs (Section~\ref{sec:dabigatran_vs_warfarin_results}).

	We are particularly interested in Bayesian sparse regression as a tool for the Observational Health Data Sciences and Informatics (\ohdsi{}) collaborative \citep{hripcsak2015ohdsi}.
	We therefore follow the \ohdsi{} protocol in pre-processing of the data as well as in estimating the treatment effect.
	In particular, sparse regression plays a critical role in eliminating hand-picking of confounding factors and of subgroups for testing treatment effect heterogeneity; this enables the application of a reproducible and consistent statistical estimation procedure to tens of thousands of observational studies \citep{tian18, schuemie2020legend}.

\subsection{Data set}
\label{sec:data_description}
	We extract patient-level data from Truven Health MarketScan Medicare Supplemental and Coordination of Benefits Database. In the database, we find $n = 72{,}489$ patients who became first-time users of either dabigatran or warfarin after diagnosis of atrial fibrillation. Among them, 19,768 are treated with dabigatran and the rest with warfarin. There are $p = 98{,}118$ predictors, consisting of clinical measurements, pre-existing conditions, as well as prior treatments and administered drugs --- all measured before exposure to dabigatran or warfarin. Following the \ohdsi{} protocol, we screen out the predictors observed in less than $0.1$\% of the cohort.  This reduces the number of predictors to $p = 22{,}175$. The precise definition of the cohort can be found
	at \url{http://www.ohdsi.org/web/atlas/#/cohortdefinition/{2978, 2979, 2981}}.

	Each patient is affected by only a small fraction of the potential pre-existing conditions and available treatments.
	The design matrix $\X$ therefore is sparse, with only 4\% of the entries being non-zero.
	Another noteworthy feature of the data is the low incidence rates of gastrointestinal bleeding; the outcome indicator $\y$ has non-zero entries $y_i = 1$  for only 713 out of 72,489 patients.

\subsection{Statistical approach: propensity score stratified regression}
\label{sec:treatment_effect_via_sparse_regression}
	To control for covariate imbalances between the dabigatran and warfarin users, we rely on propensity score method in estimating the treatment effect.
	The procedure involves two logistic models with large numbers of predictors, to deal with which we employ Bayesian sparse regression.
	We describe the procedure and essential ideas below but refer the readers to \cite{stuart2010causal-inference}, and the references therein for further details.

	Estimation of the treatment effect proceeds in two stages. First, the \textit{propensity score} $\mathbb{P}(\, T_i = 1 \given \x_i)$ of the treatment assignment to dabigatran is estimated by the logistic model
	\begin{equation}
	\label{eq:propensity_score_model}
	\textrm{logit} \! \left\{\mathbb{P}(\, T_i = 1 \given \x_i) \right\}
		= \beta_0 + \x_i^\transpose \bbeta.
	\end{equation}
	While not of direct interest within the propensity score method framework, identifying significant predictors of the score is highly relevant in the \ohdsi{} applications.
	Many of the databases are too small to fit the models with such large numbers of predictors, but the significant heterogeneity among them makes the joint estimation insensible \citep{hripcsak2016treatment_pathway}.
	Sparse regression provides a tool to screen out the predictors using the larger databases and use only the selected subset to estimate the scores within the smaller databases.

	After fitting the model \eqref{eq:propensity_score_model}, the quantiles of the estimated propensity scores are then used to stratify the population into subpopulations of equal sizes. Following a typical recommendation, we choose the number of strata as $M = 5$. Under suitable assumptions, conditioning on the strata indicator removes most of imbalances in the distributions of the predictors between the treatment ($T_i = 1$) and control ($T_i = -1$) groups.
	After the stratification, we can proceed to estimate the treatment effect via the logistic model without the main effect from the clinical covariate $\x_i$ \citep{tian2014simple_hete}:
	\begin{equation}
	\label{eq:treatment_effect_model}
	\textrm{logit} \! \left\{\mathbb{P}(\, y_i = 1 \given \balpha, \bgamma, s_i, \x_i) \right\}
		= \sum_{m = 1}^M \alpha_m \mathds{1}\!\left\{ s_i = m \right\}
				+ (\alpha_0 + \x_i^\transpose \bgamma) \, T_i,
	\end{equation}
	where a categorical variable $s_i$ denotes the strata membership of the $i$-th individual.
	The quantity $\alpha_0 + \x^\transpose \bgamma$ represents the treatment effect for a patient with covariate $\x$, with the feature $x_{ij}$ contributing to the treatment effect heterogeneity when $\gamma_j \neq 0$.
	The goal of sparse regression here is to identify such nonzero $\gamma_j$'s.


\subsection{Prior choice and posterior computation}
\label{sec:prior_and_computation}
	We fit the models \eqref{eq:propensity_score_model} and \eqref{eq:treatment_effect_model} using the Bayesian bridge shrinkage prior (Section~\ref{sec:bayesian_bridge}).
	For the main treatment and propensity score strata effects, we place weakly informative $\normal(0, 1)$ priors.
	For the global scale parameter, we use an objective prior $\pi(\tau) \propto 1 / \tau$ in the model \eqref{eq:propensity_score_model} \citep{berger2009reference-priors}.
	For the treatment effect model \eqref{eq:treatment_effect_model}, due to the low incidence rate in the outcome, we find the above prior choice to provide insufficient separation of significant predictors from the rest.
	We therefore use the bridge prior with $\alpha = 1 / 4$ and weakly informative conjugate prior $\phi = \tau^{-\alpha} \sim \textrm{Gamma}(\textrm{shape} = 339.8, \textrm{rate} = 26.58)$ so that $\log_{10}(\tau)$ has the prior mean of $-1.5$ and standard deviation of $0.5$.

	For posterior computation, we compare two Gibbs samplers that differ only in their methods for drawing $\bbeta$ from the conditional distribution \eqref{eq:beta_conditional}. One sampler uses the proposed CG sampler while the other uses a traditional direct method via Cholesky factorization. Sparse Cholesky methods offer no computational benefit here as the precision matrix, despite the sparsity in the design matrix $\X$, is almost completely dense (Supplement Section~\ref{sec:optimizing_linear_algebra}). We refer to the respective samplers as the \textit{CG-accelerated} and \textit{direct} Gibbs sampler. The other conditional updates follow the approaches described in \cite{polson2014bayes_bridge}; see Supplement Section~\ref{sec:gibbs_sampler_details} for the details.

	We implement the Gibbs samplers in Python and run on a 2015 iMac with an Intel Core i7 ``Skylake'' processor having four cores at 4 \textsc{gh}z and 32 \gb{} of memory.
	Linear algebra being the computational bottleneck, both samplers benefit from parallelization and we engage all the four cores.
	For the linear algebra operations, we interface our Python code with the Intel Math Kernel Library (\mkl{}) implementations of Basic Linear Algebra Subprograms (\blas) and sparse \blas{}, which proved computationally superior to alternatives in our preliminary benchmarking. 
	We use the sparse \blas{} for matrix-vector multiplications $\bv \to \X \bv$ and $\bw \to \X^\transpose \bw$ within the CG-accelerated Gibbs and for matrix-matrix multiplication $\X^\transpose \bOmega \X$ within the direct Gibbs.
	Exploiting the sparsity in $\X$ cuts down both computing time and memory usage by an order of magnitude.
	Details on how we optimized both Gibbs sampler are described in Supplement Section~\ref{sec:optimizing_linear_algebra}.

	We run the Gibbs samplers for 5,500 and 11,000 iterations for the propensity score and treatment effect model, discarding the first 500 and 1,000 as burn-ins.
	We confirm their convergences by examining the traceplots of the posterior log-density.
	We estimate the effective sample sizes (ESS) for all the regression coefficients using the R package \textsc{coda}.
	The smallest ESSs are found among the coefficients with bimodal posteriors, but their traceplots nonetheless indicate reasonable mixing.
	We find the minimum and median ESS to be 106.2 and 2484 for the propensity score model, and 86.04 and 2496 for the treatment effect model.

\subsection{CG acceleration magnitudes and posterior characteristics}
\label{sec:speed_up_from_cg_acceleration_on_ohdsi_data}
	The direct Gibbs sampler requires 106 and 212 hours for the propensity score and treatment effect model. On the other hand, the CG-accelerated sampler finishes in 11.4 and 11.3 hours, yielding \textbf{9.3-fold} and \textbf{18.8-fold} speed-ups.
	For both Gibbs samplers, the total computation times are dominated by the conditional updates of $\bbeta$.
	The magnitudes of CG-acceleration thus are determined by the CG convergence rate at each Gibbs iteration.

	In agreement with the theory and empirical results of Section~\ref{sec:theory_of_prior_preconditioning} and \ref{sec:cg_convergence_on_simulated_data}, the variability in the magnitudes of CG-acceleration can be explained by the posterior sparsity structures of the regression coefficients.
	For the propensity score model, 82 out of the 22,175 regression coefficients have their posterior mean magnitudes above 0.1, while 18,187 (82.0\%) of the coefficients below 0.01.
	For the treatment effect model, only 2 of the coefficients have the posterior mean magnitudes above 0.1, while 22{,}096 (99.6\%) below 0.01.
	We note that the individual posterior samples are much less sparse than the posterior mean.
	Under the treatment effect model, for example, the number of coefficients with magnitudes above 0.1 typically ranges from 265 to 529 while those below 0.01 from 16,172 to 17,632.

	For more in-depth analysis of the CG-acceleration mechanism, in Supplement Section~\ref{supp:cg_acceleration_mechanism_details} we examine the CG sampler behavior at each Gibbs iteration. In particular, we verify that the error metric discussed in Section~\ref{sec:cg_sampler_practical_details} works well in deciding when to terminate the CG iteration. We also confirm that the prior preconditioner continues to outperform the Jacobi in this real data example.

\subsection{Clinical conclusions from dabigatran vs.\ warfarin study}
\label{sec:dabigatran_vs_warfarin_results}
	The propensity score model finds substantial differences between the patients treated by dabigatran and warfarin. In particular, patients' covariate characteristics are predictive of the treatment assignments as seen in  Figure~\ref{fig:propensity_score_distribution}.
	\begin{figure}
		\begin{minipage}[t]{.48\linewidth}
			\vspace{0pt}
			\hspace{-.055\linewidth}
			\includegraphics[width=\linewidth]{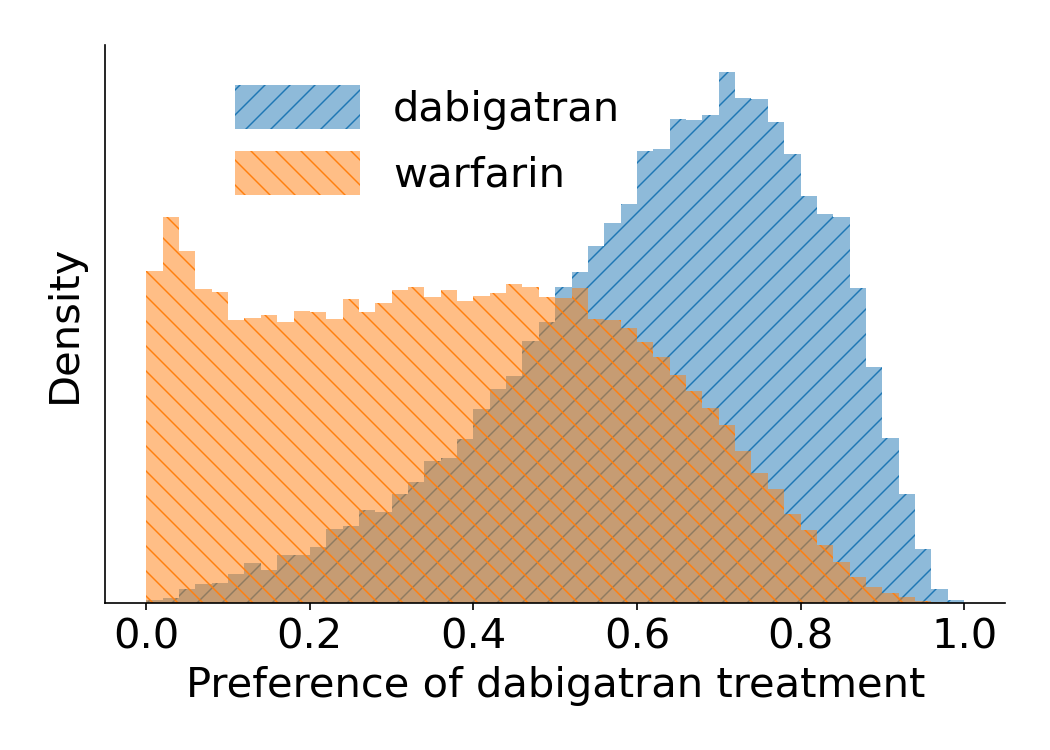}
			\vspace{-.25\baselineskip}
			\label{fig:propensity_score_distribution}
		\end{minipage}
		~
		\begin{minipage}[t]{.48\linewidth}
			\vspace*{\baselineskip}
			\caption{Normalized histogram of the estimated preference scores for each group. Preference score transforms raw propensity score to make it a more interpretable measure \citep{walker2013preference_score}.}
		\end{minipage}
	\end{figure}
	The two most significant predictors are the treatment year and age group. Both predictors have been encoded as binary indicators in the design matrix for simplicity, but the coefficients of categorical and ordinal predictors could have been estimated with shrinkage priors analogous to Bayesian grouped or fused lasso \citep{kyung2010bayes-penalized-regression, xu2015baye-group-lasso}. The posterior mean and 95\% credible intervals of the regression coefficients are shown in Figure~\ref{fig:posterior_mean_and_ci_for_propensity_model}. The figure plots the effect sizes relative to the year 2010 and the age group 65--69; when actually fitting the model, however, we use the most common category as a baseline for categorical variables.

	\begin{figure}
		\centering
		\subfigure{
			\includegraphics[height=.24\textheight]{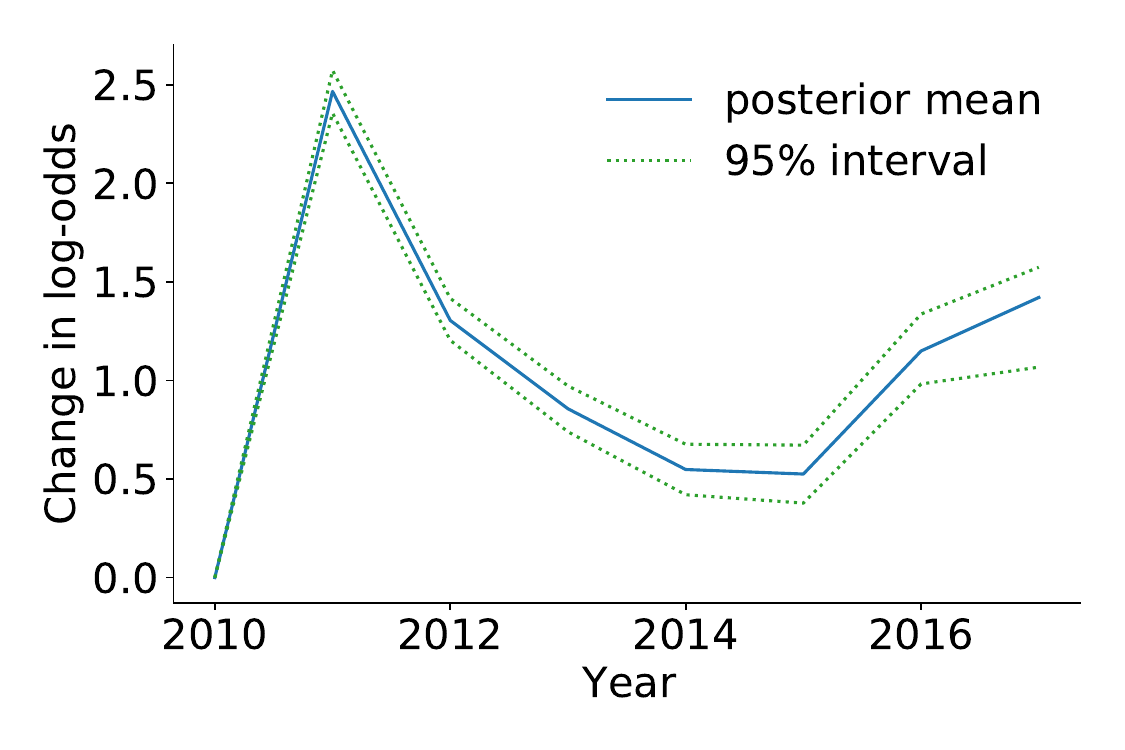}
		}
		\hspace{-.05\linewidth}
		\subfigure{
			\includegraphics[height=.24\textheight]{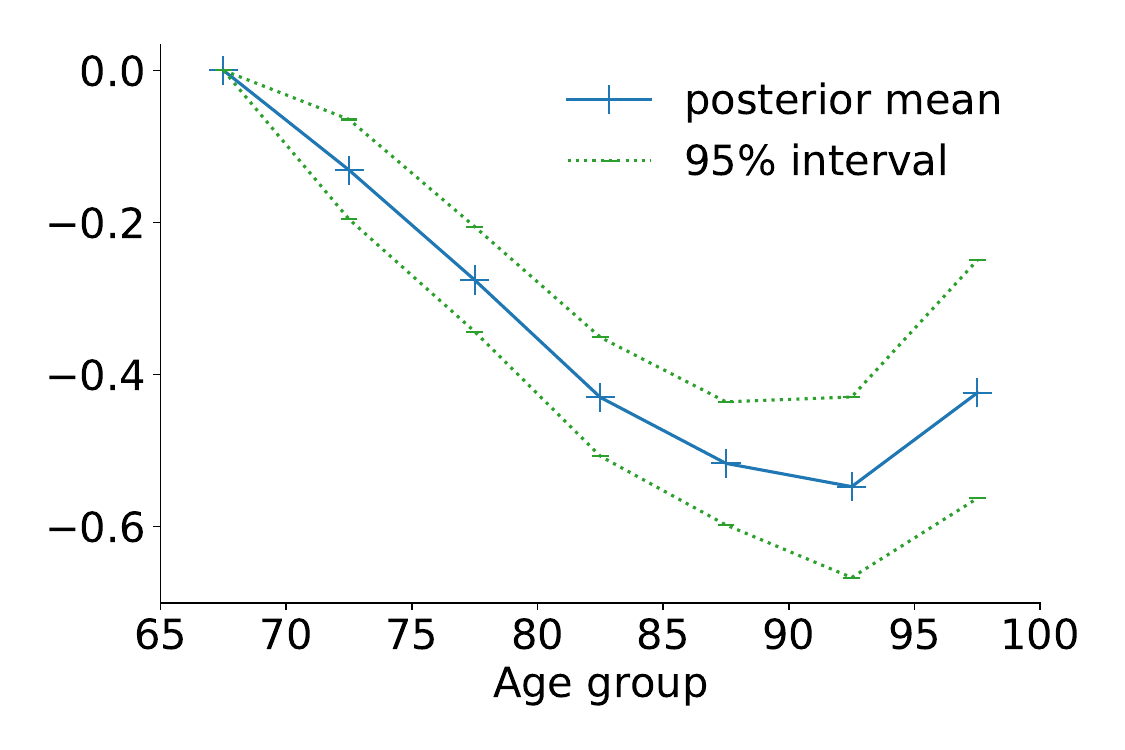}
		}
		\vspace{-1.75\baselineskip}
		\caption{Posterior means and 95\% credible intervals for the regression coefficients of the treatment year and age group indicators. The age groups are divided into 5-year windows.}
		\label{fig:posterior_mean_and_ci_for_propensity_model}
	\end{figure}

	For the treatment effect model, Figure~\ref{fig:treatment_effect_model_posterior}(a) shows the posterior distribution of the average treatment effect of dabigatran over warfarin.
	The posterior indicates an evidence for the lower incidence rate of gastrointestinal bleeding under dabigatran treatment, which is consistent with findings of \cite{graham2014dabigatran_vs_warfarin}.
	Remarkably, our sparse regression model identifies substantial interaction between the treatment and age group 65--69, with effect size potentially large enough to offset the average treatment effect.
	No other age groups exhibit significant interaction with the treatment.
	The 65--69 age group being the youngest in our Medicare cohort, our finding suggests a possibility that the relative risk of gastrointestinal bleeding only increases in the older patients.
	In fact, \cite{graham2014dabigatran_vs_warfarin} reports the risks from dabigatran and warfarin to be comparable for women  under 75 and men under 85 years old.
	A potential concern with their results is the lack of explanation on their choices of age thresholds.
	On the other hand, our subgroup detection approach based on sparse regression requires no arbitrary selection of subgroups and thus provides a more data-driven alternative to study treatment effect heterogeneity.

	\begin{figure}
		\centering
		\subfigure[
			Average effect of treatment by dabigatran over warfarin on gastrointestinal bleeding. 
		]{
			\includegraphics[height=.21\textheight]{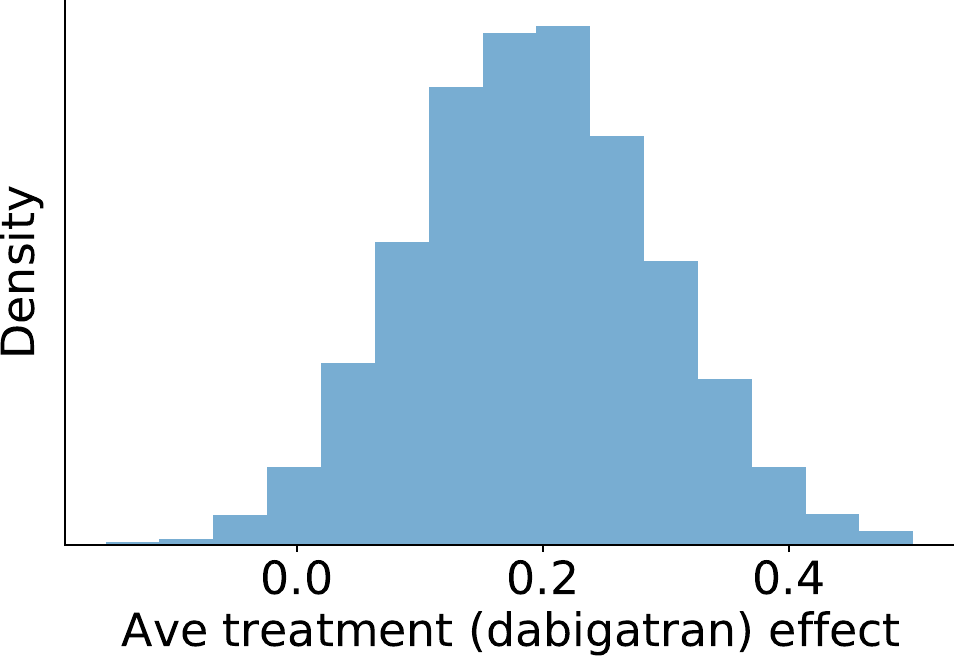}
		}
		~
		\subfigure[
			Coefficient of the interaction between treatment and age group 65--69.
		]{
			\includegraphics[height=.21\textheight]{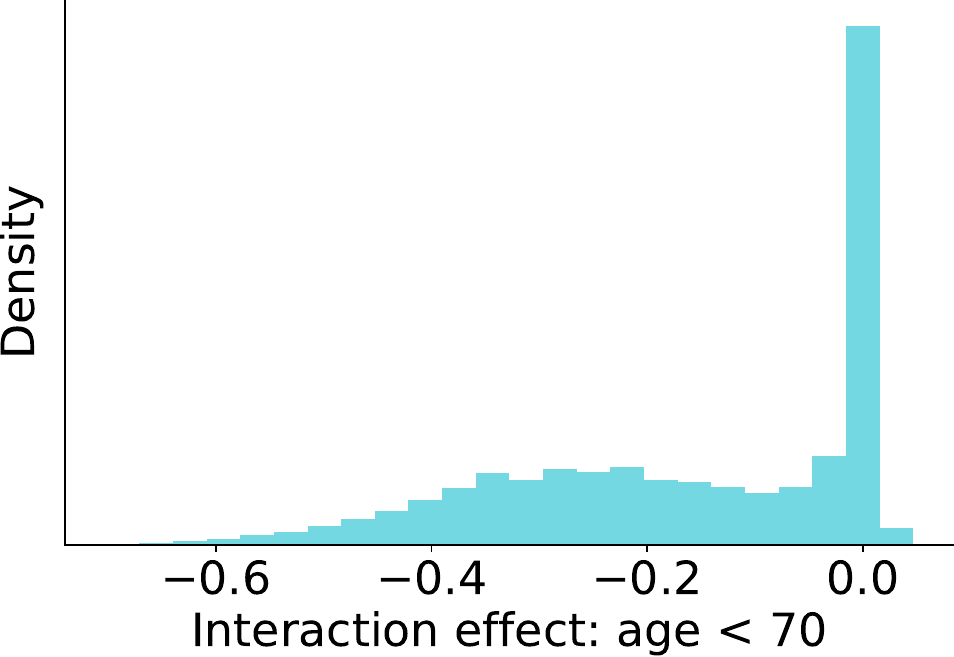}
		}
		\caption{Posterior distributions from the treatment effect model.}
		\label{fig:treatment_effect_model_posterior}
	\end{figure}

\section{Discussion}

In this article, we have developed theory and computational techniques to scale Bayesian sparse regression to a typical size of data in modern applications.
To our knowledge, our computational approach constitutes the first principled use of CG for the purpose of full Bayesian inference via MCMC.
The heuristic described in Section~\ref{sec:prior_preconditioning} suggests that prior-preconditioning may work well in other high-dimensional applications that call for structured, strongly informative prior.
For example, the application of CG to a Gaussian process model as explored by \cite{zhang2019practical-scalable-computing}
may benefit from prior-preconditioning.

As early as 1997, \citeauthor{gibbs1997effcient} emphasized the importance of avoiding expensive linear algebra operations, such as multiplying two matrices or factorizing a matrix, for Bayesian inference to be scalable.
Prior to our work, this desiderata had yet to be met for full Bayesian inference of sparse regression models.
Moreover, large-yet-sparse design matrices are increasingly common in modern applications; it is thus critical to design computational methods to exploit such sparse structure in the data \citep{friedman2010regularization_path}.
Our CG-accelerated Gibbs sampler is an important example to fill these notable gaps in the literature.

\if\blind0
\section{Acknowledgment}
We thank Yuxi Tian and Martijn Schuemie for their help in wrangling the data set used in Section~\ref{sec:application}.
We also thank Jianfeng Lu and Ilse Ipsen for useful discussions on linear algebra topics.
This work is partially supported by National Science Foundation grant DMS 1264153, National Institutes of Health grants U19 AI135995 and R01 AI153044, and Food and Drug Administration grant HHS 75F40120D00039.
\fi

\FloatBarrier
\begin{appendices}

\renewcommand{\thetable}{\thesection.\arabic{table}}
\renewcommand{\thefigure}{\thesection.\arabic{figure}}
\renewcommand{\theequation}{\thesection.\arabic{equation}}

\section{Proofs}
\label{sec:proofs}
	Before we proceed to proving Theorem~\ref{thm:clustering_of_eigenvalues}, we first derive Theorem~\ref{thm:cg_error_bound_under_prior_preconditioner} as its consequence.
	\begin{proof}[Theorem~\ref{thm:cg_error_bound_under_prior_preconditioner}]
		By Theorem~\ref{thm:cg_convergence_delay_by_large_eigenvalues}, the $(m + m')$-th CG iterate $\bbeta_{m + m'}$ satisfies the bound
		\begin{equation}
		\label{eq_for_proof:restatement_of_cg_convergence_delay_by_large_eigenvalues}
		\frac{
			\| \bbeta_{m + m'} - \bbeta \|_{\bPhi}
		}{
			\| \bbeta_0 - \bbeta \|_{\bPhi}
		}
		\leq 2 \left( \frac{
			\sqrt{\eigen_{m + 1} / \eigen_{p}} - 1
		}{
			\sqrt{\eigen_{m + 1} / \eigen_{p}}+ 1
		}
		\right)^{m'},
		\end{equation}
		where $\eigen_j$ denotes the $j$-th largest eigenvalue of $\bPhi$. By Theorem~\ref{thm:clustering_of_eigenvalues}, we know that
		\begin{equation*}
		1 \leq \eigen_p
			\leq \eigen_{m + 1}
			\leq 1 + \displaystyle{\min_{k + \ell = m}} \tau^2 \lshrink_{(k+1)}^2
			\eigen_{\ell + 1} \! \left( (\X^\transpose \bOmega \X )_{(-k)} \right)
			= \widetilde{\kappa}_{m}
		\end{equation*}
		and hence that $\eigen_{m + 1} / \eigen_{p} \leq \widetilde{\kappa}_{m}$. Since the function $\kappa \to (\sqrt{\kappa} - 1) / (\sqrt{\kappa} + 1)$ is increasing in $\kappa$, we can upper bound the right-hand side of \eqref{eq_for_proof:restatement_of_cg_convergence_delay_by_large_eigenvalues} in terms of $\widetilde{\kappa}_{m}$, yielding the desired inequality \eqref{eq:cg_error_bound_under_prior_preconditioner}.
	\end{proof}

	\begin{proof}[Theorem~\ref{thm:clustering_of_eigenvalues}]
	We prove the more general inequality \eqref{eq:prior_preconditioned_eigenvalue_bounds}. The lower bound $1 \leq \eigen_{k + \ell}(\btilPhi)$ is an immediate consequence of Proposition~\ref{prop:eigenvalue_of_perturbation_by_identity}. For the upper bound, first note that
	$\eigen_{k + \ell} (\btilPhi)
		\leq \eigen_{\ell} \! \left( \btilPhi_{(-k)} \right)$
	by the Poincar\'{e} separation theorem (Theorem~\ref{thm:poincare_separation}). From the expression \eqref{eq:prior_preconditioned_matrix} for $\btilPhi$, we have
	\begin{equation*}
	\begin{aligned}
	\eigen_{\ell} \! \left( \btilPhi_{(-k)} \right)
		&= \eigen_{\ell} \Big(
			\I_k + \tau^2 \bLshrink_{(-k)} (\X^\transpose \bOmega \X )_{(-k)} \bLshrink_{(-k)}
		\Big) \\
		&= 1 + \tau^2 \, \eigen_{\ell} \Big(
			 \bLshrink_{(-k)} (\X^\transpose \bOmega \X)_{(-k)} \bLshrink_{(-k)}
		\Big),
	\end{aligned}
	\end{equation*}
	where the second equality follows from Proposition~\ref{prop:eigenvalue_of_perturbation_by_identity}.	Applying Lemma~\ref{lem:multiplicative_perturbation_of_eigenvalues} with $\bA = (\X^\transpose \bOmega \X)_{(-k)}$ and $\bB = \lshrink_{(k+1)}^{-2} \bLshrink_{(-k)}^2$, we obtain
	\begin{equation*}
	\eigen_{\ell} \! \left( \btilPhi_{(-k)} \right)
		\leq 1 + \tau^2 \lshrink_{(k+1)}^{2}  \, \eigen_{\ell} \big(
				(\X^\transpose \bOmega \X)_{(-k)}
			\big).
	\end{equation*}
	Thus we have shown
	\begin{equation*}
	\eigen_{k + \ell} (\btilPhi)
		\leq 1 + \tau^2 \lshrink_{(k+1)}^{2} \, \eigen_{\ell} \big(
				(\X^\transpose \bOmega \X)_{(-k)}
			\big)
		\leq 1 + \tau^2 \lshrink_{(k+1)}^{2} \, \eigen_\ell \! \left(
				\X^\transpose \bOmega \X
			\right),
	\end{equation*}
	where the inequality $ \eigen_{\ell} \big(
	(\X^\transpose \bOmega \X)_{(-k)} \big) \leq \eigen_\ell \! \left(
	\X^\transpose \bOmega \X
	\right)$ follows again from the Poincar\'{e} separation theorem.
	\end{proof}

	\begin{proposition}
	\label{prop:eigenvalue_of_perturbation_by_identity}
	Given a $p \times p$ symmetric matrix $\bA$, the eigenvalues of the matrix $\I_p + \bA$ are given by $1 + \eigen_k(\bA)$ for $k = 1, \ldots, p$.
	\end{proposition}

	\begin{proof}
	The result follows immediately from the spectral theorem for normal matrices \citep{horn2012matrix-analysis}.
	\end{proof}

	\begin{theorem}[Poincar\'e separation theorem]
	\label{thm:poincare_separation}
	For a given $p \times p$ symmetric matrix $\bA$, let $\bA_{(-k)}$ denote the sub-matrix with the first $k$ rows and columns removed from $\bA$. Then the eigenvalues of $\bA$ and $\bA_{(-k)}$ satisfies
		\begin{equation*}
		\eigen_{k + \ell}(\bA)
			\leq \eigen_{\ell}(\bA_{(-k)})
			\leq \eigen_{\ell}(\bA)
			\quad \text{ for } \, \ell = 1, \ldots, p - k.
		\end{equation*}
	Since permuting the rows and columns of $\bA$  does not change its eigenvalues, the above inequality in fact holds for any sub-matrix of $\bA$ obtained by removing $k$ rows and columns of $\bA$ corresponding to a common set of indices $j_1, \ldots, j_k$.
	\end{theorem}
	\begin{proof}
	See Chapter 4.3 of \cite{horn2012matrix-analysis}.
	\end{proof}

	\begin{lemma}
	\label{lem:multiplicative_perturbation_of_eigenvalues}
	Let $\bA$ and $\bB$ be $p \times p$ symmetric positive definite matrices and suppose that the largest eigenvalue of $\bB$ satisfies $\eigen_1(\bB) \leq 1$. Then we have
		\begin{equation*}
		\eigen_k(\bB^{1/2} \bA \bB^{1/2}) \leq \eigen_k(\bA)
		\ \text{ for } \ k = 1, \ldots, p
		\end{equation*}
	where $\eigen_k(\cdot)$ denotes the $k$-th largest eigenvalue of a given matrix.
	\end{lemma}

	\begin{proof}
		The result follows immediately from Ostrowski's theorem (Theorem 4.5.9 in \cite{horn2012matrix-analysis}).
	\end{proof}

\section{Theories of CG convergence behavior}
\label{sec:theories_behind_cg}
In this section, we provide mathematical foundations behind the claims made in Rule of Thumb~\ref{quasi-thm:cg_convergence}. In essence, Rule of Thumb~\ref{quasi-thm:cg_convergence} is our attempt at describing a phenomenon known as the \textit{super-linear} convergence of CG in a quantitative yet accessible manner.
While this is a well-known phenomenon among the researchers in scientific computing, it is rarely explained in canonical textbooks and reference books in numerical linear algebra.\footnote{For example, discussions beyond Theorem~\ref{thm:cg_bound_by_sqrt_cond_num} and \ref{thm:cg_convergence_for_low_rank_perturbation} cannot be found in, to name a few, \cite{trefethen1997numerical_linalg}, \cite{demmel1997numerical_linalg},  \cite{saad2003iterative-methods}, and \cite{golub2012matrix}.} Here we bring together some of the most practically useful results found in the literature and present them in a concise and self-contained manner. Our presentation in Section~\ref{supp:cg_approximation_as_polynomial_approximation} and \ref{supp:cg_bound_via_polynomial_approximation} is roughly based on Section 5.3 of \cite{vorst2003iterative} with details modified, added, and condensed as needed. More comprehensive treatment of the known results related to CG is found in \cite{meurant2006lancoz-and-cg}. \cite{kuijlaars2006krylov_method_and_potential_theory} sheds additional light on CG convergence behaviors by studying them from the potential theory perspective.

Section~\ref{supp:cg_approximation_as_polynomial_approximation} explains the critical first step in understanding the convergence of CG applied to a positive definite system $\bPhi \bbeta = \bb$ --- relating the CG approximation error to polynomial interpolation error over the set $\{\eigen_1, \ldots, \eigen_p \}$ comprising the eigenvalues  of $\bPhi$.
From this perspective, one can understand Theorem~\ref{thm:cg_bound_by_sqrt_cond_num} as a generic and crude bound, ignoring the distributions of  $\eigen_j$'s in-between the largest and smallest eigenvalues (Theorem~\ref{thm:polynomial_bound_over_interval}).
Theorem~\ref{thm:cg_convergence_for_low_rank_perturbation} similarly follows from the polynomial approximation perspective.

The effects of the largest eigenvalues on CG convergence, as stated in Rule of Thumb~\ref{quasi-thm:cg_convergence}, is made mathematically precise in Theorem~\ref{thm:cg_convergence_delay_by_large_eigenvalues}. Analyzing how the smallest eigenvalues delay CG convergence is more involved and requires a discussion of how the eigenvalues are approximated in the Krylov subspace. The amount of initial delay in CG convergence is closely related to how quickly these eigenvalue approximations converge. A precise statement is given in Theorem~\ref{thm:cg_convergence_extreme_eigenvalue_removal}.

The proofs of all the results stated in this section are provided in Supplement Section~\ref{supp:cg_theory_proofs}.

\subsection{CG approximation error as polynomial interpolation error}
\label{supp:cg_approximation_as_polynomial_approximation}
The space of polynomials $\polyspace_k$ as defined below plays a prominent role in the behavior of a worst-case CG approximation error:
\begin{equation*}
	\polyspace_k = \{Q_k(\eigen) : \text{$Q_k$ is a polynomial of degree $k$ with $Q_k(0) = 1$} \}.
\end{equation*}
Proposition~\ref{prop:cg_error_as_polynomial_sum} below establishes the connection between CG and the space $\polyspace_k$.

\begin{proposition}
	\label{prop:cg_error_as_polynomial_sum}
	The difference between the $k$-th CG iterate $\bbeta_k$ and the exact solution $\bbeta$ can be expressed as
	\begin{equation}
		\label{eq:cg_error_as_polynomial_sum}
		\bbeta_k - \bbeta = R_k(\bPhi) \left( \bbeta_0 - \bbeta \right)
		\ \text{ for } \, R_k = \underset{Q_k \in \polyspace_k}{\rm argmin} \, \| Q_k(\bPhi) (\bbeta_0 - \bbeta)  \|_{\bPhi}.
	\end{equation}
	In particular, the following inequality holds for any $Q_k \in \polyspace_k$:
	\begin{equation}
		\label{eq:cg_error_bound_by_polynomial_sum}
		\| \bbeta_k - \bbeta \|_{\bPhi}
		\leq \| Q_k(\bPhi) (\bbeta_0 - \bbeta)  \|_{\bPhi}.
	\end{equation}
\end{proposition}

Theorem~\ref{thm:cg_bound_over_discrete_polynomial_values} below uses Proposition~\ref{prop:cg_error_as_polynomial_sum} to establish the relation between the CG approximation error and a polynomial interpolation error. We can interpret the result as saying the following: a worst-case CG approximation error can be quantified via how well the set of points $\{(\eigen_j, 0)\}_{j = 1, \ldots, p}$ can be interpolated by the graph $\eigen \to (\eigen, Q_k(\eigen))$ of a $k$-th degree polynomial $Q_k$ with the constraint $Q_k(0) = 1$.
\begin{theorem}
	\label{thm:cg_bound_over_discrete_polynomial_values}
	\begin{equation} \abovedisplayshortskip=-.5\baselineskip
		\label{eq:cg_bound_over_discrete_polynomial_values}
		\frac{
			\| \bbeta_k - \bbeta \|_{\bPhi}
		}{
			\| \bbeta_0 - \bbeta \|_{\bPhi}
		}
		\leq \min_{Q_k \, \in \, \polyspace_k} \max_{j = 1, \ldots, p} | Q_k(\eigen_j) |,
	\end{equation}
	where $\eigen_j$ denotes the $j$-th largest eigenvalue of $\bPhi$. The bound is sharp in a sense that, for each $k$, there exists an initial vector $\bbeta_0$ for which the equality holds.
\end{theorem}

\subsection{Bounding CG error via its polynomial characterization}
\label{supp:cg_bound_via_polynomial_approximation}
We now derive bounds on the CG approximation error through its characterization as a polynomial interpolation error (Theorem~\ref{thm:cg_bound_over_discrete_polynomial_values}).
Minimizing the interpolation error over the entire interval between the largest and smallest eigenvalues yields the following bound.
\begin{theorem}
	\label{thm:polynomial_bound_over_interval}
	\begin{equation} \abovedisplayskip=-.5\baselineskip
		\label{eq:polynomial_bound_over_interval}
		\min_{Q_k \, \in \, \polyspace_k} \, \max_{\eigen \in [ \eigen_{\textrm{min}},  \eigen_{\textrm{max}} ]} |Q_k(\eigen)|
		\leq 2 \left( \frac{
			\sqrt{\eigen_{\textrm{max}} / \eigen_{\textrm{min}}} - 1
		}{
			\sqrt{\eigen_{\textrm{max}} / \eigen_{\textrm{min}}} + 1
		}
		\right)^k.
	\end{equation}
\end{theorem}
\noindent Theorem~\ref{thm:cg_bound_over_discrete_polynomial_values} and \ref{thm:polynomial_bound_over_interval} together yield the well-known CG approximation error bound of Theorem~\ref{thm:cg_bound_by_sqrt_cond_num}. As the bound of Theorem~\ref{thm:cg_bound_over_discrete_polynomial_values} depends only on the maximum over a discrete set of the eigenvalues $\{ \eigen_p, \ldots, \eigen_1 \}$, rather than the entire interval $[\eigen_p, \eigen_1]$, the actual CG convergence rate can be faster.

Theorem~\ref{thm:cg_convergence_delay_by_large_eigenvalues} below is a basis of the following claim made in Rule of Thumb~\ref{quasi-thm:cg_convergence}: ``the $\nlargest$ largest eigenvalues are effectively removed within $\nlargest$ iterations.''
\begin{theorem}
	\label{thm:cg_convergence_delay_by_large_eigenvalues}
	The following bound holds for all $\nlargest, k \geq 0$ with $\nlargest < p$:
	\begin{equation}
		\label{eq:cg_convergence_removal_of_largest_eigenvalues}
		\frac{
			\| \bbeta_{\nlargest + k} - \bbeta \|_{\bPhi}
		}{
			\| \bbeta_0 - \bbeta \|_{\bPhi}
		}
		= \min_{Q_{\nlargest + k} \, \in \, \polyspace_{\nlargest + k}} \,  \max_{j = 1, \ldots, p} |Q_{\nlargest + k}(\eigen_j)|
		\leq 2 \left( \frac{
			\sqrt{\eigen_{\nlargest + 1} / \eigen_{p}} - 1
		}{
			\sqrt{\eigen_{\nlargest + 1} / \eigen_{p}}+ 1
		}
		\right)^k,
	\end{equation}
	where the first equality is given by Theorem~\ref{thm:cg_bound_over_discrete_polynomial_values}.
\end{theorem}

The smallest eigenvalues affect the CG convergence rate differently from the largest ones due to the constraint $Q_k(0) = 1$ in $\polyspace_k$. Intuitively, this constraint makes the smallest eigenvalues more significant contributers to the polynomial interpolation error because it competes with the objective of minimizing $|Q_k(\eigen)|$ for small $\eigen$. This is why we state in Rule of Thumb~\ref{quasi-thm:cg_convergence} that ``the same number of smallest eigenvalues tends to delay the convergence longer.'' Nonetheless, the effects of the smallest eigenvalues on the CG approximation error becomes attenuated as the CG iterations proceed. To quantify this phenomenon, we need to introduce the notion of \textit{Ritz values} and describe their roles in the CG convergence behavior.

The Ritz values at the $\nRitzIter$-th CG iteration refer to the roots $\big\{ \ritz_1^{(\nRitzIter)}, \ldots, \ritz_\nRitzIter^{(\nRitzIter)} \big\}$ of the optimal CG polynomial $R_\nRitzIter$ as defined in \eqref{eq:cg_error_as_polynomial_sum}. Unless the eigenvalues $\eigen_p, \ldots, \eigen_1$ are distributed in a highly unusual manner, the largest and smallest Ritz values have a property that they converges quickly to to the largest and smallest eigenvalues of $\bPhi$ \citep{trefethen1997numerical_linalg, driscoll1998potential-theory, kuijlaars2006krylov_method_and_potential_theory}. More precisely, we have $\ritz^{(\nRitzIter)}_i \to \eigen_i$ for $i = 1, \ldots, \nlargest$ and $\ritz^{(\nRitzIter)}_{\nRitzIter-i} \to \eigen_{p - i}$ for $i = 0, \ldots, \nsmallest$ as $\nRitzIter \to p$. While the convergence rates of the Ritz values can be shown to be exponential, in practice quite a large number of CG iterations may be required to obtain good approximations unless $\max\{\nlargest, \nsmallest \} \ll p$ \citep{saad2011methods-for-eigenvalue}.

Theorem~\ref{thm:cg_convergence_extreme_eigenvalue_removal} below quantifies how the convergence of the Ritz values are related to the subsequent acceleration of the CG convergence rates.
\begin{theorem}
	\label{thm:cg_convergence_extreme_eigenvalue_removal}
	The CG approximation error of the $(\nRitzIter + \nSubseqIter)$-th iterate relative to the $\nRitzIter$-th iterate satisfies the following bound:
	\begin{equation}
		\label{eq:cg_convergence_extreme_eigenvalue_removal}
		\frac{
			\left\| \bbeta_{\nRitzIter + \nSubseqIter} - \bbeta \right\|_{\bPhi}
		}{
			\left\| \bbeta_{\nRitzIter} - \bbeta \right\|_{\bPhi}
		}
		\leq C_{\nRitzIter, \nlargest, \nsmallest} \, 2 \left( \frac{
			\sqrt{\eigen_{\nlargest + 1} / \eigen_{p - \nsmallest}} - 1
		}{
			\sqrt{\eigen_{\nlargest + 1} / \eigen_{p - \nsmallest}} + 1
		} \right)^\nSubseqIter,
	\end{equation}
	where $C_{\nRitzIter, \nlargest, \nsmallest} = C_{\nRitzIter, \nlargest, \nsmallest}\big( \ritz_1^{(\nRitzIter)}, \ldots, \ritz_\nRitzIter^{(\nRitzIter)} \big) \to 1$ as $\nRitzIter \to p$ for any fixed $\nlargest, \nsmallest \geq 0$ with $\nlargest + \nsmallest < p$. More precisely, $C_{\nRitzIter, \nlargest, \nsmallest}$ tends to 1 as the $\nlargest$ largest and $\nsmallest$ smallest Ritz values converge to the largest and smallest eigenvalues of $\bPhi$.
\end{theorem}

\end{appendices}

\FloatBarrier

\setlength{\bibsep}{5.5pt plus 0.3ex}
\bibliographystyle{agsm}
\bibliography{cg_accelerated_gibbs}{}

\newpage

\renewcommand{\thesection}{S\arabic{section}}
\renewcommand{\thetable}{S\arabic{table}}
\renewcommand{\thefigure}{S\arabic{figure}}
\renewcommand{\theequation}{S\arabic{equation}}

\setcounter{section}{0}
\setcounter{figure}{0}
\setcounter{table}{0}
\setcounter{equation}{0}
\setcounter{page}{1}

\nolinenumbers
{
	\bigskip
	\bigskip
	\bigskip
  \begin{center}
  \spacingset{1}
    {\LARGE\bf
    	Supplement to
    	``Prior-preconditioned conjugate gradient method for accelerated Gibbs sampling in \mbox{`large $n$ \& large $p$'} \\ Bayesian sparse regression'' \par
    }
  \end{center}
  \medskip
}
\linenumbers

\section{Sparse logistic regression Gibbs sampler}
\label{sec:gibbs_sampler_details}

Here we provide a detailed description of sparse logistic regression Gibbs samplers, the computational bottleneck of which is the focus of this article.

The conditional distributions of $\tau$ and $\blshrink$ obviously depends on particular shrinkage priors used and is given in terms of the local and global scale prior $\localPrior(\cdot)$ and $\globalPrior(\cdot)$ as
\begin{equation}
\label{eq:scale_parameter_posterior}
\pi(\tau, \blshrink \given \bbeta, \bomega, \y, \X)
	= \pi(\tau, \blshrink \given \bbeta)
	\propto \globalPrior(\tau) \prod_j \frac{1}{\tau \lshrink_j} \exp\!\left( - \frac{\beta_j^2}{2 \tau^2 \lshrink_j^2}\right) \localPrior(\lshrink_j).
\end{equation}
For the Bayesian bridge prior used in our simulations, a Gamma distribution is a conjugate prior for $\tau^{-\alpha}$, and $\lshrink_j$'s can be updated via the double-rejection sampler of \cite{devroye2006random_variate_generation} \citep{polson2014bayes_bridge}.
For the popular horseshoe prior \citep{carvalho2010horseshoe}, which corresponds to a half-Cauchy prior for $\localPrior(\cdot)$, an efficient rejection sampler is available for the full conditional of $\blshrink$ \citep{nishimura2019regularized_shrinkage}.
Coincidentally, a half-Cauchy is also a common prior choice for $\tau$ \citep{gelman2008default_prior, piironen2017regularized-horseshoe}, so the same rejection sampler can be used to update $\tau$ from its full conditional.

The Polya-Gamma data-augmentation of \cite{polson2013polya_gamma} is a widely-used approach for posterior computation under the logistic model.
By introducing an auxiliary parameter $\bomega = (\omega_1, \ldots, \omega_n)$, with each element having a Polya-Gamma distribution, the Gibbs sampler induces a transition kernel: $(\bomega^*, \bbeta^*, \blshrink^*, \tau^*) \to (\bomega, \bbeta, \blshrink, \tau)$ through the following cycle of conditional updates:
\begin{enumerate} 
\item Draw $\tau \given \bbeta^*, \blshrink^*$ from the density proportional to \eqref{eq:scale_parameter_posterior}.
When using Bayesian bridge priors, draw from the collapsed distribution $\tau \given \bbeta^*$ \citep{polson2014bayes_bridge}.

\item Draw $\blshrink \given \bbeta^*, \tau$ from the density proportional to \eqref{eq:scale_parameter_posterior}.
\item Draw
$\omega_i \given \bbeta^*, \X
	\sim \textrm{PolyaGamma}(\textrm{shape}=1, \textrm{tilting}=\x_i^\transpose \bbeta^*)$
for $i = 1, \ldots, n$.
\item Draw $\bbeta \given \bomega, \tau, \blshrink, \y, \X$ from the multivariate-Gaussian as given in \eqref{eq:beta_conditional}.
\end{enumerate}
We refer readers to \cite{polson2013polya_gamma} for more details on this data augmentation scheme.

\section{Pseudo-code for conjugate gradient method}
\label{sec:cg_pseudo_code}
Algorithm~\ref{alg:conjugate_grad} below describes the steps of CG for solving a linear system $\bPhi \bbeta = \bb$ from an initial guess $\bbeta_0$.
It essentially coincides with the form as described in Section~11.3.8 of \cite{golub2012matrix} but with suitably adapted notations.

\newcommand{\resid}{\bm{r}}
\newcommand{\searchDirec}{\bm{s}}
\newcommand{\bPhiSearchDirec}{\bm{v}}
\newcommand{\cgStepsize}{\varsigma}
\newcommand{\residNormPrev}{\rho_{\textrm{prev}}}
\newcommand{\residNormCurr}{\rho_{\textrm{curr}}}
\begin{algorithm}[H]
\caption{Conjugate gradient method\rule[-4pt]{0pt}{16pt}}
\label{alg:conjugate_grad}
\begin{algorithmic}
	\Function{ConjugateGradient\rule{0pt}{14pt}}{$\bPhi, \bbeta_0, \bb$}
		\State $\bbeta \gets \bbeta_0$
		\State $\resid \gets \bPhi \bbeta - \bb$
		\State $\residNormCurr \gets \| \resid \|^2$
		\State $\searchDirec \gets \resid$
		\While{convergence criteria is unmet} \Comment{See Section~\ref{supp:termination_criteria_for_cg_sampler} for the criteria}
			\State $\bPhiSearchDirec \gets \bPhi \searchDirec$
			\State $\cgStepsize \gets \residNormCurr / \langle \searchDirec, \bPhiSearchDirec \rangle$
			\State $\bbeta \gets \bbeta + \cgStepsize \searchDirec$
			\State $\resid \gets \resid - \cgStepsize \bPhiSearchDirec$
			\State $\residNormPrev \gets \residNormCurr$
			\State $\residNormCurr \gets \| \resid \|^2$
			\State $\searchDirec \gets \resid + (\residNormCurr / \residNormPrev) \searchDirec$
		\EndWhile
     \EndFunction
\end{algorithmic}
\end{algorithm}

\section{General principle behind prior-preconditioning}
\label{sec:principle_behind_prior_preconditioning}
	In the context of the CG sampler, the preconditioned matrix $\M^{-1/2} \bPhi \M^{-1/2}$ represents the precision matrix of the transformed parameter $\bm{\tilde{\beta}} = \M^{1/2} \bbeta$. In fact, preconditioning the linear system \eqref{eq:linear_system_for_cg_sampler} with a preconditioner $\bm{M}$ is equivalent to applying a parameter transformation $\bbeta \to \M^{1/2} \bbeta$ before employing the CG sampler. That is, we can apply one of the two strategies --- precondition the linear system or apply the parameter transformation --- to achieve exactly the same effect on the speed of the CG sampler.

	When we choose the prior precision as the preconditioner, the transformed parameter $\M^{1/2} \bbeta$ a priori has the identity precision matrix, before its distribution is modified via the likelihood. This perspective, combined with the fact that the eigenvalues of $\btilPhi$ represents the posterior precisions of $\M^{1/2} \bbeta$ along its principal components, suggests the following principle:
	\begin{principle-behind}
		Under a strongly informative prior, the posterior looks like the prior except in a small number of directions along which the data provide significant information. This translates into the eigenvalues of the prior-preconditioned matrix $\btilPhi$ clustering around 1 except for a relatively small number of large eigenvalues.
	\end{principle-behind}
	The eigenvalue structure of the prior-preconditioned matrix as predicted above is indeed observed across all of our numerical examples --- see Figure~\ref{fig:eigval_distributions_for_data_simulated_with_10_correlated_signal}, \ref{fig:eigval_distributions_for_treatment_model}, and \ref{fig:eigval_distributions_for_outcome_model}.

\section{Practical details on deploying CG sampler for sparse regression}
\label{supp:cg_sampler_practical_details}
	Throughout this section, we write $\bv \bw$ and $\bv / \bw$ to denote an element-wise multiplication and division of two vectors $\bv$ and $\bw$.

\subsection{Choice of initial vector for CG iterations}
\label{supp:cg_init_vector}
	Generally speaking, the CG iterations decrease the distance between the iterates $\bbeta_k$'s and the exact solution $\bbeta$ relative to the initial error $\|\bbeta_0 - \bbeta \|_{\bPhi}$. However, a choice of the initial vector is not as significant as that of the preconditioner which determines the eventual exponential convergence rate of CG. In other words, once the initial vector is chosen within a reasonable range, we should not expect a dramatic gain from further fine-tuning.
	When sampling $\bbeta$ from a sparse regression posterior, we indeed find it difficult to improve much over a simple initialization $\bbeta_0 = \bm{0}$, which is a reasonable choice as most coefficients are shrunken to zero. We achieve only small ($\lesssim 10$\%), though consistent, improvements by one of the alternative approaches we experimented with. We describe these approaches below.

	As an alternative to $\bbeta_0 = \bm{0}$, we consider three approaches for constructing the initial vector. At the $m$-th Gibbs update, the CG sampler needs to draw $\bbeta^{(m)}$ from the distribution $\bbeta \given \bomega^{(m - 1)}, \blshrink^{(m - 1)}, \tau^{(m - 1)}, \y, \X$. We have no control over the variability in $\bbeta^{(m)}$, so we focus on getting $\bbeta_0$ as close as possible to the mean of $\bbeta^{(m)}$. The two seemingly obvious choices of $\bbeta_0$ are 1) the previous MCMC sample $\bbeta^{(m - 1)}$ and 2) the MCMC estimate $m^{-1} \sum_{i = 0}^{m - 1} \bbeta^{(i)}$ of the expectation $\mathbb{E}[\bbeta \given \y, \X]$. These options, however, ignore the fact that the distribution of $\bbeta^{(m)}$ depends strongly on $\tau^{(m - 1)} \blshrink^{(m - 1)}$, which generally is very different from $\tau^{(m - i)} \blshrink^{(m - i)}$ for $i \geq 2$.

	We found the following approach, implemented in our CG-accelerated Gibbs sampler of Section~\ref{sec:application}, to yield a better estimate of the mean and hence a better initialization for $\bbeta^{(m)}$.
	We first estimate $\mathbb{E}[\tau^{-1} \blshrink^{-1} \bbeta \given \y, \X]$ by the estimator $\bm{\tilde{\beta}}_0 = m^{-1} \sum_{i = 0}^{m - 1} \bbeta^{(i)} / \tau^{(i - 1)} \blshrink^{(i - 1)}$, where we define $\tau^{(- 1)} \lshrink_j^{(- 1)} = 1$.
	Then we rescale it with the current conditioned values of $\tau$ and $\blshrink$, setting $\bbeta_0 = \tau^{(m-1)} \blshrink^{(m-1)} \bm{\tilde{\beta}}_0$ to obtain the initial vector. We compared this approach to the other two through a simulation study and found our choice to consistently yield smaller $\bPhi$-norm errors and faster convergence.

\subsection{Termination criterion for CG iterations}
\label{supp:termination_criteria_for_cg_sampler}
	An iterative method must be supplied with a termination criterion to decide when the current iterate $\bbeta_k$ is close enough to the exact solution.
	While different convergence metrics can be computed as bi-products of the CG iterations \citep{meurant2006lancoz-and-cg}, most existing linear algebra libraries uses the $\ell^2$ norm of the residual $\bm{r}_k= \bPhi \bbeta_k - \bm{b}$. It is possible to relate the residual norm to $\| \bbeta_k - \bbeta \|_2$ as
	\begin{equation*}
	\| \bbeta_k - \bbeta \|_2
	= \| \bPhi^{-1} \bm{r}_k \|_2
	\leq \| \bPhi^{-1} \|_2 \| \bm{r}_k \|_2.
	\end{equation*}
	For the purpose of sampling a Gaussian vector $\bbeta$, however, it is not at all clear when $\| \bm{r}_k \|_2$ or $\| \bbeta_k - \bbeta \|_2$ can be considered small enough. To address this problem, we develop an alternative metric tailored toward the CG sampler for sparse regression.

	We propose to assess the CG convergence in terms of the $\ell^2$ norm of the prior-preconditioned residual $\bm{\tilde{r}}_k = \btilPhi \bm{\tilde{\beta}}_k - \bm{\tilde{b}} = \tau \blshrink \bm{r}_k$. More specifically, we use the termination criterion
	\begin{equation}
	\label{eq:mean_sq_resid_for_cg_termination}
	p^{-1/2} \| \bm{\tilde{r}}_k  \|_{2}
		= \left\{ p^{-1} \textstyle{\sum}_{j = 1}^p (\bm{\tilde{r}}_k)_j^2 \right\}^{1/2}
		\leq 10^{-6},
	\end{equation}
	in terms of the root-mean-squared residual $p^{-1/2} \| \bm{\tilde{r}}_k  \|_{2}$.
	The criteria is justified by the norm $\| \bm{\tilde{r}}_k \|_2$ being an approximate upper bound to the following quantity:
	\begin{equation}
	\label{eq:cg_sampler_error_metric}
	\big\| \bm{\secmom}^{-1} \left( \bbeta_k - \bbeta \right) \big\|_2
	\ \text{ with } \ \xi_j^2 = \mathbb{E} \big[ \, \beta_j^2 \given \bomega, \blshrink, \tau, \y, \X \big].
	\end{equation}
	The standardization by second moment ensures that, when the computed error is small, all the coordinates of $\bbeta_k$ are close to those of $\bbeta$ either in terms of their means or variances of the target Gaussian distribution.

	To relate the norm of the prior-preconditioned residual $\bm{\tilde{\nlargest}}_k = \btilPhi \bm{\tilde{\beta}}_k - \bm{\tilde{b}}$ to the error metric \eqref{eq:cg_sampler_error_metric}, observe that $\bbeta_k - \bbeta = \M^{-1/2} \btilPhi^{-1} \bm{\tilde{\nlargest}}_k$ with $\M = \tau^{-2} \bLshrink^{-2}$ and hence
	\begin{equation}
	\label{eq:normalized_error_relation_to_precond_resid}
	\begin{aligned}
	\big\| \bm{\secmom}^{-1} \left( \bbeta_k - \bbeta \right) \big\|_2
		&= \big\|
			\bm{\secmom}^{-1} (\tau \blshrink) \big( \btilPhi^{-1} \bm{\tilde{\nlargest}}_k \big)
		\big\|_2
		\leq \left(\max_j \secmom_j^{-1} \tau \lshrink_j \right) \big\| \btilPhi^{-1} \bm{\tilde{\nlargest}}_k \big\|_2.
	\end{aligned}
	\end{equation}
	The inequality in the above equation only represents the worst-case scenario; in more typical settings, one expects the norm of $\bm{\secmom}^{-1} (\tau \blshrink) \bv$ to be related to that of $\bv$ through some average of $\secmom_j^{-1} \tau \lshrink_j$'s. In any case, we proceed to analyze a typical behavior of $\secmom_j^{-1} \tau \lshrink_j$ as the parameters $\bomega, \blshrink, \tau$ are drawn from a sparse regression posterior. As before, we interpret $\tau^2 \lshrink_j^2$ as the prior variance of $\beta_j$ (conditional on $\bomega, \blshrink, \tau$) before observing $\y, \X$. Note that
	\begin{equation*}
	(\secmom_j^{-1} \tau \lshrink_j)^2
		= \frac{
			\tau^2 \lshrink_j^2
		}{
			\mu_j^2 + \sigma_j^2
		},
	\end{equation*}
	where $\mu_j$ and $\sigma_j$ are the conditional mean and variance of $\beta_j \given \bomega, \blshrink, \tau, \y, \X$. So the quantity $\secmom_j^{-1} \tau \lshrink_j$ is not too far from 1 if either $|\mu_j|$ or $\sigma_j$ is in the same order of magnitude as $\tau \lshrink_j$. If $\beta_j$'s posterior is dominated by the prior shrinkage, we expect the posterior (conditional) variance to be not much smaller than the prior one and hence $\sigma_j \approx \tau \lshrink_j$. Otherwise, if $\sigma_j \ll \tau \lshrink_j$ and the likelihood is a dominant contributer to the posterior, then the posterior of $\tau \lshrink_j$ should concentrate around of $|\mu_j|$ to maximize the marginal likelihood of $\beta_j \approx \mu_j$. Either way, we can expect $\secmom_j^{-1} \tau \lshrink_j$ to be in the same order of magnitude as 1.

	From the relation \eqref{eq:normalized_error_relation_to_precond_resid} and our analysis above, we deduce that
	\begin{equation*}
	\begin{aligned}
	\big\| \bm{\secmom}^{-1} \left( \bbeta_k - \bbeta \right) \big\|_2
		&\lesssim \big\| \btilPhi^{-1} \bm{\tilde{\nlargest}}_k \big\|_2
		\leq \| \bm{\tilde{\nlargest}}_k \|_2,
	\end{aligned}
	\end{equation*}
	where the latter inequality follows from the fact that the largest eigenvalue of the prior-preconditioned matrix $\btilPhi^{-1}$ is bounded above by 1 by Theorem~\ref{thm:clustering_of_eigenvalues}.

\subsection{Modified preconditioner to handle coefficients with uninformative priors}
\label{supp:preconditioning_unshrunken_coef}
	When fitting a sparse regression model, standard practice is to include an intercept $\beta_0$ without any shrinkage, often with the improper flat prior $\pi(\beta_0) \propto 1$ \citepsupplement{park2008bayes-lasso}. Additionally, there may be predictors of particular interests, inference for whose regression coefficients is more appropriately carried out with uninformative or weakly-informative priors without shrinkage; see \citesupplement{zucknick2015variable-selection-for-cox} as well as the application in Section~\ref{sec:application} for examples of such predictors. The CG sampler can accommodate such predictors with an appropriate modification.

	For notational convenience, suppose that the regression coefficients are indexed so that the first $(\nunshrunk + 1)$-th coefficients $\beta_0, \beta_1, \ldots, \beta_{\nunshrunk}$ are to be estimated without shrinkage. We further assume that the unshrunk coefficients are given independent Gaussian priors $\beta_j \sim \normal(0, \sigma_j^2)$ for $0 < \sigma_j \leq \infty$ where $\sigma_j = \infty$ denotes an improper prior $\pi(\beta_j) \propto 1$. The precision matrix of $\bbeta \given \bomega, \blshrink, \tau, \y, \X$ then is given by
	\begin{equation*}
	\bPhi = \X^\transpose \bOmega \X +
		\begin{bmatrix}
		\textrm{diag}( \bm{\sigma} )^{-2} & \bm{0} \\
		\bm{0} & \tau^{-2} \bLshrink^{-2}
		\end{bmatrix}
	\end{equation*}
	for $\bm{\sigma} = (\sigma_0, \ldots, \sigma_\nunshrunk)$ where we employ the convention $1 / \sigma_j = 0$ if $\sigma_j = \infty$. The unshrunk coefficients $\beta_0, \ldots, \beta_\nunshrunk$ are distinguished from the shrunk ones by the fact that their prior scales $\sigma_j$ (before conditioning on $\y$ and $\X$) typically have little to do with their posterior scales (after conditioning on $\y$ and $\X$). For this reason, a naively modified preconditioner $\M = \textrm{diag}(\bm{\sigma}^{-2}, \tau^{-2} \blshrink^{-2})$ may not be appropriate, especially for coefficients with $\sigma_j \gg 1$ corresponding to uninformative priors.

	We propose a modified preconditioner of the form $\M = \textrm{diag}(\bm{\gamma}^{-2},\tau^{-2} \blshrink^{-2})$ for appropriately chosen $\bm{\gamma} = (\gamma_0, \gamma_1, \ldots, \gamma_q)$. For the corresponding preconditioned matrix $\btilPhi = \M^{- 1/2} \bPhi \M^{- 1/2}$, let $\btilPhi_{(-\nunshrunk - 1)}$ denote the sub-matrix with the first $q + 1$ rows and columns removed.
 	 As shown in Section~\ref{sec:theory_of_prior_preconditioning}, the sub-matrix $\btilPhi_{(-\nunshrunk - 1)}$ has an eigenvalue distribution particularly well-suited to induce rapid CG convergence.
	By the Poincar\'{e} separation theorem (Theorem~\ref{thm:poincare_separation}), all but $\nunshrunk + 1$ eigenvalues of the original matrix $\btilPhi$ lie within the largest and smallest eigenvalues of the sub-matrix $\btilPhi_{(-\nunshrunk-1)}$. In choosing $\gamma_j$'s, therefore, we are concerned with the behavior of the $\nunshrunk + 1$ additional eigenvalues introduced by the unshrunk coefficients. Additionally, we should err on the side of introducing larger eigenvalues than smaller ones as small eigenvalues impact CG convergence rates more significantly (Rule of Thumb~\ref{quasi-thm:cg_convergence}).

	With the above objectives in mind, we propose a choice
	\begin{equation}
	\label{eq:precond_scale_for_unshrunk_modified}
	\gamma_j = c \, \widehat{\psi}_j
		\ \text{ for } \
		\widehat{\psi}_j^2 \approx \textrm{var}(\beta_j \given \bomega, \blshrink, \tau, \y, \X)
	 	\, \text{ and } \, c \geq 1.
	\end{equation}
	To explain the reasoning behind the above choice, let $\bbeta_q = (\beta_0, \ldots, \beta_q)$ and $\bbeta_{(-q)} = (\beta_{q + 1}, \ldots, \beta_p)$. The smallest eigenvalues of $\btilPhi$ correspond to the largest variances (conditional on $\bomega, \blshrink, \tau, \y, \X$) of the Gaussian vector $\M^{1/2} \bbeta = (\bm{\gamma}^{-1} \bbeta_q, \tau^{-1} \blshrink^{-1} \bbeta_{(-q)})$ along its principal components. The variances of $\tau^{-1} \blshrink^{-1} \bbeta_{(-q)}$ conditional on $\bm{\gamma}^{-1} \bbeta_q$ are bounded above by 1 along any directions because the eigenvalues of the conditional precision matrix $\btilPhi_{(-\nunshrunk)}$ are bounded below by 1. Therefore, we do not expect $(\bm{\gamma}^{-1} \bbeta_q, \tau^{-1} \blshrink^{-1} \bbeta_{(-q)})$ to have variances much larger than 1 unless the marginal variances of $\bm{\gamma}^{-1} \bbeta_q$ are large. The proposed choice of $\gamma_j$'s ensure that the marginal variances of $\gamma_j^{-1} \beta_j$'s are less than $c^{-1}$
	and thus prevent an introduction of small eigenvalues to $\btilPhi$. The multiplicative factor $c \geq 1$ provides a further safeguard as we are more concerned about small eigenvalues than large ones.

	As the parameters $\bomega, \tau, \lshrink$ are constantly updated during Gibbs sampling, technically we cannot estimate $\textrm{var}(\beta_j \given \bomega, \blshrink, \tau, \y, \X)$ from earlier MCMC samples. In practice, we instead use
	\begin{equation}
	\label{eq:precond_scale_for_unshrunk}
	\gamma_j = c \, \widehat{\eta}_j
		\ \text{ for } \
		\widehat{\eta}_j^2 \approx \textrm{var}(\beta_j \given \y, \X).
	\end{equation}
	Using \eqref{eq:precond_scale_for_unshrunk} in place of \eqref{eq:precond_scale_for_unshrunk_modified} is justified in two ways. First, by the variance decomposition formula we have
	\begin{equation*}
	\mathbb{E}_{\bomega, \tau, \blshrink \given \y, \X} \left[
		\textrm{var} \left(
			\beta_j \given \bomega, \blshrink, \tau, \y, \X
		\right)
	\right]
		\leq \textrm{var}(\beta_j \given \y, \X).
	\end{equation*}
	In other words, on average $\widehat{\eta}_j$ is an overestimate of $\widehat{\psi}_j$ which, as we have discussed, is more preferable to an underestimate. Secondly, the unshrunk coefficients $\beta_0, \ldots, \beta_q$ have only limited dependency on the shrinkage parameters $\tau$ and $\blshrink$ through $\beta_{q + 1}, \ldots, \beta_{p}$. Also, in our experience we have never noticed any obvious correlations between the posterior samples of $\bomega$ and $\bbeta$. For these reasons, we suspect that $\textrm{var}(\beta_j \given \y, \X)$ is generally not too far from $\textrm{var}(\beta_j \given \bomega, \blshrink, \tau, \y, \X)$.


	Once chosen within reasonable ranges, the precise values of $\gamma_j$'s have limited effect on the CG convergence rate.
	This is because all but $(q + 1)$ eigenvalues are well-behaved regardless of the choice of $\gamma_0, \ldots, \gamma_q$ and CG has an ability to eventually ``remove'' the extreme eigenvalues (Rule of Thumb~\ref{quasi-thm:cg_convergence}). In our simulations (not presented in the manuscript), we found the delay in the CG convergence to be no more than $20 \sim 30$\% even when the values of $\gamma_j$'s were off by two orders of magnitude from empirically-determined optimal values. The convergence rate achieved by the proposed choice of $\bm{\gamma}$ was essentially indistinguishable from that achieved by an optimal choice.

\section{Additional simulation results on CG sampler convergence behaviors for Section~\ref{sec:cg_sampler_demonstration}}

\subsection{Effects of error metrics and right-hand vectors}
\label{supp:alt_cg_convergence_plot_for_simulated_data}
	The CG convergence behavior as illustrated in Figure~\ref{fig:cg_convergence_plot_for_simulated_data} remains qualitatively consistent across different random draws of the right-hand vector $\bb$ and across various metrics of the approximation error.
	Figure~\ref{fig:cg_convergence_plot_with_random_seed_for_simulated_data} shows the average of the coordinate-wise relative error as a function of the CG iterations as in Figure~\ref{fig:cg_convergence_plot_for_simulated_data}, but with an individual line for each of the random draws of $\bb$. The convergence behaviors under the prior and Jacobi preconditioners are plotted in the two separate sub-figures to avoid  cluttering the plot with too many lines. Figure~\ref{fig:cg_convergence_plot_with_different_error_metrics_for_simulated_data} shows the CG convergence behaviors under the two additional error metrics: the $\ell^2$-norm and $\bPhi$-norm distance between $\bbeta_k$ and $\bbeta_{\textrm{direct}}$.

	We also investigate how much the coordinate-wise error $\epsilon_j(k) = | (\bbeta_k - \bbeta_{\textrm{direct}})_j / (\bbeta_{\textrm{direct}})_j |$ varies across the coefficient index $j$.
	To summarize this high-dimensional information --- with $k = 1, \ldots, 500$ and $j = 1, \ldots, 10{,}000$ --- we focus on the error trajectories $\{ \epsilon_j(k) \}_{k = 1, 2, \ldots}$ along the coordinates with the largest and smallest errors.
	More precisely, we rank the coordinates by their running maximum error values $\max_k \{ \epsilon_j(k) \}$ and select the top and bottom fifty coordinates accordingly.
	We then plot the error trajectories along these coordinates in Figure~\ref{fig:cg_rel_coord_wise_error_traj_for_simulated_data}.
	The plot shows that the CG error varies considerably across the coordinates but that they all decay roughly at the same rate as a function of $k$.
	The plot is based on a single draw of the right-hand vector $\bb$, but the results are qualitatively similar across all the random draws.
	Additionally, we found no obvious pattern in the distribution of relative errors across the coordinates.
	For instance, one might wonder if larger values of $\bbeta_{\textrm{direct}, j}$ corresponds to larger (or smaller) relative errors, but Figure~\ref{fig:cg_coef_val_vs_final_rel_error_for_simulated_data} shows this not to be the case.

	\FloatBarrier
	\begin{figure}
		\hspace{-.03\linewidth}
		\includegraphics[width=1.03\linewidth]{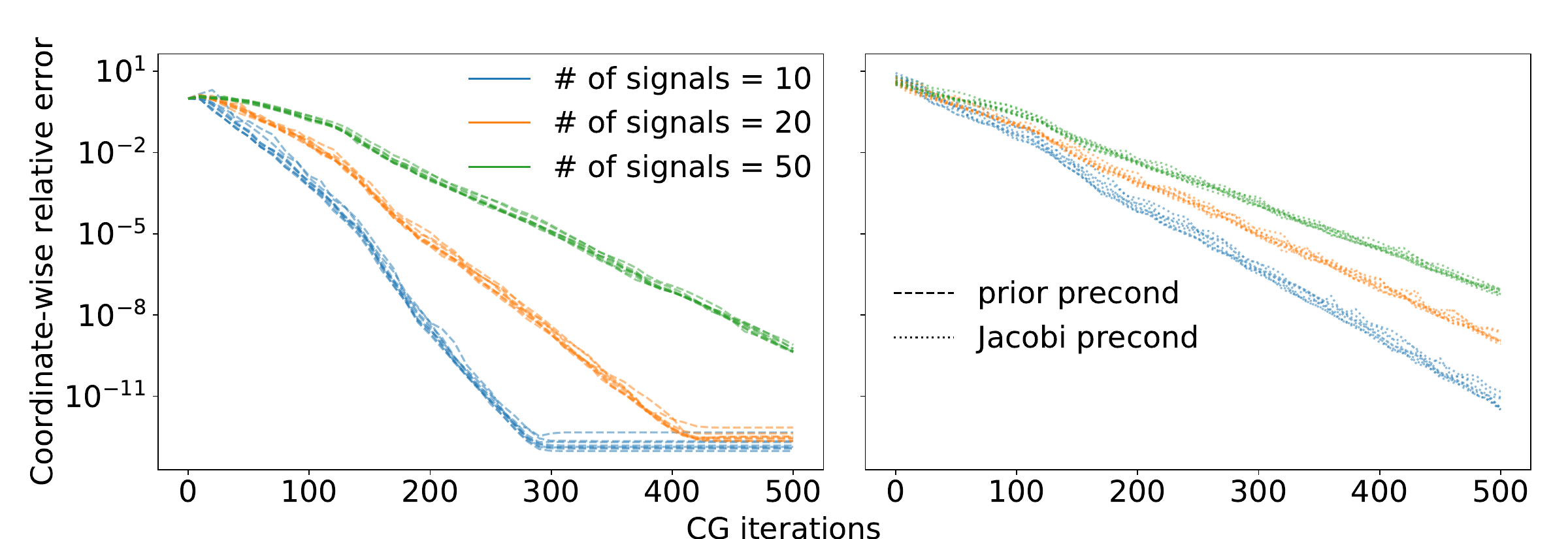}
		\caption{Plots of the CG approximation errors (with the same error metric as used in Figure~\ref{fig:cg_convergence_plot_for_simulated_data}) as a function of the number of CG iterations. Shown on the left is under the prior preconditioner and on the right is under the Jacobi preconditioner. The three different colors corresponds to the three different posterior conditional distributions of $\bbeta$ with the varying numbers of true signals. Within the same color, the different lines correspond to the different random draws of the right-hand vector $\bb$ generated as in \eqref{eq:cg_sampler_target_vector}.}
		\label{fig:cg_convergence_plot_with_random_seed_for_simulated_data}
	\end{figure}

	\begin{figure}
		\hspace{-.03\linewidth}
		\includegraphics[width=1.04\linewidth]{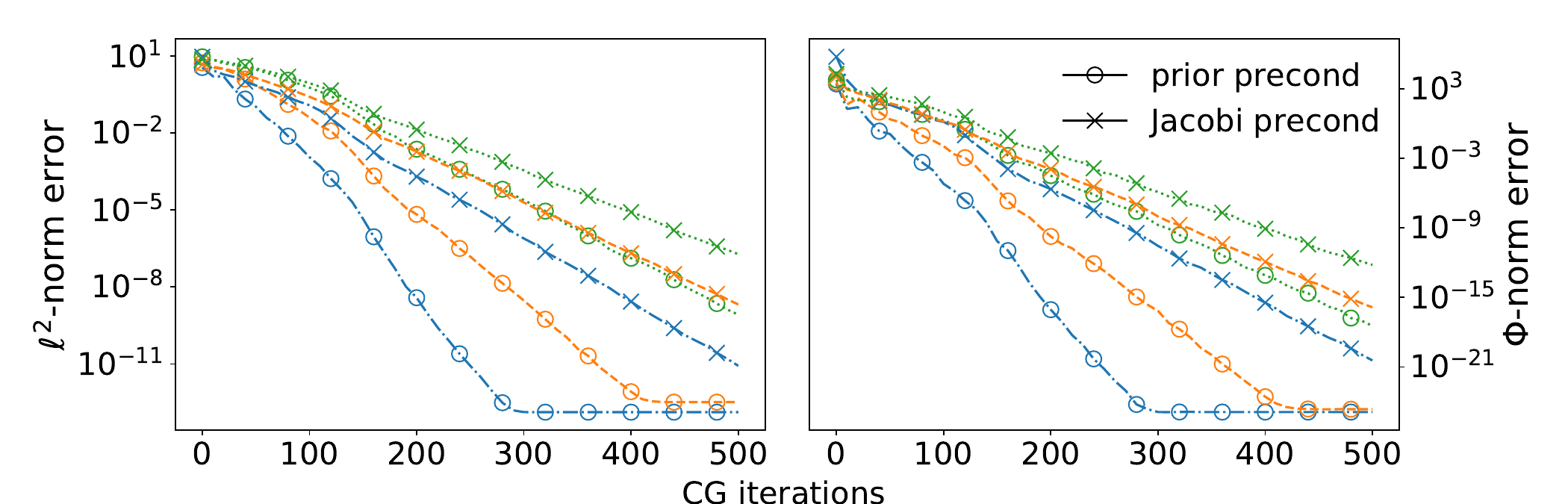}
		\caption{Plots of the $\ell^2$-norm (on the left) and the $\bPhi$-norm (on the right) between $\bbeta_k$ and $\bbeta_{\textrm{direct}}$ as a function of the number of CG iterations. Other than the use of the two alternative error metrics for the $y$-axes, each of the plotted lines directly corresponds to the one with the same color and marker in Figure~\ref{fig:cg_convergence_plot_for_simulated_data}.}
		\label{fig:cg_convergence_plot_with_different_error_metrics_for_simulated_data}
	\end{figure}

	\begin{figure}
	\centering
	\includegraphics[width=.975\linewidth]{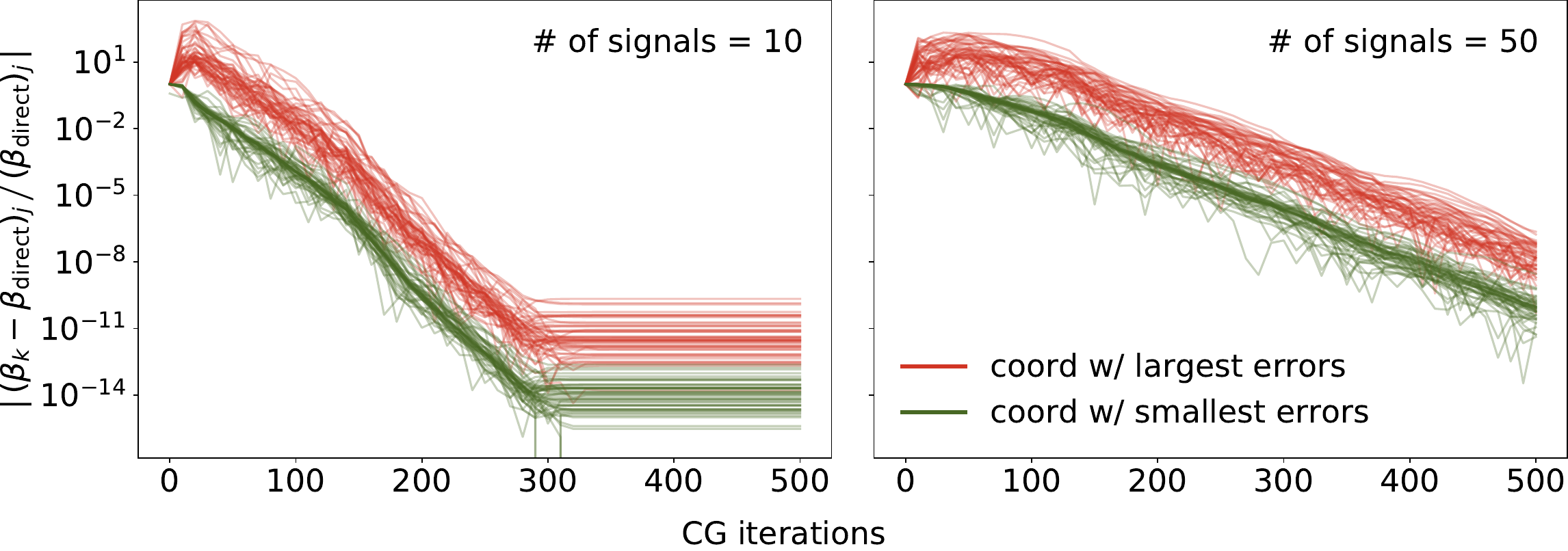}
	\caption{%
		Fifty largest and smallest of coordinate-wise relative errors $| (\bbeta_k - \bbeta_{\textrm{direct}})_j / (\bbeta_{\textrm{direct}})_j |$ as a function of the number of CG iterations.
		The left and right plot correspond to the CG sampler applied to the synthetic posteriors with 10 and 50 signals respectively.
	}
	\label{fig:cg_rel_coord_wise_error_traj_for_simulated_data}
	\end{figure}

	\begin{figure}
	\centering
	\includegraphics[width=\linewidth]{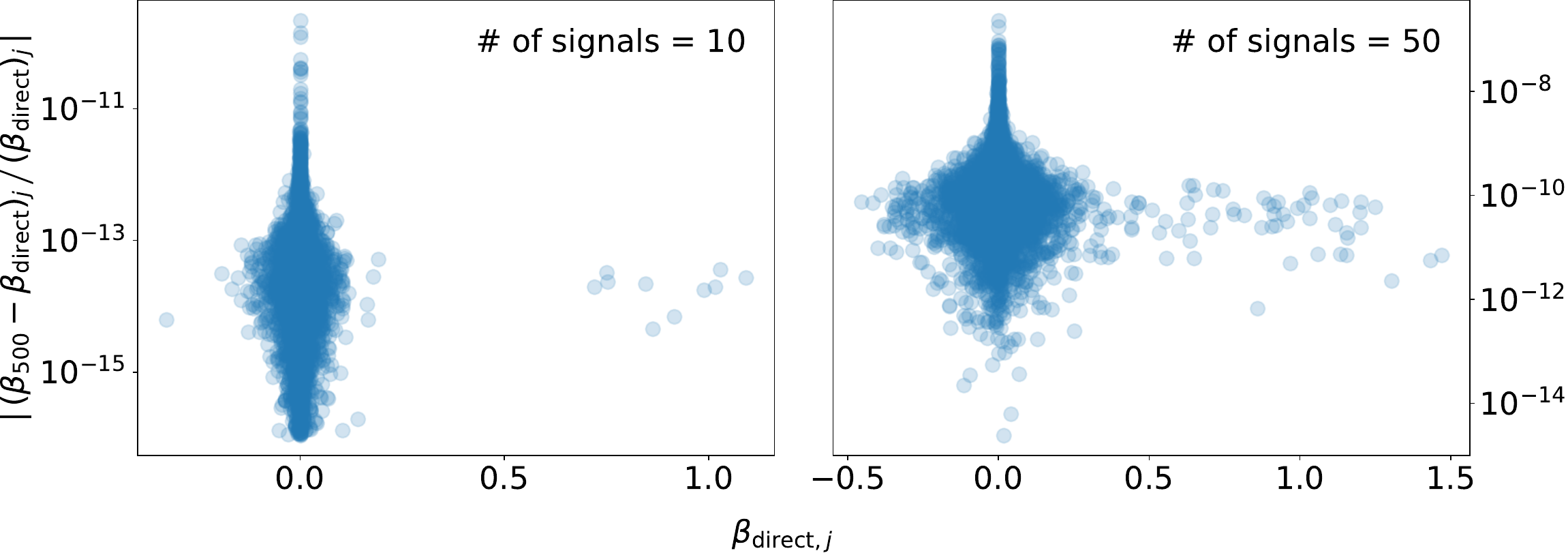}
	\caption{%
		Coordinate-wise relative error after 500 CG iterations, $| (\bbeta_{500} - \bbeta_{\textrm{direct}})_j / (\bbeta_{\textrm{direct}})_j |$, plotted against the exact solution value $\bbeta_{\textrm{direct}, j}$.
		As can be seen, there is no obvious relation between the magnitudes of relative error and the solution values.
	}
	\label{fig:cg_coef_val_vs_final_rel_error_for_simulated_data}
	\end{figure}

\FloatBarrier
\subsection{Effects of correlation among predictors}
\label{sec:effects_of_correlation}
As discussed in Section~\ref{sec:cg_convergence_on_simulated_data}, the convergence rate of CG sampler is also a function of correlation among the predictors as well as the number of true signals.
To demonstrate this, we repeat the experiment of Section~\ref{sec:cg_sampler_demonstration} with a synthetic design matrix having independent columns but otherwise with the exact same set-ups.
The design matrix has its entries simulated from i.i.d.\ Gaussians and is subsequently standardized.
Comparing Figure~\ref{fig:cg_convergence_plot_for_simulated_data_with_iid_design} below with Figure~\ref{fig:cg_convergence_plot_for_simulated_data}, it is clear that CG converges more quickly for the same number of true signals when applied to the posteriors under this set-up.

\begin{figure}
\centering
\includegraphics[width=0.75\linewidth]{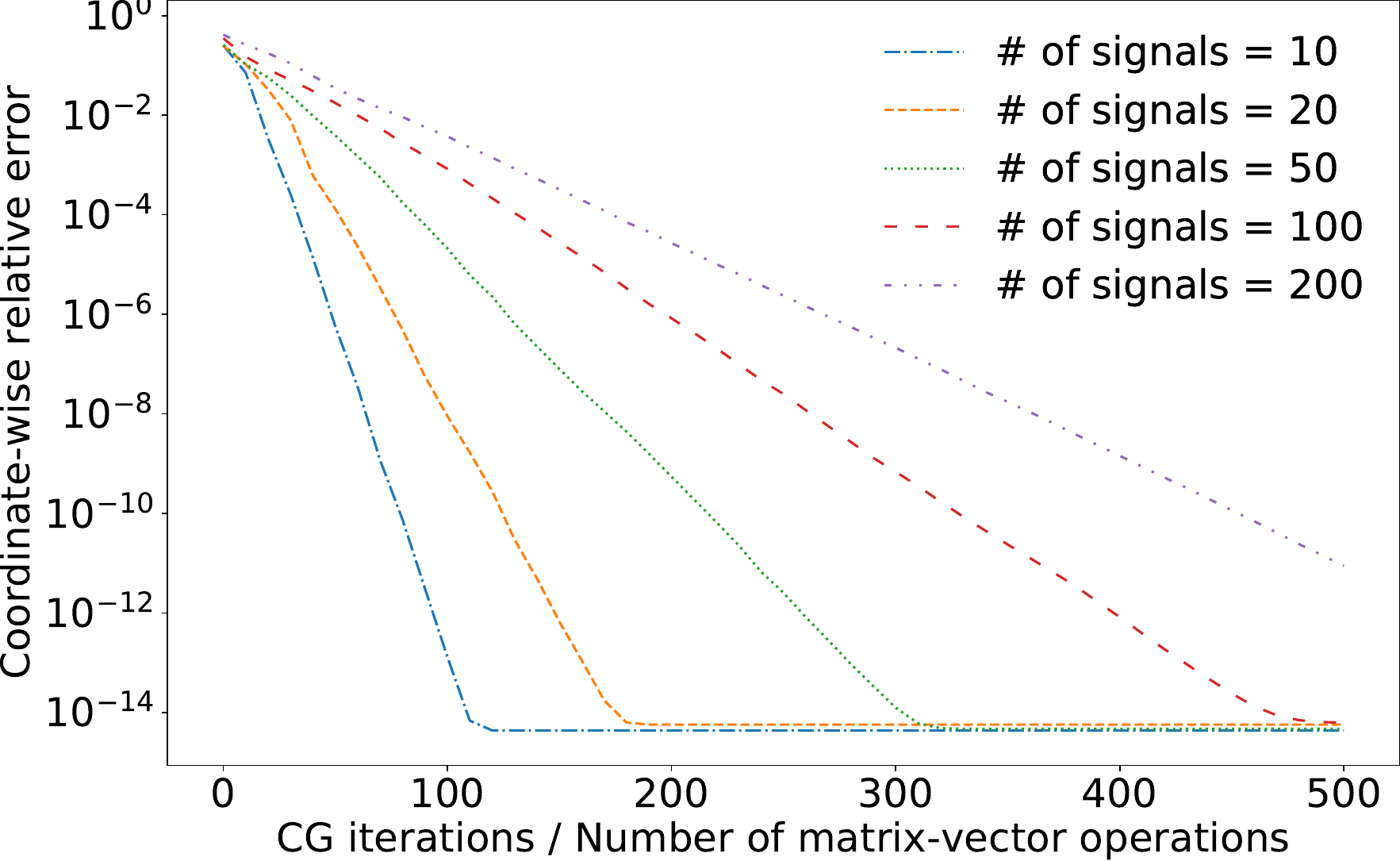}
\caption{%
	Plot of the prior-preconditioned CG approximation error  vs.\ the number of CG iterations.
	The CG sampler is applied to synthetic sparse regression posteriors based on the synthetic design matrix with independent columns. 
	The different line styles correspond to the different numbers of true signals.
}
\label{fig:cg_convergence_plot_for_simulated_data_with_iid_design}
\end{figure}

\subsection{Effects of number of factors in synthetic design matrix}
\label{sec:effects_of_num_factors}
Our results in Section~\ref{sec:cg_sampler_demonstration} are based a synthetic design matrix with $m = 99$ underlying factors as given in \eqref{eq:factor_model_design}.
Here we repeat the same simulation but using synthetic design matrices with a larger ($m = 199$) and smaller ($m = 49$) number of underlying factors.
As we vary the number of factors, we keep the eigenvalues of the resulting covariance matrix $\textrm{Cov}(\x)$ uniformly spaced in the range $[1, 100]$ with distance of $99 / m$ in-between.

To be more precise, following the procedure described in Section~\ref{sec:simulation_setup}, we first sample a set of $m = 49$ and $= 199$ orthonormal vectors $\bm{u}_1, \ldots, \bm{u}_{m} \in \mathbb{R}^p$ uniformly from a Stiefel manifold.
We then set the predictor $\x_i$ for the $i$-th observation as
	\begin{equation}
	\label{eq:factor_model_with_varying_num_factors}
	\x_i = \sum_{\ell = 1}^{m} f_{i, \ell} \bm{u}_\ell + \bm{\epsilon}_i
		\ \text{ for }
		f_{i, \ell} \sim \normal \left(0, \left[ 100 - \frac{99}{m} \left( \ell - 1 \right) \right]^2 - 1 \right)
		\text{ and }
		\bm{\epsilon}_i \sim \normal \left( \bm{0}, \I_p \right).
	\end{equation}
As before, this is equivalent to sampling $\x_i \sim \normal \! \left(\bm{0}, \bm{U} \bm{D} \bm{U}^\transpose \right)$ for a diagonal matrix $\bm{D}$ with $\sqrt{D_{\ell \ell}} = \max \{100 - \frac{99}{m} (\ell - 1), 1\}$ and orthonormal matrix $\bm{U}$ sampled uniformly from the space of orthonormal matrices.

It is worth noting that more factors do \textit{not} mean more correlations among the predictors.
In fact, we empirically find that more factors lead to less correlations among the predictors (Figure~\ref{fig:pairwise_corr_hist_with_varying_number_of_factors}).
This makes intuitive sense ---
when the predictors' variability is concentrated on a smaller number of factors, it induces stronger correlations among the predictors.

Figure~\ref{fig:cg_convergence_plot_with_varying_num_factors} shows the results of the CG sampler performance evaluation based on the synthetic design matrices with $m = 49$ and $m = 199$ underlying factors.
The simulation results here show the same patterns as Section~\ref{sec:cg_sampler_demonstration} for the most parts, but also provide us with additional insight.
The new pattern to emerge here is that the CG convergence rate is faster with a fewer number of underlying factors behind the design matrix.
In light of the observation that fewer factors correspond to more correlation (Figure~\ref{fig:pairwise_corr_hist_with_varying_number_of_factors}), however, this new finding is again explained by the overall principle as previously observed: the sparser the regression coefficient posteriors, the faster the prior-preconditioned CG's convergence rate.
Stronger correlation among the predictors make it difficult for the likelihood to separate out significant coefficients from the rest.
This leads to less sparse regression coefficients under the posterior and hence to a slower CG convergence.

Another curious phenomenon observed here is that the Jacobi preconditioner becomes competitive with the prior preconditioner in the $m = 49$ case with $50$ true signals.
Here, the posterior is less sparse due to both the strong correlation among the predictors and the large number of signals.
It is unclear whether this is a general pattern that would hold outside the specific generative model we chose for design matrices here.
Although the Jacobi preconditioner is accepted as one of the best choices for a diagonally dominant matrix \citep{golub2012matrix}, as is the case for the conditional precision of $\bbeta$ under the sparse regression posterior \eqref{eq:beta_conditional}, it is difficult to obtain a quantitative result on the eigenvalue structure of the Jacobi-preconditioned matrix.
Diagonal dominance plays prominent roles in qualitative properties of a matrix, such as its invertibility and positive definteness, but quantitative results remain scarce \citep{golub2012matrix, horn2012matrix-analysis}.

At a minimum, our results here suggest that the Jacobi preconditioner warrants consideration despite lagging behind the prior-preconditioner in all but one of our examples.
We still recommend the prior-preconditioner as the first choice given its sound theoretical support and the overall superiority demonstrated in the range of examples, including the real data case of Section~\ref{sec:application}.
That said, it is easy enough to adaptively choose the better of the two preconditioners for a specific posterior with negligible computational overhead, e.g. by comparing their relative performance at every 100 Gibbs iterations.
All in all, the CG sampler performs well under both the prior and Jacobi preconditioner, both delivering the convergence within $k \ll p$ iterations for solving the $p \times p$ linear system.

\begin{figure}
\centering
	\subfigure{
		\includegraphics[width=.5\linewidth]{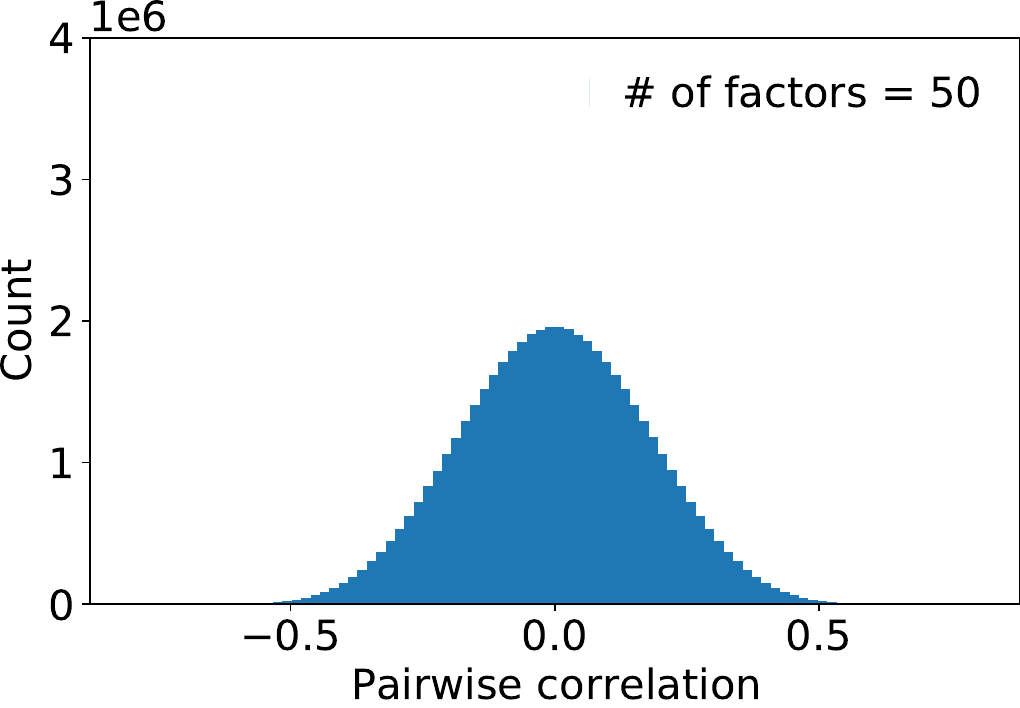}
	}
	\subfigure{
		\includegraphics[width=.5\linewidth]{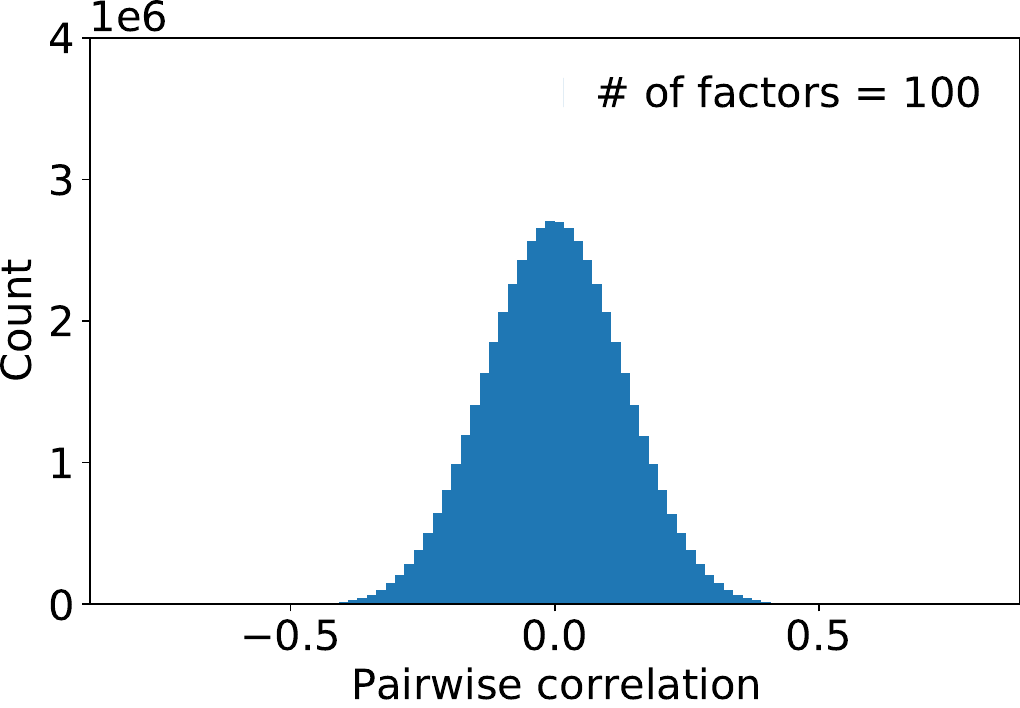}
	}
	\subfigure{
		\includegraphics[width=.5\linewidth]{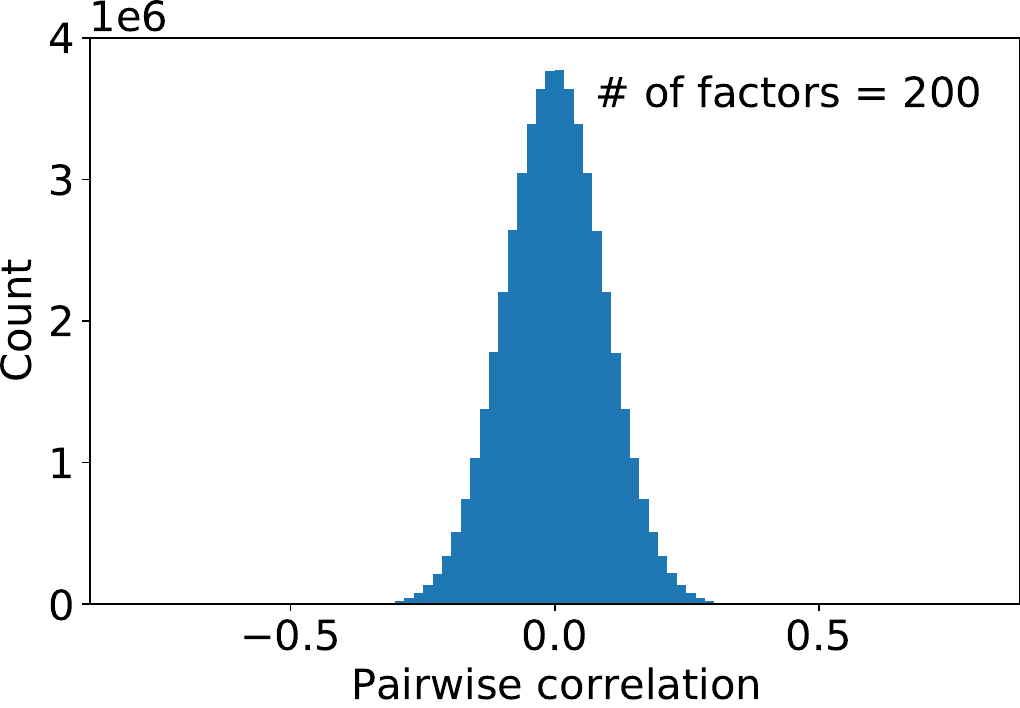}
	}
\caption{%
	Histograms of pairwise correlations among the $10{,}000$ predictors simulated with the varying numbers of underlying factors $m = 49$, $99$, and $199$ according to \eqref{eq:factor_model_with_varying_num_factors}.
}
\label{fig:pairwise_corr_hist_with_varying_number_of_factors}
\end{figure}

\begin{figure}
\centering
\subfigure[Using the synthetic design matrix with $m = 49$ underlying factors]{
	\includegraphics[width=0.75\linewidth]{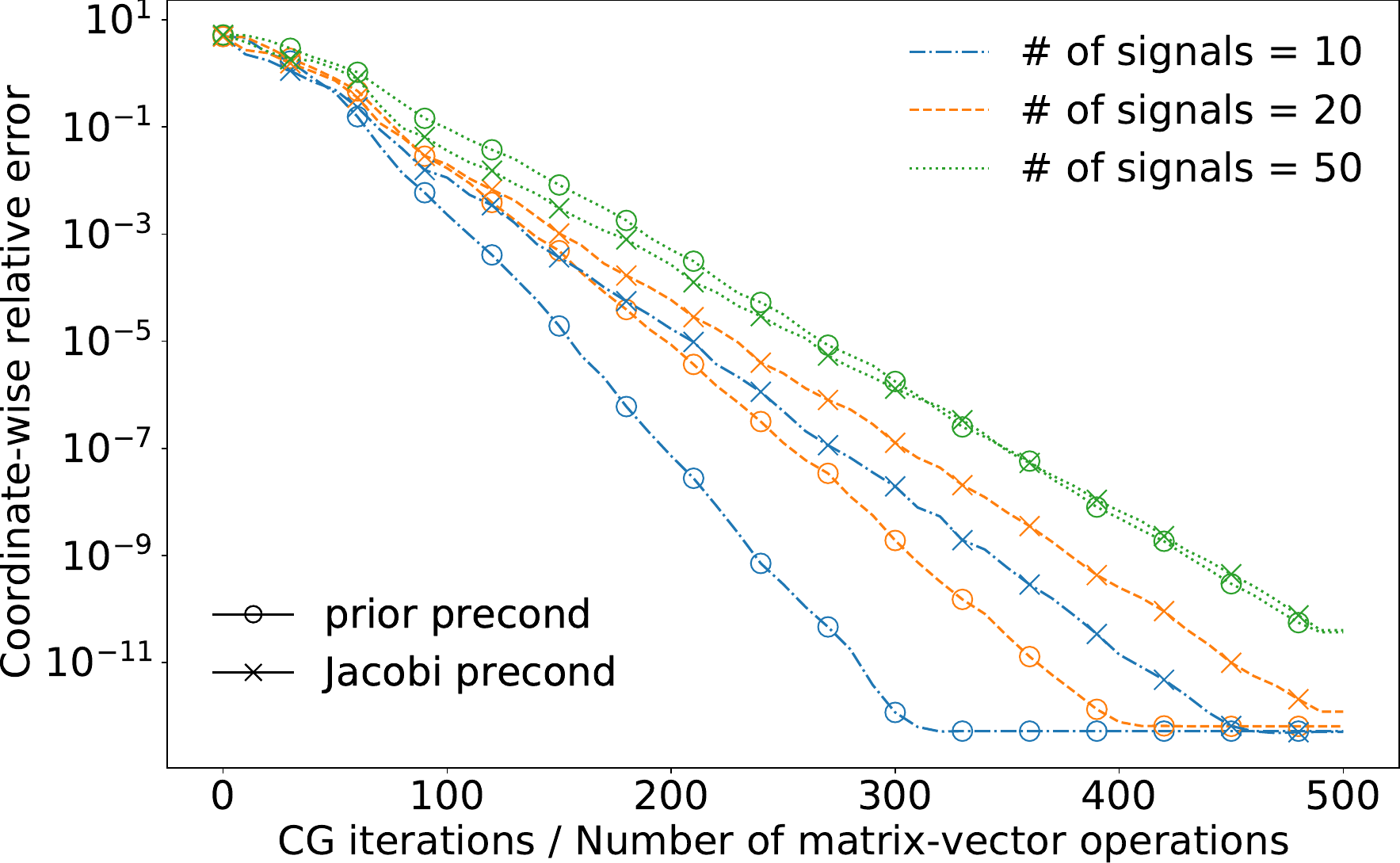}
}
\subfigure[Using the synthetic design matrix with $m = 199$ underlying factors]{
	\includegraphics[width=0.75\linewidth]{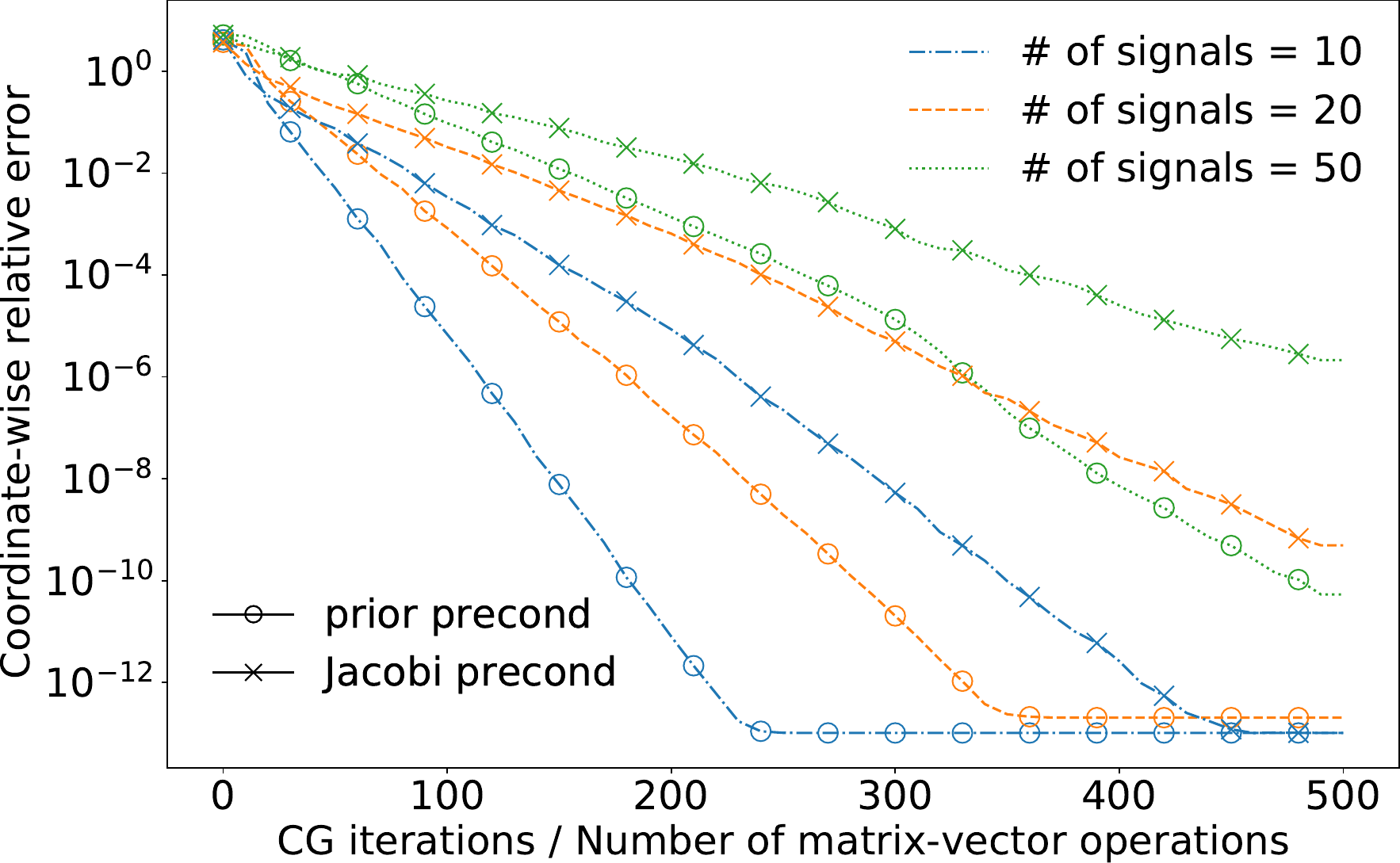}
}
\caption{%
	Plot of the CG approximation error vs.\ the number of CG iterations.
	The set-up is identical to that for Figure~\ref{fig:cg_convergence_plot_for_simulated_data}, except for the fact that the CG sampler is applied to synthetic sparse regression posteriors base on synthetic design matrices with $m = 49$ and $m = 199$ underlying factors, as opposed to that with $m = 99$. The different line styles correspond to the different numbers of true signals. The circle and cross markers denote the uses of the prior and Jacobi preconditioners.
}
\label{fig:cg_convergence_plot_with_varying_num_factors}
\end{figure}

\FloatBarrier
\section{CG-accelerated Gibbs on synthetic data of Section~\ref{sec:cg_sampler_demonstration}: quality of posterior samples and computational speed}
\label{sec:cg_sampler_accuracy_and_speed_for_synthetic_data}

In Section~\ref{sec:cg_sampler_demonstration}, we focus on the prior-preconditioned CG's performance within one iteration of the Gibbs sampler.
Here we more holistically compare the performance of the two Gibbs samplers: one based on the CG sampler and the other on the direct linear algebra sampler.
We show in particular that, when using the criteria of Section~\ref{supp:termination_criteria_for_cg_sampler} in terminating CG iterations, the outputs of the two Gibbs samplers are statistically indistinguishable.
This confirms that we can use the CG sampler as a drop-in replacement within the Gibbs sampler to deal with conditional updates of $\bbeta$ from the high-dimensional Gaussian \eqref{eq:beta_conditional}.
Consequently, the prior-preconditioned CG's performance in solving the linear system \eqref{eq:linear_system_for_cg_sampler} directly translates into the performance of the CG-accelerated Gibbs sampler.

Also investigated in this section is how the two Gibbs sampler perform in terms of actual computing time.
As evident from our discussion in Section~\ref{sec:complexity_analysis}, computational gains from CG-acceleration depends as much on the size of a problem as posterior sparsity level.
Therefore, we complement our simulation study of Section~\ref{sec:cg_sampler_demonstration} by varying not only the number of true signals but also the size of the synthetic design matrices.
Our results clearly show that the CG-accelerated Gibbs sampler delivers increasing advantage as the problem size grows.

Finally, we investigate how the CG sampler's performance depends on the choice of shrinkage prior. 
The Bayesian bridge prior $\pi(\beta_j \given \tau) \propto \tau^{-1} \exp\left( - | \beta_j / \tau |^\alpha \right)$ allows us to conveniently vary its behavior through the exponent $\alpha > 0$, so we assess this question by running the CG-accelerated Gibbs sampler with $\alpha \in \{1/4, 1/8\}$ in addition to $\alpha = 1/2$ as used in Section~\ref{sec:cg_sampler_demonstration}. 
We find that the smaller $\alpha$'s improve separation of true signals from the rest, which in turn induce faster convergences of CG. 
We also replicate our findings using synthetic data sets generated from different random seeds, thereby verifying that  our findings are not artifacts of quirks in a specific synthetic data set.

\subsection{Accuracy of CG sampler}
\label{sec:cg_sampler_accuracy}
	Since CG technically does not yield the exact solution when terminated at $k < p$ iterations, we assess the accuracy of the samples generated by the CG-accelerated Gibbs by comparing them against the ``ground truth'' samples generated by the direct Gibbs.
	While the nascent field of \textit{probabilistic numerics} provides potentially useful theoretical quantification of the CG sampler accuracy, the current state-of-the-art appears to fall short of practical uncertainly quantification \citep{cockayne2018bayesian-cg, hennig2019discussion-on-bayes-cg}.
	Instead, we empirically demonstrate that perturbation, if any, of the target distribution due to the CG approximation error is so small that it is essentially negligible within the Bayesian sparse regression context. 
	From a more qualitative perspective, the theoretical results of \cite{roberts1998perturbed_chain} guarantee that, under a sufficiently small numerical error, a geometrically ergodic chain retains its convergence rate and has its stationary distribution close to the original.

	We compare the two sets of samples in terms of the primary parameter of interest $\bbeta$.
	The mixing of $\bbeta$ is generally fast for any fixed $\tau$, but the dependency between $\bbeta$ and $\tau$ somewhat reduces the overall mixing rate.
	To ensure that the effective sample sizes for $\bbeta$ are large enough to adequately characterize the stationary distribution, therefore, we employ an empirical Bayes approach.
	We first find a value $\hat{\tau}$ which approximately maximizes the marginal likelihood through Monte Carlo expectation-maximization algorithm \citep{park2008bayes-lasso}.
	We then run the two samplers conditional on $\tau = \hat{\tau}$ for 5,000 iterations.

	We test for differences between the two sets of the MCMC samples as follows. We first set $\bm{\hat{\beta}}_{\text{bench}}$ and $\bm{\hat{\beta}}_{\text{cg}}$ to be the posterior means estimated by averaging the samples from the direct Gibbs (used as a benchmark) and CG-accelerated Gibbs.
	The plots on the left column of Figure~\ref{fig:mcmc_estimators_first_moment_comparison_plot_for_simulated_data} compares these two estimators graphically as an informal sanity check.
	We then estimate the effective sample sizes of $\beta_j$ from the respective samplers using the R {CODA} package \citep{plummer2006coda}. These estimated effective sample sizes can be used to estimate the Monte Carlo standard deviations $\hat{\sigma}_j$ of the differences ${\hat{\beta}}_{\text{bench}, \, j} - {\hat{\beta}}_{\text{cg}, \, j}$. When the two sets of samples have the same stationary distribution, the standardized differences $({\hat{\beta}}_{\text{bench}, \, j} - {\hat{\beta}}_{\text{cg}, \, j} ) / \hat{\sigma}_j$ are approximately distributed as the standard Gaussians by the Markov chain central limit theorem \citep{geyer11intro_to_mcmc}.
	The plots on the right column of Figure~\ref{fig:mcmc_estimators_first_moment_comparison_plot_for_simulated_data} confirm that the histograms of the standardized distances closely match the ``null'' Gaussian distribution.

	We perform the same diagnostic on the estimators of the posterior second moment of $\bbeta$ and obtain similar results as shown in Figure~\ref{fig:mcmc_estimators_second_moment_comparison_plot_for_simulated_data}.

	Finally, we compare the effective sample sizes (ESS) of $\beta_j$'s generated by the two Gibbs samplers.
	Given that the two sets of samples are statistically indistinguishable, we expect their ESS's to also coincide.
	Figure~\ref{fig:mcmc_estimators_ess_comparison_plot_for_simulated_data} confirm that this is indeed the case;
	besides some natural variations from the statistical estimations of ESS (by the R package \textsc{coda}), we find no obvious differences in the ESS's from the two Gibbs samplers.

	 \begin{figure}
	 \centering
	 \includegraphics[width=0.875\linewidth]{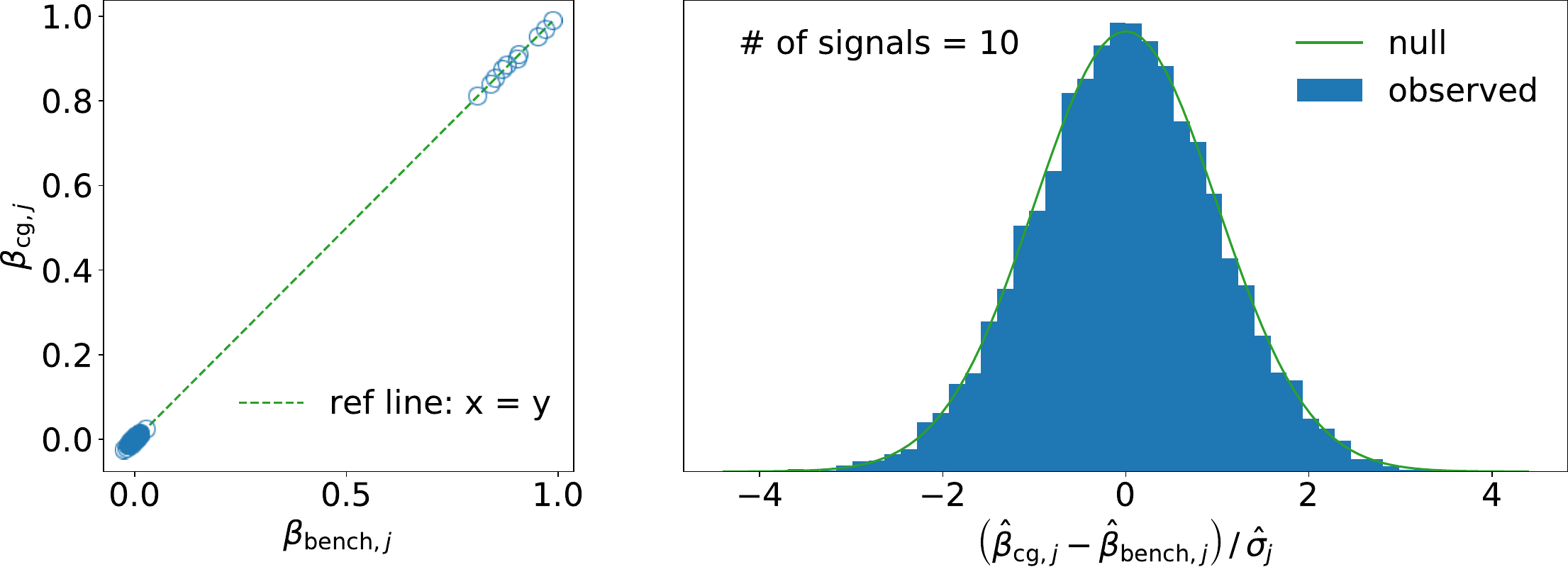}
	 \vspace*{.5\baselineskip}

	 \includegraphics[width=0.875\linewidth]{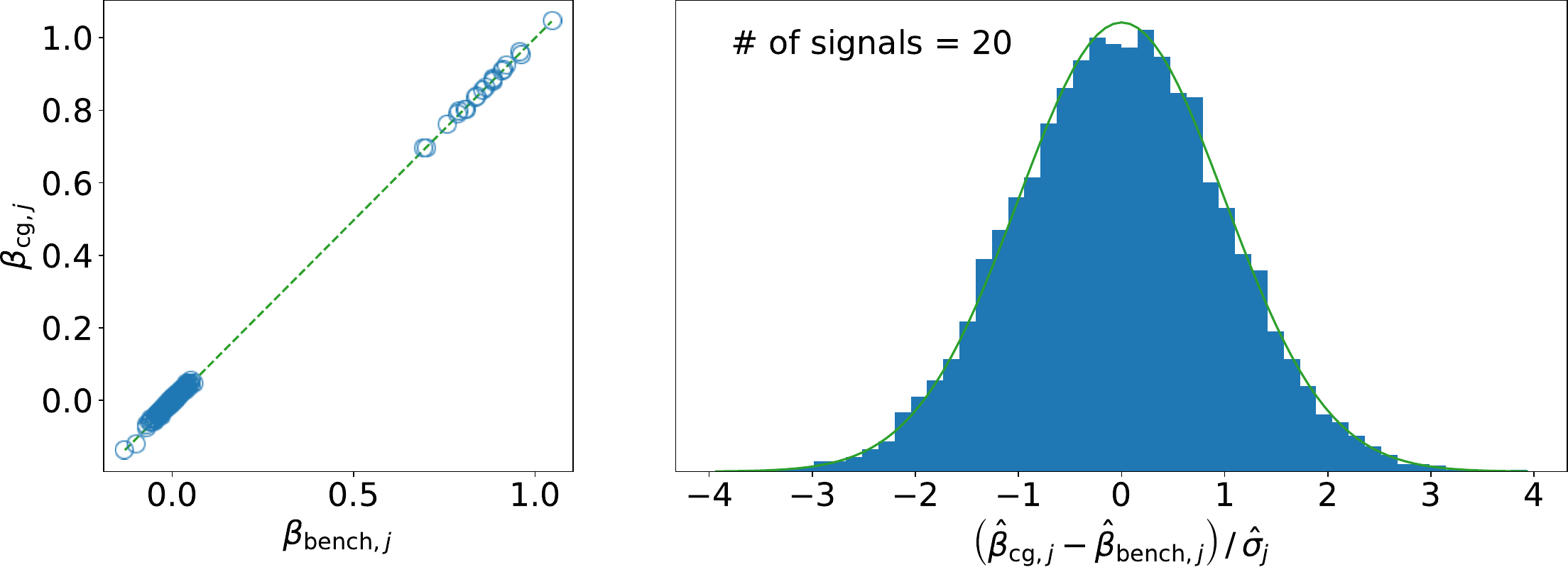}
	 \vspace*{.5\baselineskip}

	 \hspace*{-.035\linewidth}
	 \includegraphics[width=0.9\linewidth]{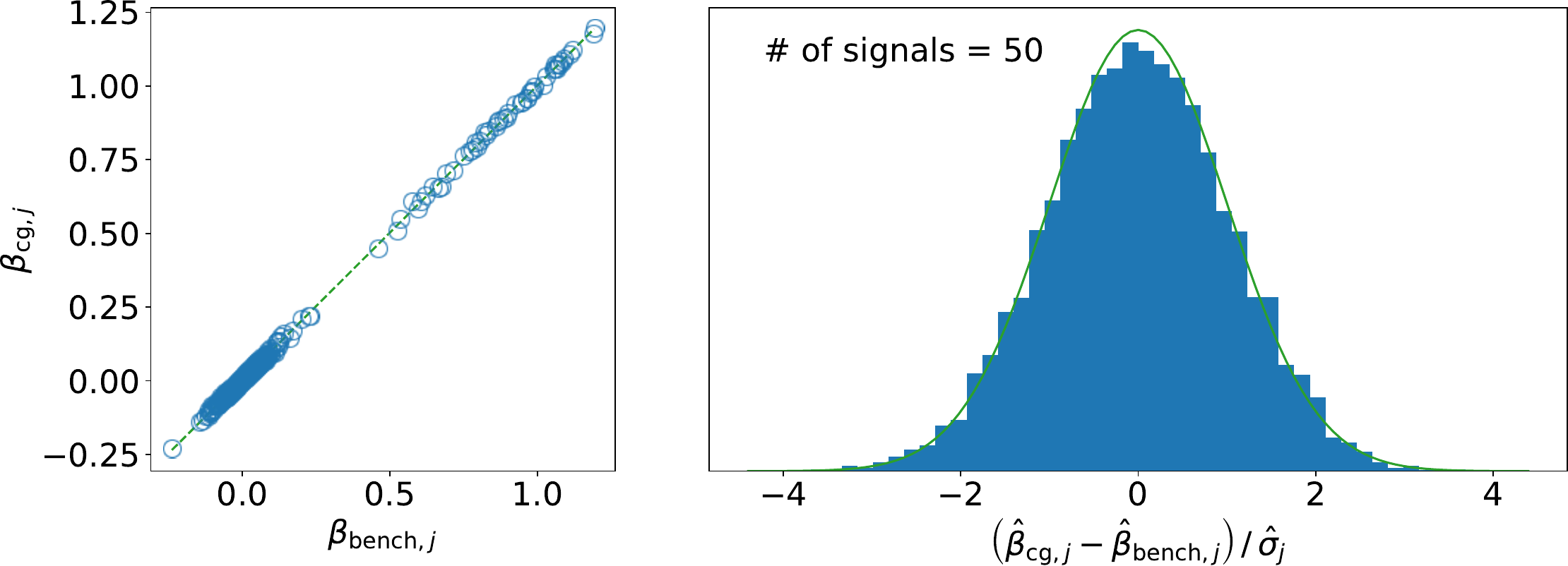}
	 \vspace*{-.5\baselineskip}
	 \caption{%
	 	Diagnostic plots to check for statistically significant differences between the two MCMC outputs.
	 	The three rows of the figure correspond to the results based on the synthetic data simulated with 10, 20, and 50 true signals.
	 	The plots on the left compare the regression coefficient estimates (posterior means) between those based on the direct and CG-accelerated Gibbs samplers.
	 	On the right are normalized histograms for the standardized differences $(\hat{\beta}_{\text{bench}, j} - \hat{\beta}_{\text{cg}, j}) / \hat{\sigma}_i$, where $\hat{\sigma}_i^2$ is an estimate of the Monte Carlo variance of $\hat{\beta}_{\text{bench}, j} - \hat{\beta}_{\text{cg}, j}$.
	 	Gaussianity of the histogram indicates no statistically significant difference between the two MCMC outputs.
	   }
	 \label{fig:mcmc_estimators_first_moment_comparison_plot_for_simulated_data}
	 \end{figure}

	 \begin{figure}
	 \centering
	 \includegraphics[width=0.875\linewidth]{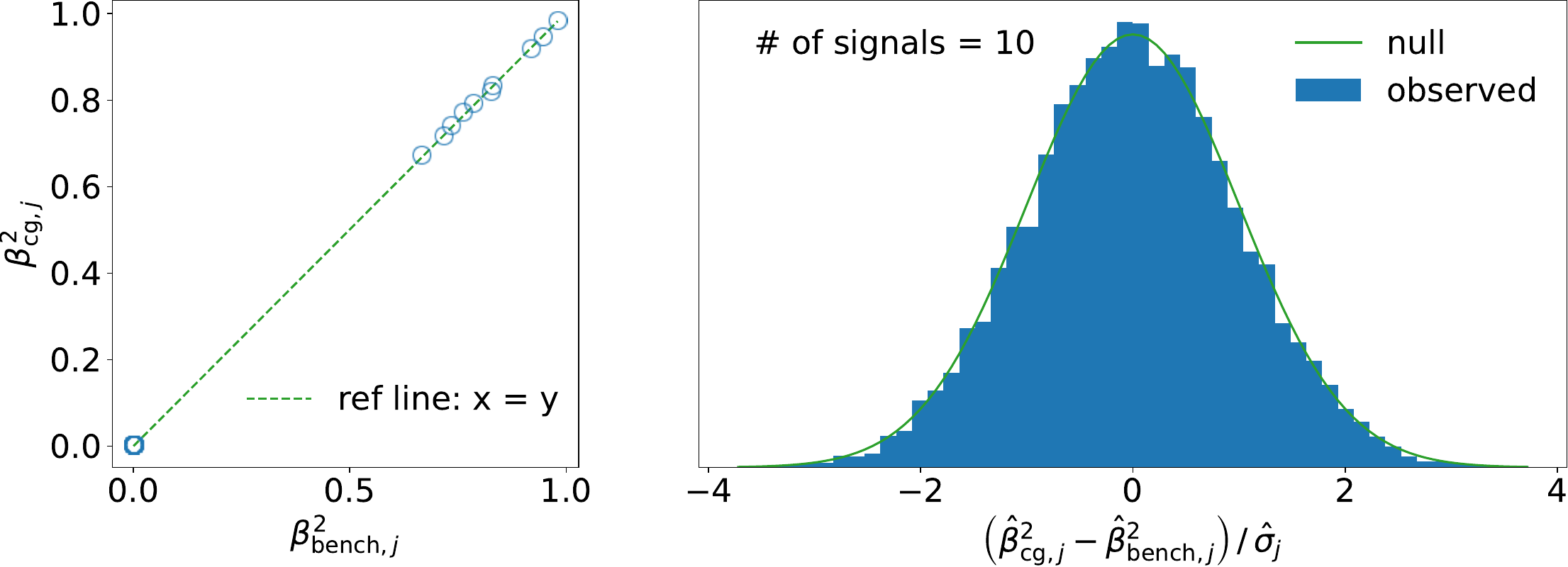}
	 \vspace*{.5\baselineskip}

	 \includegraphics[width=0.875\linewidth]{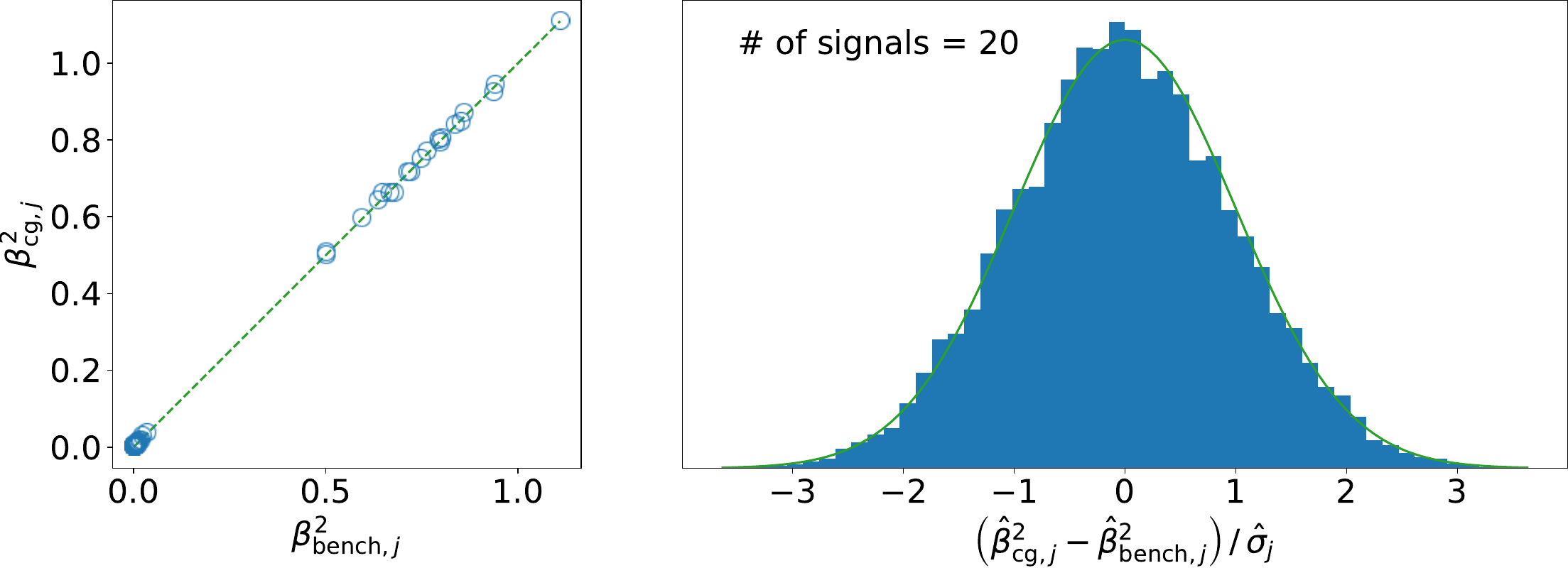}
	 \vspace*{.5\baselineskip}

	 \hspace*{-.035\linewidth}
	 \includegraphics[width=0.9\linewidth]{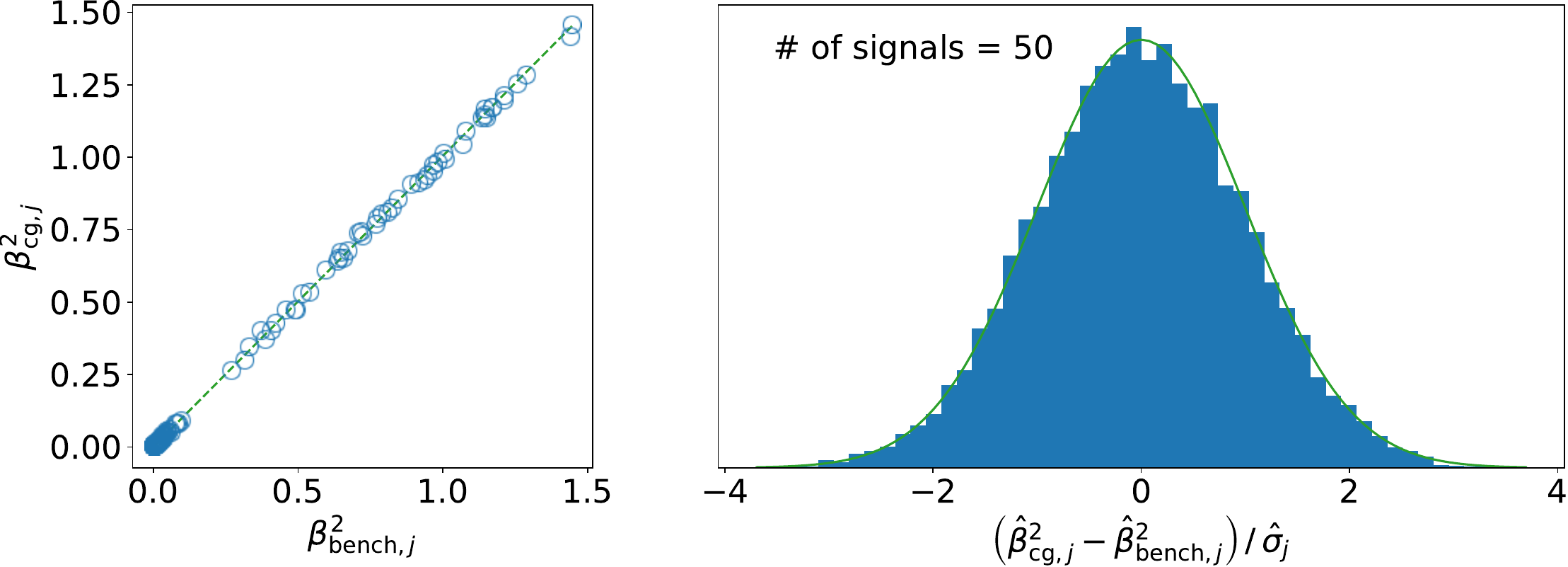}
	 \vspace*{-.5\baselineskip}
	 \caption{%
 	 	Diagnostic plots as in Figure~\ref{fig:mcmc_estimators_first_moment_comparison_plot_for_simulated_data} to check for statistically significant differences between the two MCMC outputs.
 	 	The only difference from Figure~\ref{fig:mcmc_estimators_first_moment_comparison_plot_for_simulated_data} is that here we compare the second moment estimates $\hat{\beta}_{\text{bench}, j}^2$ and $\hat{\beta}_{\text{cg}, j}^2$.
 	 }
	 \label{fig:mcmc_estimators_second_moment_comparison_plot_for_simulated_data}
	 \end{figure}

	 \begin{figure}
	 \centering
	 \includegraphics[width=0.9\linewidth]{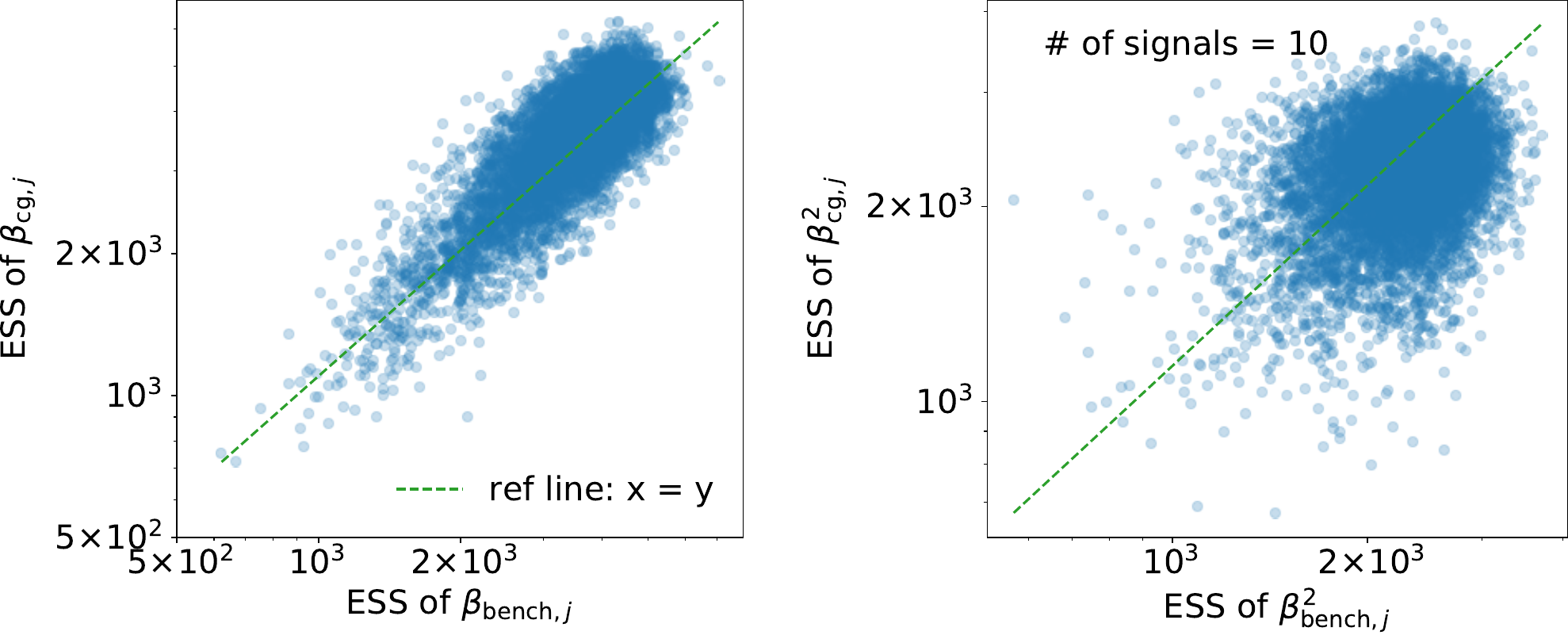}
	 \vspace*{.5\baselineskip}

	 \includegraphics[width=0.9\linewidth]{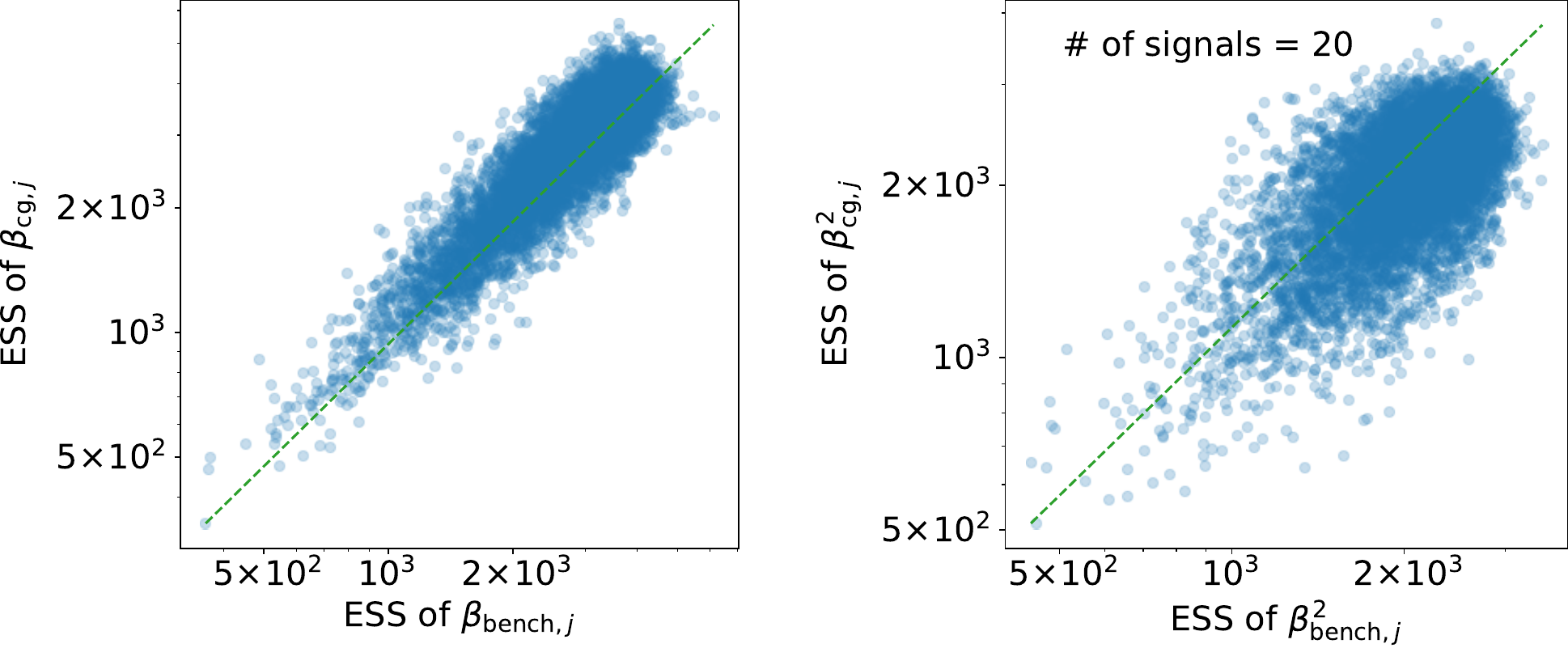}
	 \vspace*{.5\baselineskip}

	 \includegraphics[width=0.9\linewidth]{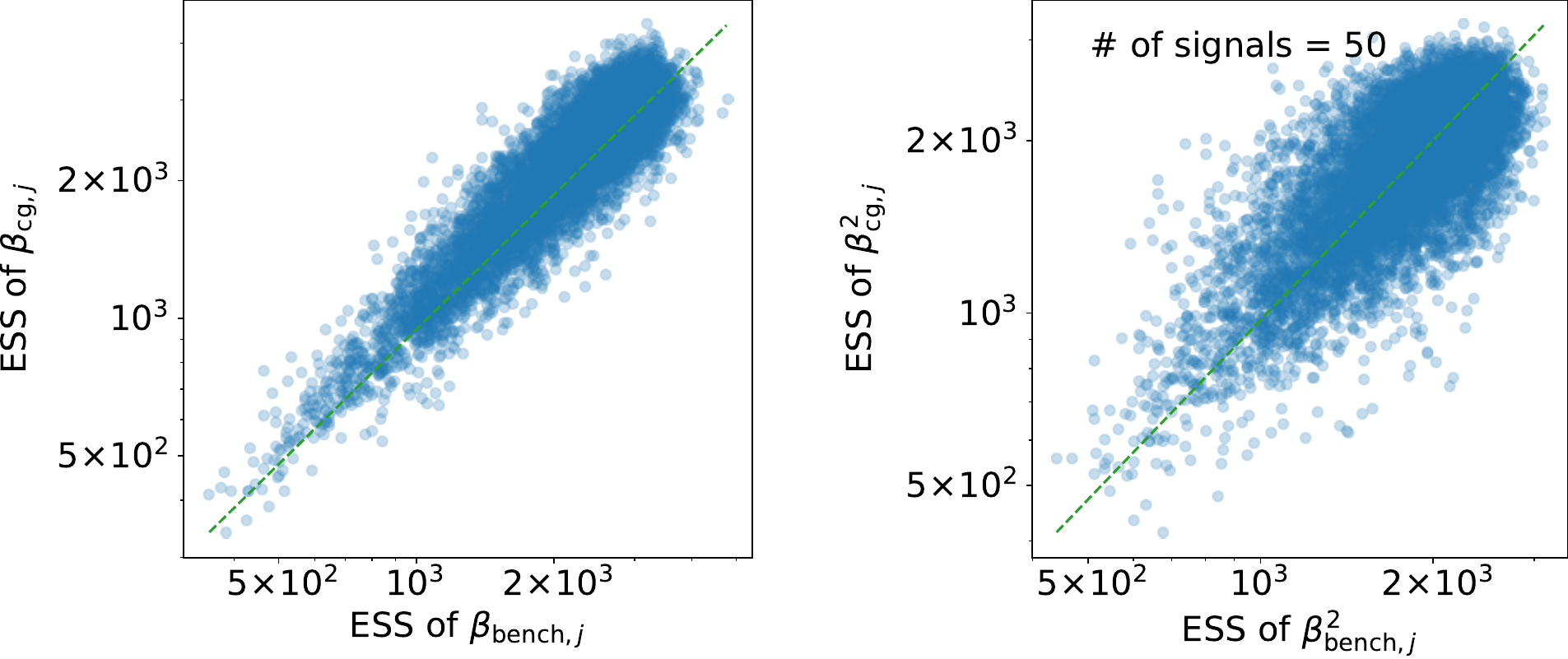}
	 \vspace*{-.5\baselineskip}
	 \caption{%
	 	Comparison of ESS's from the two Gibbs samplers.
	 	Each Gibbs sampler is run for 5,000 iterations.
	 	The three rows of the figure correspond to the results based on the synthetic data simulated with 10, 20, and 50 true signals.
	 	ESS's of $\beta_j$ are shown on the left and those of $\beta_j^2$ on the right.
 	 }
	 \label{fig:mcmc_estimators_ess_comparison_plot_for_simulated_data}
	 \end{figure}

\FloatBarrier

\subsection{Relative computational efficiency}
\label{sec:rel_computational_efficiency_on_synthetic_data}
The diagnostics of Section~\ref{sec:cg_sampler_accuracy} show that the two Gibbs samplers are essentially identical in their output.
Consequently, their relative computational efficiency as MCMC algorithms --- as measured by ESS per unit time, for example --- can be quantified directly by their relative computational time.
We thus compare the actual computing times of the two Gibbs samplers run for 5,000 iterations as in Section~\ref{sec:cg_sampler_accuracy}.

Since computational gains from CG-acceleration depends as much on the size of a problem as posterior sparsity level, we complement our simulation study of Section~\ref{sec:cg_sampler_demonstration} by varying not only the number of true signals but also the size of the synthetic design matrices.
More specifically, we use the same data generating model as described in Section~\ref{sec:simulation_setup} but generates the design matrix $\X$ of size $12{,}500 \times 5{,}000$; $25{,}000 \times 10{,}000$; and $50{,}000 \times 20{,}000$.
We then carry out repeat the same simulations with these three matrices.
As in Section~\ref{sec:application}, we measure the computing times on a 2015 iMac equipped with an Intel Core i7 processor.\footnote{%
	The simulation here is extremely computationally intensive.
	In order to complete this simulation within a reasonable amount of time, therefore, we carried out the actual computations using the Joint High Performance Computing Exchange at Johns Hopkins University (\url{https://jhpce.jhu.edu/}).
	Since the linear algebra operations required for updating regression coefficients account for over 99\% of total computing times of both Gibbs samplers, we calculated the computational time we would have obtained had we run the Gibbs samplers on the original 2015 iMac as follows.
	We first measured on the 2015 iMac the computing times required for these linear algebra operations:
	matrix-matrix multiplication and Cholesky decomposition for each iteration of the direct method and matrix-vector multiplication by $\X$ and $\X^\transpose$ for each CG iteration.
	(These linear algebra benchmarks are based on single-threaded implementations; see Section~\ref{sec:optimizing_linear_algebra} for details on how we optimized each operation and on possibilities of multi-threading.)
	We then counted the number of the linear algebra operations required in the actual runs of the Gibbs samplers.
	Finally, we calculated the computing times we would have obtained on the 2015 iMac by multiplying the costs of these linear algebra operations with the number of times they were used in the Gibbs sampler runs.
	Overall, the simulation required over 10,000 hours of \cpu{} time (but less than 3 weeks in actual clock time thanks to parallelization) and 372 \textsc{gb} of memory.
}

Figure~\ref{fig:computing_time_vs_mat_size} shows the time required for posterior computation by the two Gibbs samplers as the number of signals and size of design matrix vary.
We see that the CG-accelerated Gibbs demonstrate increasing advantage over the direct Gibbs as the problem size grows.
This is as expected from our discussion in Section~\ref{sec:complexity_analysis}.
The computational complexity of the direct Gibbs is $O(n^2 p + p^3)$ while that of of the CG-accelerated Gibbs is approximately $O(nps)$, where $s$ is the number of $\tau \lshrink_j$'s --- and hence of $\beta_j$'s --- significantly away from $0$.
Therefore, the required computing time increases 8-fold every time $n$ and $p$ double in size.
On the other hand, the increase is roughly only 4-fold for the CG-accelerated Gibb.
For the previous statement to hold, the quantity $s$ must remain roughly constant over varying problem sizes as long as the underlying number of true signals remain fixed.
We are not aware of any theoretical results guaranteeing such behavior, but characterizing the cost of CG-accelerated Gibbs in this manner seems like a reasonable and conceptually useful approximation that agrees with our empirical results here.
Figure~\ref{fig:rel_speed_vs_mat_size_single} facilitate comparison of the two Gibbs samplers' performances by plotting their relative computational speed as the number of signals and problem size varies.

The CG-accelerated Gibbs is not necessarily faster than the direct Gibbs for smaller problems as seen in Figure~\ref{fig:computing_time_vs_mat_size}.
This is because existing computing architectures are typically more optimized for BLAS Level 3 operations, such as matrix-matrix multiplications used in the direct Gibbs, than for BLAS Level 2 operations, such as matrix-vector multiplication used in the CG-accelerated Gibbs (Section~\ref{sec:optimizing_linear_algebra}).
In newer computing architectures, however, there is an increasingly greater emphasis on high bandwidth and low latency, both of which are critical e.g.\ for high-performance (sparse) matrix-vector multiplications \citep{dongarra2016hpc_connjugate_gradient}.
We thus expect the advantage of CG-accelerated Gibbs on typical computing environments to grow over time as architectures and software adapt to modern large-scale applications.

\begin{figure}
\centering
\subfigure[Computing time vs.\ problem size]{
	\includegraphics[width=0.475\linewidth]{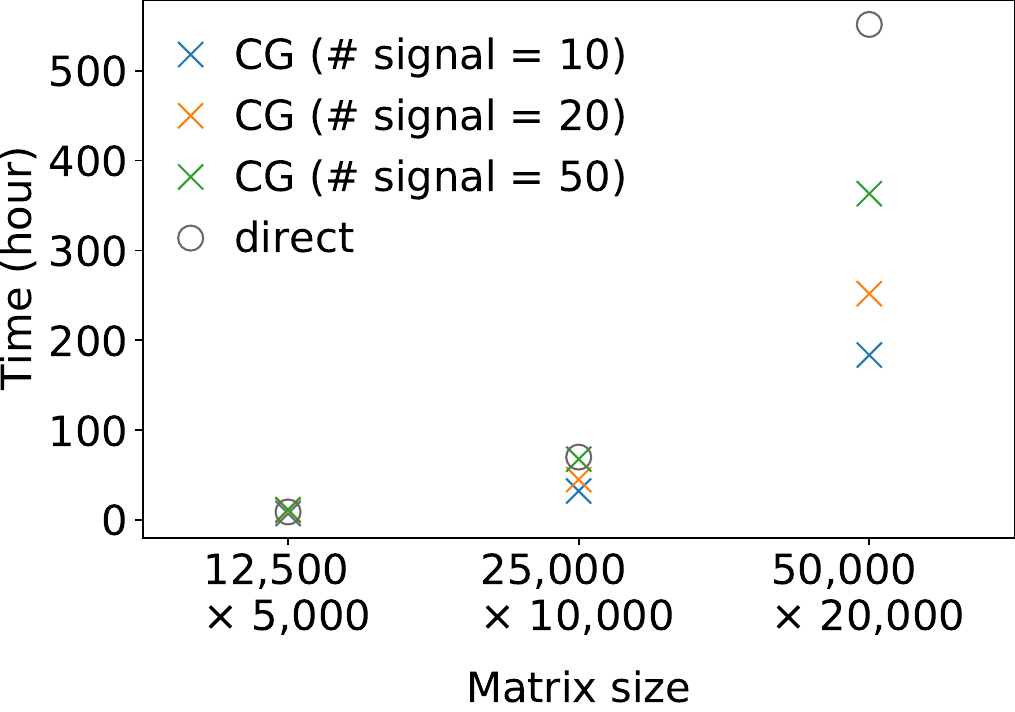}
}%
~
\subfigure[Same as the left figure, but with log-scale $y$-axis]{
	\includegraphics[width=0.475\linewidth]{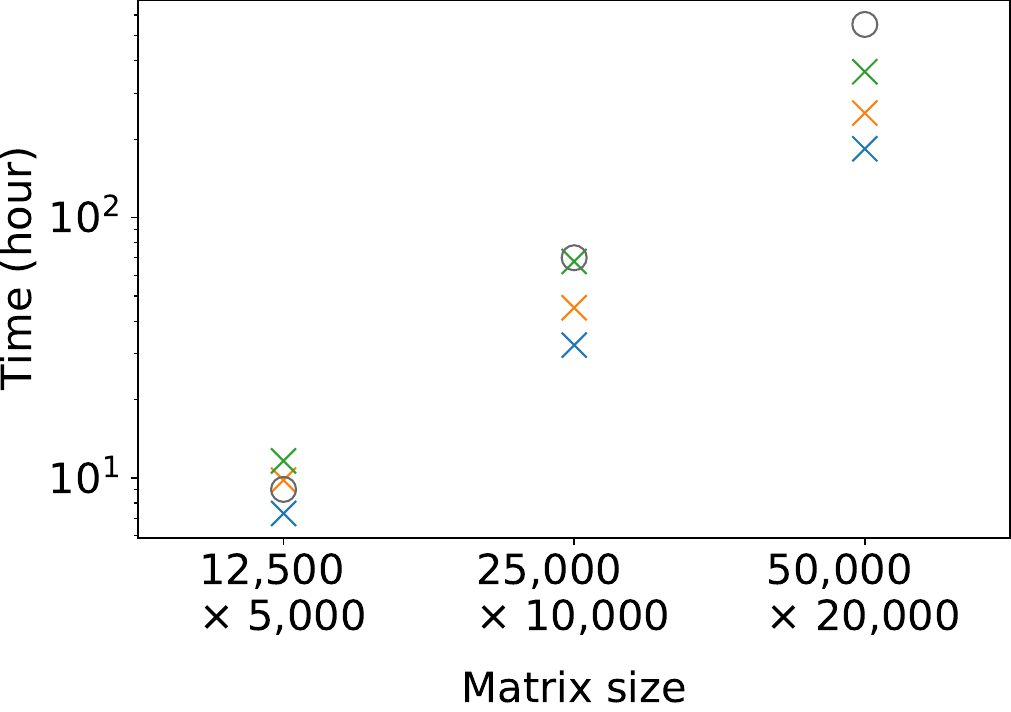}
}
\caption{%
	Posterior computation time required for 5,000 Gibbs sampler iterations.
	Note that the computational cost of the direct Gibbs sampler depends only on the problem size and not on the number of true signals in the data.
	The right figure plots computing time in log-scale, but otherwise shows information identical to that in the left figure.
	We can see that the direct Gibbs has a steeper slope in log-scale and hence more rapid increase in the computational cost as the problem size increases.
}
\label{fig:computing_time_vs_mat_size}
\end{figure}

\begin{figure}
\centering
\includegraphics[width=0.6\linewidth]{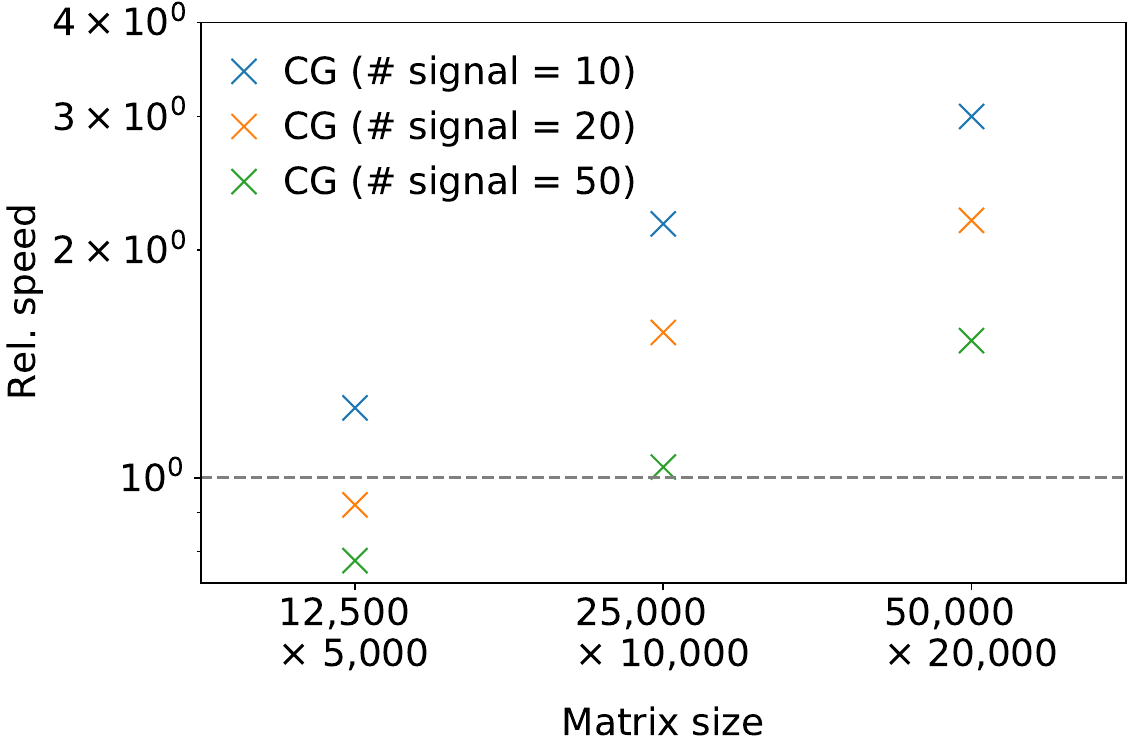}
\caption{%
	Relative speed of the two Gibbs samplers i.e.\ ratio of posterior computation time required by direct Gibbs to that required by CG-accelerated Gibbs.
	Values above 1 indicate superior performance of CG-accelerated Gibbs.
}
\label{fig:rel_speed_vs_mat_size_single}
\end{figure}

\newpage
\subsection{Choice of shrinkage prior and its effect on CG performance}
\label{supp:shrinkage_prior_choice_and_cg_performance}

The choice of shrinkage prior can significantly affect Bayesian sparse logistic regression's practical performance \citep{wei2020shrinkage_in_logistic_reg}. 
Ideally, we would like to achieve both statistical and computational efficiency.
It is of interest, therefore, to investigate how the CG-accelerated Gibbs sampler performance might depend on different shrinkage priors.

We run the CG-accelerated Gibbs sampler for 5,500 iterations on the synthetic data sets generated in the same manner as in Section~\ref{sec:cg_sampler_demonstration}, but this time varying $\alpha \in \{1/2, 1/4, 1/8\}$.
A smaller $\alpha$ corresponds to a larger peak at $\beta_j / \tau = 0$, heavier tail as $|\beta_j / \tau| \to \infty$, and generally superior statistical performance \citep{nishimura2019regularized_shrinkage}.
For each simulation set-up, we use 16 replicate data sets generated from different random seeds in order to assess how the variability in data sets might affect the posteriors and in turn the CG sampler's performances. 

The results are summarized in Table~\ref{tab:computing_time_across_varying_shrinkage}.\footnote{%
	As in Section~\ref{sec:rel_computational_efficiency_on_synthetic_data}, the simulation study here is extremely computationally intensive. 
	We thus deployed the same procedure as described in the footnote of Section~\ref{sec:rel_computational_efficiency_on_synthetic_data} to calculate the computational time we would have obtained had we run the Gibbs samplers on the original 2015 iMac.
	The actual computations using the Joint High Performance Computing Exchange required about 2 weeks and 864\gb{} of memory.
}
We find that, when fixing the number of signals and the overall synthetic data structure, the CG-accelerated sampler runs faster with smaller $\alpha$. 
Examination of the posterior structures reveals that this result here is consistent with our main finding throughout the paper --- the sparser the regression coefficient posterior, the faster the preconditioned CG's convergence rate.
Table~\ref{tab:sparsity_across_varying_shrinkage} for example shows that the coefficient estimates under a smaller $\alpha$ are sparser in terms of the posterior mean magnitudes. 
Table~\ref{tab:coverage_across_varying_shrinkage} additionally shows that a posterior under a smaller $\alpha$ in fact achieves better coverage of the true signals while shrinking the rest more strongly towards zero. 
The simulation study of \cite{nishimura2019regularized_shrinkage} finds a similar relation between $\alpha$ and posterior sparsity structure.

\begin{table}
\centering
	\begin{tabular}{ccccc}
	\hline
	& \multicolumn{4}{c}{Computing time (hours)}\\
	& \multicolumn{3}{c}{CG-accelerated Gibbs} & \hspace{1em}\multirow{2}{*}[-1ex]{\shortstack{Direct Gibbs\\(Fixed cost)}} \hspace{1em}\\
	\cmidrule(lr){2-4}
	& $\alpha = 1/2$ & $\alpha = 1/4$ & $\alpha = 1/8$ \\
	\hline
	\# of signals
	&  &  &  \\
	10
	& $38.4 \pm 1.45$ & $22.5 \pm 1.26$ &  $22.3 \pm 0.50$ & \\
	20
	& $50.5 \pm 2.87$ & $46.3 \pm 2.45$ & $36.6 \pm 1.91$ & 77.1 \\
	50
	& $73.1 \pm 2.88$ & $71.8 \pm 3.19$ & $71.4 \pm 2.14$ \\
	\hline
	\end{tabular}\vspace*{.1\baselineskip}
	\caption{%
		Computing times in hours for 5,500 iterations of the CG-accelerated Gibbs sampler on the synthetic data sets (left three columns) as the number of signals and the bridge prior's exponent $\alpha$ are varied.
		The values after the $\pm$ signs show variability in computing times across the 16 replicates, summarized as 1.96 times the standard deviation. 
		The computing time for the direct Gibbs sampler is shown on the right.
		The direct sampler's cost depends only on the design matrix size and is thus unaffected by the randomness in synthetic data or by the type of shrinkage prior. 
	}
	\label{tab:computing_time_across_varying_shrinkage}
\end{table}


\begin{table}
	\centering
	\begin{tabular}{ccccc}
	\hline
	& \multicolumn{3}{c}{
		\multirow{3}{*}[-1ex]{\shortstack{Number of coefficients with\\posterior mean magnitude\\ $> 0.1$ (and $> 0.01$)}}
	} \\ \\ \\
	\cmidrule(lr){2-4} 
	& $\alpha = 1/2$ & $\alpha = 1/4$ & $\alpha = 1/8$ \\
	\hline
	\# of signals & & & \\
	10 
	& $10 \pm 0$ \,($185 \pm 89$) & $10 \pm 0$ \,($10.4 \pm 1.0$) & $10 \pm 0$ \,($10.4 \pm 1.0$) \\
	20 
	& $21.8 \pm 3.4$ \,($1702 \pm 483$) & $20.1 \pm 0.5$ \,($59.1 \pm 35.4$) & $20 \pm 0$ \,($20.8 \pm 1.3$) \\
	50 
	& $72.8 \pm 15.3$ \,($4411 \pm 356$) & $58.2 \pm 7.4$ \,($1267 \pm 270$) & $51.8 \pm 2.7$ \,($252 \pm 67$) \\
	\hline
	\end{tabular}\vspace*{.1\baselineskip}
	\caption{%
		Numbers of coefficients whose estimated posterior means have magnitude above 0.1 (and above 0.01) as the number of signals and the bridge prior's exponent $\alpha$ are varied.
		In other words, the table shows the numbers of $j = 1, \ldots, 10{,}000$ for which $\big| \mathbb{E}[\beta_j \given \y, \X] \big| > 0.1$ (and $> 0.01$). 
		The values after the $\pm$ signs show variability across the 16 replicates, summarized as 1.96 times the standard deviation. 
	}
	\label{tab:sparsity_across_varying_shrinkage}
\end{table}

\begin{table}
	\centering
	\begin{tabular}{ccccc}
	\hline
	& \multicolumn{3}{c}{
		Coverage of true signals at 95\% level
	} \\ 
	\cmidrule(lr){2-4} 
	& $\alpha = 1/2$ & $\alpha = 1/4$ & $\alpha = 1/8$ \\
	\hline
	\# of signals & & & \\
	10 
	& $9.06 \pm 1.76$ & $9.94 \pm 0.47$ & $10 \pm 0$ \\
	20 
	& $17.3 \pm 3.2$ & $19.1 \pm 1.8$ & $19.9 \pm 0.7$  \\
	50 
	& $43.1 \pm 4.2$ & $46.4 \pm 3.3$ & $47.8 \pm 2.6$ \\
	\hline
	\end{tabular}\vspace*{.1\baselineskip}
	\caption{%
		Coverage of the true signal value $\beta_j = 1$ by 95\% posterior credible intervals as the number of signals and the bridge prior's exponent $\alpha$ are varied.
		The values after the $\pm$ signs show variability across the 16 replicates, summarized as 1.96 times the standard deviation. 
	}
	\label{tab:coverage_across_varying_shrinkage}
\end{table}

\FloatBarrier

\section{Optimizing linear algebra computations for Gibbs samplers}
\label{sec:optimizing_linear_algebra}

In the subsections to follow, the computation times are measured on a 2015 iMac with Intel Core i7 processor as in the main manuscript.
We first focus on a single-threaded implementation before exploring parallelization opportunities in Section~\ref{sec:opportunities_for_parallelization}.
Unless otherwise stated, all the benchmarks use the $72{,}489 \times 22{,}175$ sparse design matrix $\X$ in the application of Section~\ref{sec:application}.

\subsection{Dense vs.\ sparse numerical linear algebra}
\label{sec:dense_vs_linear_algebra}
When a design matrix $\X$ is sparse as in the application of Section~\ref{sec:application}, one may wonder if the precision matrix $\bPhi$ can be factorized efficiently using sparse numerical linear algebra techniques. This is not the case in typical sparse regression applications, however, for the following reasons. First, the matrix $\X^\transpose \bOmega \X$ and hence $\bPhi$ is typically much denser than $\X$ itself, especially when $n > p$. In particular, the $(j, j')$-th element of  $\X^\transpose \bOmega \X$ is non-zero if the $j$-th and $j'$-th predictors co-occur in any of the $n$ samples. Secondly, when employing sparse methods, time spent on irregular data access completely dominate over that on arithmetic operations \citepsupplement{duff2017direct-sparse}. In the absence of sufficient sparsity, therefore, it can be more computationally efficient to ignore the sparse structure and employ dense methods.

In the application of Section~\ref{sec:application}, we find the precision matrix $\bPhi$ to be $85.4$\% dense. The Cholesky factor of a sparse matrix is typically denser than the matrix itself (see Theorem 2.8 in \citealtsupplement{rue2005gmrf}); indeed, we find the Cholesky factor $\bm{L}$ to be over $99$\% dense. Sparse methods have no advantage whatsoever for such a dense matrix.


Conceivably, we can avoid dealing with the near-dense $\bPhi$ matrix in \eqref{eq:linear_system_for_cg_sampler} as follows. Noting that $\bPhi = \tilX^\transpose \tilX$ for $\tilX^\transpose = [\X^\transpose \bOmega^{1/2}, \tau^{-1} \bLshrink^{-1}]$, we see that a ($Q$-less) QR decomposition of the sparse matrix $\tilX$ would provide the Cholesky factor of $\bPhi$. For the application of Section~\ref{sec:application}, we experimented with this idea using the popular SparseSuites package \citepsupplement{davis2011sparseQR}. Sparse QR decomposition first attempts to find a permutation of the matrix columns to reduce the subsequent computation as much as possible. The package implements approximate minimum degree (\textsc{amd}), column \textsc{amd} (\textsc{colamd}), as well as graph-partitioning-based nested dissection (\textsc{metis}) algorithm. None of these algorithms find column orderings that differ significantly from the original arbitrary ordering with only handful of columns permuted. The subsequent factorization of the matrix requires about 25 minutes regardless of the column permutation algorithm chosen. On the other hand, explicitly computing $\bPhi$ and finding its Cholesky factor via dense linear algebra requires about 3 minutes only (Section~\ref{sec:linalg_library_choice}).

Incidentally, we can extend the observation $\bPhi = \tilX^\transpose \tilX$ for $\tilX^\transpose = [\X^\transpose \bOmega^{1/2}, \tau^{-1} \bLshrink^{-1}]$ and realize that solving the linear system \eqref{eq:linear_system_for_cg_sampler}, $\bPhi \bbeta = \bm{b}\,$ for $\,\bm{b} = \X^\transpose \bOmega \tilde{\y} + \X^\transpose \bm{\Omega}^{1/2} \bm{\eta} + \tau^{-1} \bLshrink^{-1} \bm{\delta},$ is equivalent to solving a least-square problem.
More explicitly, we can express \eqref{eq:linear_system_for_cg_sampler} as
\begin{equation}
\label{eq:normal_equation_formulation}
\tilX^\transpose \tilX \, \bbeta 
	= \tilX^\transpose \bm{\tilde{b}}
	\ \text{ where } \
	\bm{\tilde{b}} = 
	\begin{bmatrix}
	\bOmega^{1/2} \tilde{\y} + \bm{\eta} \\
	\bm{\delta}
	\end{bmatrix}.
\end{equation}
An equation of the form \eqref{eq:normal_equation_formulation} is known as a normal equation in linear algebra and its solution coincides with that of the minimizer of $\big\| \tilX \bbeta - \bm{\tilde{b}}  \big\|^2$.
It follows that, to draw a Gaussian vector using the algorithm of Proposition~\ref{prop:gaussian_as_linear_system_solution}, we could apply the LSQR algorithm of \cite{paige1982lsqr} to minimize $\big\| \tilX \bbeta - \bm{\tilde{b}}  \big\|^2$ instead of applying CG to solve \eqref{eq:linear_system_for_cg_sampler}.
While LSQR and CG are mathematically equivalent when applied to normal equations, LSQR may provide more numerically stable behavior when $\bPhi = \tilX^\transpose \tilX$ is ill-conditioned.

\subsection{Choice of linear algebra library}
\label{sec:linalg_library_choice}
For linear algebra operations involving large matrices, hardware-specific compilation and optimization are essential for achieving good computational efficiency. Major linear algebra libraries all achieve reasonable computational efficiency, but some variation in performance may occur depending on types of operations and computing environments. We therefore compare a few options for implementing the direct and CG-accelerated Gibbs sampler in Section~\ref{sec:application}. Efficiency of sparse matrix operations also depends critically on the underlying representations of sparse matrices \citepsupplement{saad2011methods-for-eigenvalue}. For each linear algebra library, we try all the major sparse matrix formats and report only the result with best performance. All the benchmarks here are run using the design matrix of Section~\ref{sec:application}.

The computational bottleneck of the direct Gibbs sampler is computing $\X^\transpose \bOmega \X$ and the subsequent Cholesky factorization of $\bPhi$. Since $\bOmega$ is a diagonal matrix, computing $\X^\transpose \bOmega \X$ can be carried out as multiplying a sparse matrix $\X^\transpose \bOmega^{1/2}$ with its transpose. The high-performance computing community refers to the operation of multiplying two sparse matrices by the acronym \spgemm{} (generalized sparse matrix-matrix multiplication). \spgemm[\MakeUppercase] is a surprisingly complicated operation to optimize in modern computing architectures.  No specification is provided for such an operation by sparse \blas{} \citepsupplement{duff2002spblas}, and its hardware-specific implementation is an active area of research \citepsupplement{matam2012spgemm, azad2016spgemm}.

The Scipy library provides an \spgemm{} implementation via the algorithm of \citesupplement{bank1993scipy-sparse-algorithms}. The Scipy \spgemm{} operation requires 144 seconds. We find an alternative implementation in the Intel \mkl{} library with the option of returning a dense (instead of sparse) matrix, which is faster here as the multiplied matrix is almost completely dense. The \mkl{} \spgemm{} requires 64.6 seconds, being a clear winner and our choice for the simulation results of Section~\ref{sec:application}.

For the (dense linear algebra) Cholesky factorization of $\bPhi$, the Scipy library by default calls the \mkl{} library. The computation requires 78.0 seconds.  The \openblas{} implementation performs comparably, requiring 79.0 seconds.

The computation time of the CG-accelerated Gibbs is dominated by the (sparse) matrix-vector multiplications $\bv \to \X \bv$ and $\bw \to \X^\transpose \bw$ required for CG iterations. For these operations, the Scipy library uses its own C-extension code. On average, the operation $\bv \to \X \bv$ requires $6.70 \times 10^{-2}$ seconds and the operation $\bw \to \X^\transpose \bw$ $6.07 \times 10^{-2}$ seconds. With the \mkl{} library implementations, these matrix-vector multiplications on average require $4.69 \times 10^{-2}$ and $5.49 \times 10^{-2}$ seconds respectively.

\subsection{Number of arithmetic operations v.s.\ actual computing time}
\label{sec:flop_vs_actual_time}
Here we elaborate on why we cannot compare the relative performance of the direct and CG-accelerated Gibbs sampler by simply counting the number of required floating point operations (``flop'' for short).
By way of empirical demonstration, for each linear algebra operation we compare the actual computation times to the numbers of arithmetic operations.
A detailed discussion of how data movement --- and not arithmetic operations --- creates a bottleneck in modern hardware is beyond the scope of this section, and we refer the readers to existing references such as \cite{guntheroth2016optimizedCpp} and \cite{holbrook2020massive_parallelization}.

\subsubsection*{CG vs.~direct linear algebra}
In the discussion to follow, we count the number of flops up to the leading order terms, ignoring contributions that are essentially negligible for any moderately-sized design matrices.
The sparse design matrix $\X$ of Section~\ref{sec:application} has $\nNonzero = 6.43 \times 10^7$ non-zero elements.
The matrix-vector operations $\bv \to \X \bv$ and $\bw \to \X^\transpose \bw$ both require $\nNonzero$ multiplications and additions, for the total of $\nMatvec = 1.28 \times 10^8$ flops.
The two matrix-vector operations result in $\nCg = 2.56 \times 10^8$ flops for each CG iteration.
Incidentally, as we have already seen in Section~\ref{sec:linalg_library_choice}, $\bv \to \X \bv$ and $\bw \to \X^\transpose \bw$ take different amounts of time despite requiring the same number of flops.

The total number of flops in multiplying $\X^\transpose \bOmega^{1/2}$ with its transpose, which coincides with that for multiplying $\X^\transpose$ with its transpose, is proportional to the sum of the overlaps between the pairs of the columns of $\X$:
\begin{equation*}
\nOverlap
	= \sum_{k, \ell = 1}^p \sum_{i = 1}^n \ind \left\{ x_{i k}  \neq 0, \ x_{i \ell}  \neq 0 \right\}.
\end{equation*}
We have $\nOverlap = 6.44 \times 10^{10}$ for the sparse design matrix $\X$ of Section~\ref{sec:application}.
Computing $\X^\transpose \bOmega \X$ (or $\X^\transpose \X$) requires $\nOverlap$ multiplications and additions, resulting in $\nMatmat = 1.29 \times 10^{11}$.
The subsequent Cholesky decomposition of $\bPhi = \X^\transpose \bOmega \X + \tau^{-2} \bLshrink^{-2}$ requires $\nCholesky = p^3 / 3 = 3.63 \times 10^{12}$ flops \citep{trefethen1997numerical_linalg}.

Table~\ref{tab:flop_vs_time} summarizes the preceding analysis and contrasts the numbers of flops to the actual computation times as measured in Section~\ref{sec:linalg_library_choice}.
We can clearly see that the number of flops does not directly correlate with computation time.
Note in particular that, despite requiring 28 times more flops, the Cholesky decomposition of the dense precision matrix $\bPhi$ takes only a little more time than the matrix-matrix multiplication for $\X^\transpose \bOmega \X$.
This is because dense matrix operations can take advantage of streamlined and highly efficient access to data stored contiguously in memory \citep{dongarra2016hpc_connjugate_gradient}.

\begin{table}
\centering
	\begin{tabular}{cccc}
	& CG iteration & Matrix-matrix & Cholesky \\
	\hline
	Number of flops
	& $2.56 \times 10^8$ & $1.29 \times 10^{11}$ & $3.63 \times 10^{12}$  \\
	\hline
	Computation time (sec)
	& $0.102$ & $64.6$ & $78.0$ \\
	\hline
	Relative number of flops
	& 1 & 504 & 14,200 \\
	\hline
	Relative computation time
	& 1 & 633 & 765
	\end{tabular}
	\caption{%
		Numbers of flops and actual computation times for the following operations:
		one CG iteration (whose cost is dominated by matrix-vector multiplications $\bv \to \X \bv$ and $\bw \to  \X^\transpose \bw$), matrix-matrix multiplication for computing $\X^\transpose \bOmega \X$, and Cholesky decomposition of $\bPhi = \X^\transpose \bOmega \X + \tau^{-2} \bLshrink^{-2}$.
		The benchmarks are based on the $72{,}489 \times 22{,}175$ sparse design matrix $\X$ of Section~\ref{sec:application}.
	}
	\label{tab:flop_vs_time}
\end{table}

\subsubsection*{Effect of data access efficiency on computational speed}
We carry out a few more experiments to further illustrate the effects of data access efficiency on computation time.
To this end, we first create a synthetic design matrix $\Xsynthetic$ of size $n = 14{,}499$ and $p = 4{,}435$.
We store this design matrix as a dense array, so that the total number of entries $np$ roughly equals that of non-zero entries $\nNonzero = 6.43 \times 10^7$ in the $72{,}489 \times 22{,}175$ sparse design matrix $\X$.
In particular, the matrix-vector multiplications by the two design matrices require the same number of flops.
The actual computation times differ markedly, however: $2.04 \times 10^{-2}$ seconds for the synthetic matrix and $4.67 \times 10^{-2}$ seconds for the real-data one.

The effects of data access efficiency are not limited to sparse vs.\ dense matrices; the effects manifest themselves also within dense linear algebra.
As an illustration, we create a synthetic design matrix $\Xsynthetic$ of the same dimension ($n = 72{,}489$ and $p = 22{,}175$) as the real-data one but stored as a dense array instead.
The matrix-vector operation $\bv \to \Xsynthetic \bv$ requires $2np$ flops, while the matrix-matrix operation $\Xsynthetic^\transpose \Xsynthetic$ requires $2np^2$ flops.
In particular, the matrix-matrix operation requires $p = 22{,}175$ times more flops than the matrix-vector one.
In terms of the actual computation time, however, the matrix-vector operation requires $0.571$ seconds while the matrix-matrix one requires $592$ seconds, or only $1{,}037$ times longer in duration.
The matrix-matrix operation consumes less time than otherwise expected from the number of flops because it re-uses the same pieces of data in cache many times, rather than fetching them from main memory every time \citep{golub2012matrix}.
Moreover, such efficient data movement enables potential use of vector processing to add multiple floating point numbers at the same time;
while most modern \cpu{}'s have vector processing capability, bottlenecks in data movement often prevent algorithms from exploiting it \citep{holbrook2020massive_parallelization}. 
Incidentally, storing $\Xsynthetic$ as a dense array requires 12.9\gb{} of memory, so results from this particular benchmark will depend strongly on specific hardware and amounts of available \textsc{ram}.

\subsection{Opportunities for parallelization within memory constraints}
\label{sec:opportunities_for_parallelization}
The above comparisons of computational efficiency are based on a single-threaded \cpu{} computing environment. Computational gain from parallelization is highly architecture dependent for large-scale problems and is difficult to draw any general conclusions \citepsupplement{dongarra2016hpc_connjugate_gradient, duff2017direct-sparse}. Nonetheless, here we provide a qualitative discussion of to what extent each algorithm can benefit from parallelization. We complement the discussion with illustrative quantitative results, obtained by using all the four cores of Intel i7 \cpu{} on 2015 iMac.

Before any discussion of computational gains from parallelization, we emphasize the following point regarding the two alternative Gibbs samplers for Bayesian sparse regression: as the problem size grows, memory constraints make the CG-accelerated sampler \emph{the only option} in a typical computing environment. We can run the CG-accelerated sampler as long as we have enough memory to store the (sparse) design matrix $\X$. On the other hand, as we have discussed, the direct Gibbs sampler generally cannot avoid having to store the near-dense precision matrix $\bPhi$. In case of the sparse design matrix $\X$ of size $72{,}489 \times 22{,}175$ in Section~\ref{sec:application}, for example, storing $\X$ in the compressed sparse row format only requires 0.719\gb{} of memory while storing the dense $\bPhi$ requires 3.67\gb{} of memory. In fact, further memory burden is incurred by temporary allocation of extra memory necessary for the Cholesky factorization of $\bPhi$. Profiling the memory usage by the \mkl{} library reveal that temporary memory allocation of 3.75\gb{}, requiring at least 8.14 ($= 0.719 + 3.67 + 3.75$) \gb{} of memory for running the direct Gibbs sampler.

Dense linear algebra operations benefit most from parallelization when using a typical modern hardware, which performs best at accessing data stored contiguously in memory \citepsupplement{dongarra2016hpc_connjugate_gradient}. In fact, computation time for the Cholesky factorization goes down from 78.0 to 22.7 seconds when using the four cores with the \mkl{} library. It is also worth noting that the speed-up is significantly smaller when using the \openblas{} library; the time goes down from 79.0 only to 33.1 seconds, illustrating the importance of hardware-specific optimizations in parallel computing.

Parallelizing sparse linear algebra operations are more complex due to their bottleneck being irregular data access \citepsupplement{duff2017direct-sparse}. Speed-ups thus tend to be smaller, though there are growing efforts in building hardwares optimized for sparse operations \citepsupplement{dongarra2016hpc_connjugate_gradient}. The \mkl{} \spgemm{} delivers a modest speed-up when using the four cores, cutting the time from 64.6 to 46.8 seconds. The sparse matrix-vector multiplications $\bv \to \X \bv$ and $\bw \to \X^\transpose \bw$ benefit slightly more from parallelization. The \mkl{} library implementations reduces the time from $4.69 \times 10^{-2}$ to $3.13 \times 10^{-2}$ seconds and from $5.49 \times 10^{-2}$ to $3.14 \times 10^{-2}$ seconds respectively.


\section{In-depth look at mechanism of CG-acceleration in Section~\ref{sec:application}}
\label{supp:cg_acceleration_mechanism_details}

\subsection{Accuracy of CG sampler}
\label{sec:cg_sampler_accuracy_real_data}
	We assess accuracy of the CG sampler in the real data setting of Section~\ref{sec:application} by applying the same diagnostics as in Section~\ref{sec:cg_sampler_accuracy}.
	We confirm again that the distribution of the CG sampler output is exact for practical purposes.

	When comparing the outputs of the direct and CG-accelerated Gibbs samplers, the estimated posterior means of regression coefficients closely aligns with each other (Figure~\ref{fig:cg_sampler_accuracy_diagnostic_plots}(a)).
	As the more formal statistical test of difference in the two estimators, Figure~\ref{fig:cg_sampler_accuracy_diagnostic_plots}(b) shows that the distribution of the standardized differences closely follows the ``null'' distribution.
	Figure~\ref{fig:cg_sampler_accuracy_diagnostic_plots}(a) and \ref{fig:cg_sampler_accuracy_diagnostic_plots}(b) are based on the posterior samples for the propensity score model, but we obtained essentially the same result under the treatment effect model.
	We additionally performed the same diagnostic on the estimators of the posterior second moment of $\bbeta$ and obtained similar results.

  \begin{figure}
  	\begin{minipage}{.4\linewidth}
  		\subfigure[
  			Comparison of the regression \newline coefficient estimates (posterior means) between those based on the direct and CG-accelerated Gibbs samplers.
  		]{
  			\includegraphics[width=\linewidth]{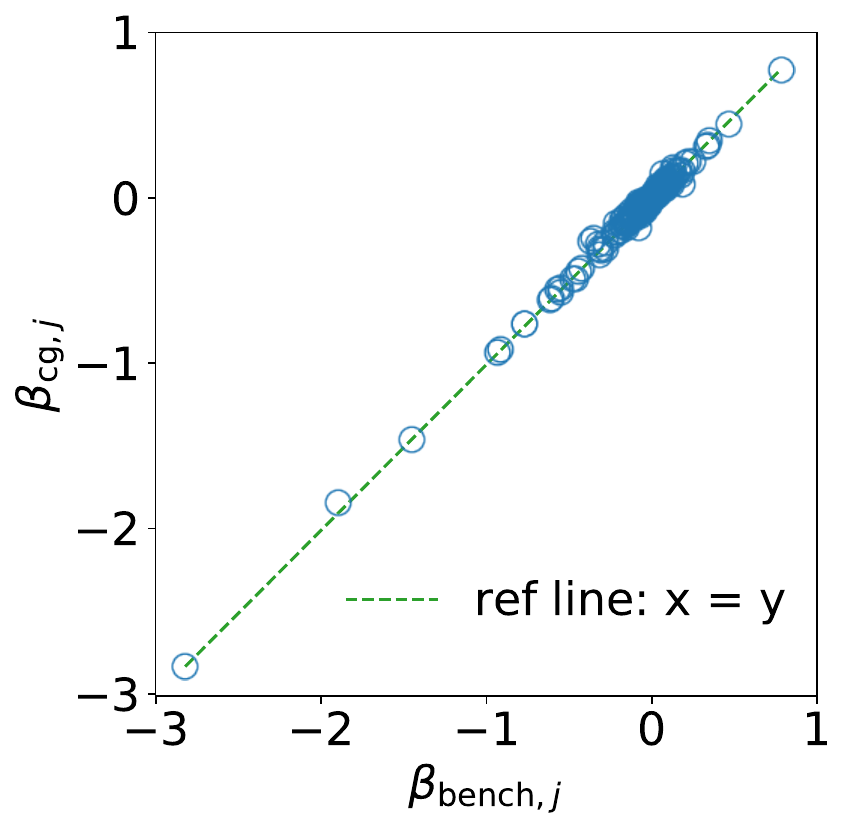}
  		}
  	\end{minipage}
  	\hspace{2ex}
  	\begin{minipage}{.55\linewidth}
  		\subfigure[
  			Normalized histogram for the standardized differences $(\hat{\beta}_{\text{bench}, j} - \hat{\beta}_{\text{cg}, j}) / \hat{\sigma}_i$, where $\hat{\sigma}_i^2$ is an estimate of the Monte Carlo variance of $\hat{\beta}_{\text{bench}, j} - \hat{\beta}_{\text{cg}, j}$. Gaussianity of the histogram indicates no statistically significant difference between the two MCMC outputs.
  		]{
  			\includegraphics[width=.95\linewidth]{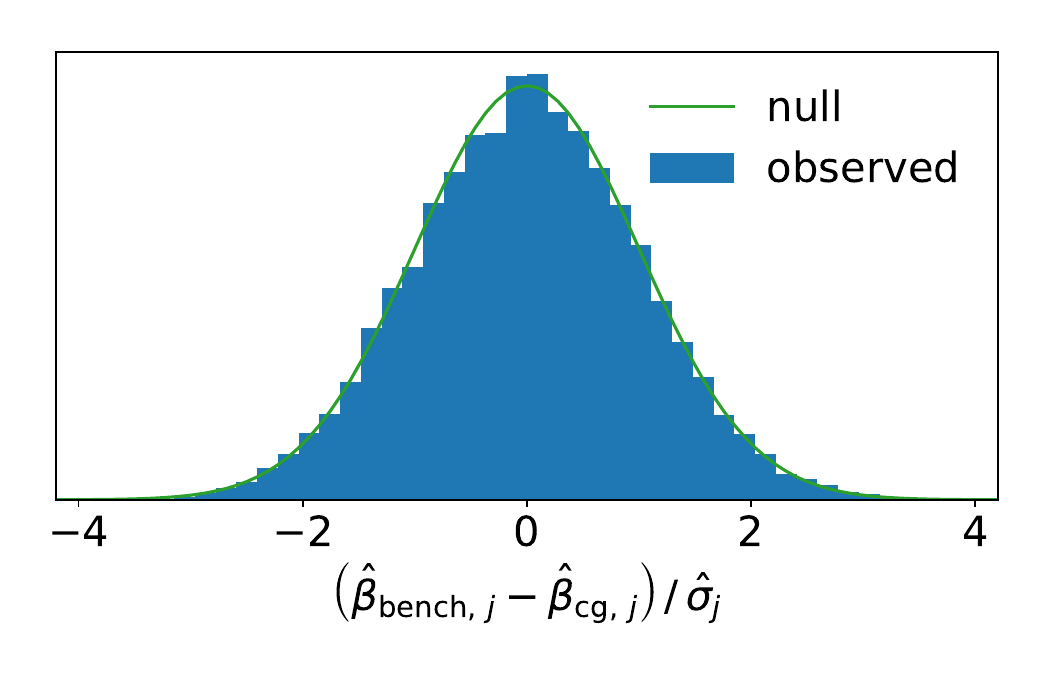}
  		}
  	\end{minipage}
  	\caption{Diagnostic plots to check for statistically significant differences between the two MCMC outputs.}
  	\label{fig:cg_sampler_accuracy_diagnostic_plots}
	\end{figure}

\subsection{Number of CG iterations at each Gibbs step}
\label{sec:n_cg_iter_within_gibbs}
	To study computational cost of the CG sampler at each Gibbs iteration, we first focus on a post-burn-in update of $\bbeta$ for the propensity score model.
	As in Section~\ref{sec:cg_sampler_demonstration}, we compare the CG iterates $\bbeta_k$ against the exact solution $\bbeta_{\textrm{direct}}$ of the linear system \eqref{eq:linear_system_for_cg_sampler} found via the Cholesky-based direct method. Figure~\ref{fig:treatment_model_cg_error_plot} plots the distances between $\bbeta_k$ and $\bbeta_{\textrm{direct}}$ as a function of $k$, the number of CG iterations or equivalently of matrix vector multiplications $\bv \to \bPhi \bv$.

	The solid blue line tracks the root mean squared residual $p^{-1/2} \| \bm{\tilde{r}}_k \|_2$ as introduced in Section~\ref{supp:termination_criteria_for_cg_sampler}. The dotted vertical line indicates when the magnitude of the prior-preconditioned residual $\bm{\tilde{r}}_k$ falls below the termination criterion of \eqref{eq:mean_sq_resid_for_cg_termination}. The termination occurs at $k = 133$ and the CG sampler consequently spends less than $10\%$ of the computational time relative to the direct sampler.
	Note also how the solid blue line upper-bound the dashed one which tracks the following error metric computed as a proxy for \eqref{eq:cg_sampler_error_metric}:
	\begin{equation}
	\label{eq:mean_secmom_normalized_error}
	\left\{
		p^{-1} \textstyle{\sum}_j \, \hat{\xi}_j^{-2} (\bbeta_k - \bbeta_{\textrm{direct}})_j^2
	\right\}^{1/2}
		\ \text{ with } \ \hat{\xi}_j^{2} \approx \mathbb{E} [ \beta_j^2 \given \y, \X].
	\end{equation}
	This empirical result provides a further support to our theoretical analysis in Section~\ref{supp:termination_criteria_for_cg_sampler} and hence to the use of $p^{-1/2} \| \bm{\tilde{r}}_k \|_2$ in the termination criterion.
	As confirmed in Section~\ref{sec:cg_sampler_accuracy}, the CG error at termination is so small that it does not affect the stationary distribution of the Gibbs sampler in any statistically significant way.

	Figure~\ref{fig:treatment_model_cg_error_plot} also makes it clear that the advantage of the prior preconditioner, as demonstrated in the simulated data examples of Section~\ref{sec:cg_sampler_demonstration}, continues to hold in this real data example.  The observed convergence behaviors under the two preconditioners are again well explained by the eigenvalue distributions of the respective preconditioned matrices (Figure~\ref{fig:eigval_distributions_for_treatment_model});
	prior preconditioning leads to a tighter cluster of the eigenvalues and avoids introducing small eigenvalues.

	\begin{figure}
	\centering
		\includegraphics[width=.6\linewidth]{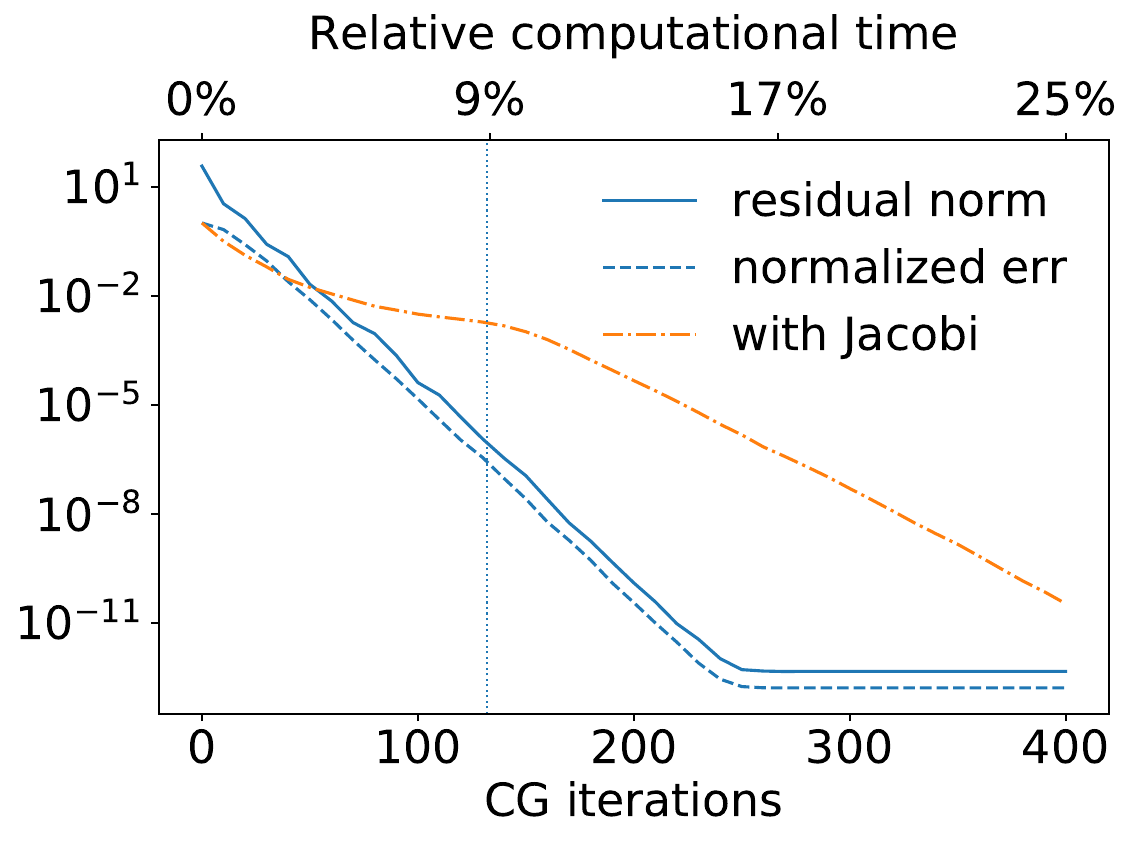}
		\vspace*{-.5\baselineskip}
		\caption{Plot of the CG errors during a conditional update of $\bbeta$ within the propensity score model posterior computation. The errors are plotted as a function of the number of CG iterations (bottom axis) and of the computational time relative to the direct linear algebra (top axis). The solid line shows the root mean squared residual $p^{-1/2} \| \bm{\tilde{r}}_k \|_2$ used in the stopping criteria \eqref{eq:mean_sq_resid_for_cg_termination}. The dotted vertical line indicates when $p^{-1/2} \| \bm{\tilde{r}}_k \|_2$ reaches the threshold $10^{-6}$. The 2nd-moment normalized errors \eqref{eq:mean_secmom_normalized_error} are shown as the blue dashed line for the prior preconditioner and as the orange dash-dot line for the Jacobi preconditioner.}
		\label{fig:treatment_model_cg_error_plot}
	\end{figure}


	\begin{figure}
		\includegraphics[width=\linewidth]{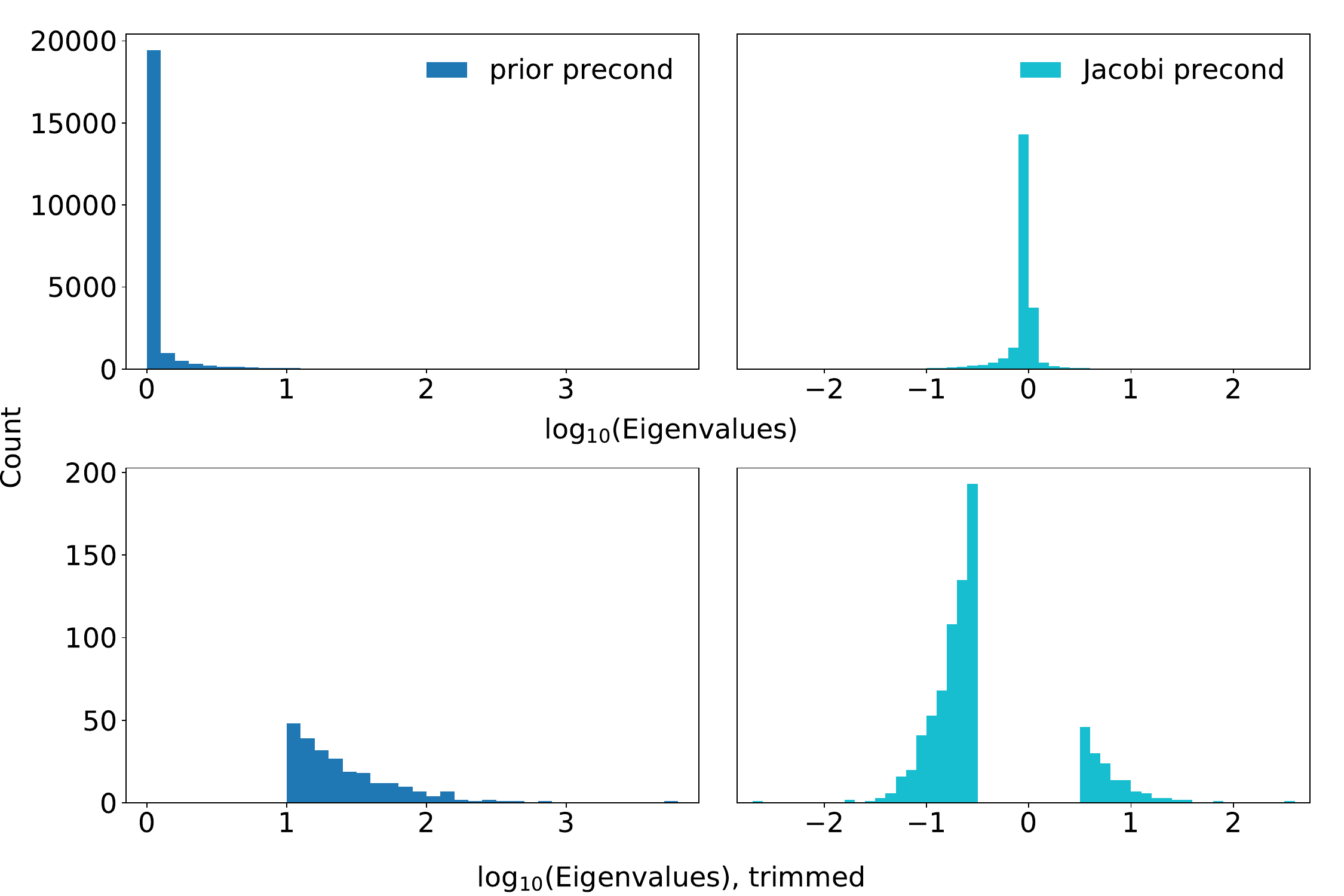}
		\caption{Histograms of the eigenvalues of the preconditioned matrices as in Figure~\ref{fig:eigval_distributions_for_data_simulated_with_10_correlated_signal}. The only differences here are that 1) the preconditioned matrices are based on a posterior sample from the propensity score model \eqref{eq:propensity_score_model} and 2) the trimmed version for the Jacobi preconditioner removes the eigenvalues in the range $[-0.5, 0.5]$ as this choice better demonstrates the tail behavior here.}
		\label{fig:eigval_distributions_for_treatment_model}
	\end{figure}

	As mentioned in Section~\ref{sec:theory_of_prior_preconditioning}, the number of required CG iterations is in practice random since the linear system \eqref{eq:linear_system_for_cg_sampler} depends on the quantities $\bomega$, $\tau$, $\blshrink$, and $\bb$, which vary from one Gibbs iteration to another.
	Even with substantial variation in these random quantities, however, we consistently observe fast decay in all $\tau \lshrink_{(k)}$ and rapid CG convergence at every iteration.
	To illustrate, Figure~\ref{fig:n_cg_iter_till_convergence} shows the number of required CG iterations at each iteration of the Gibbs sampler at stationarity; 95\% of the numbers falls in the range $[107, 120]$ in this propensity score model example.

	\begin{figure}
	\centering
	\includegraphics[width=0.7\linewidth]{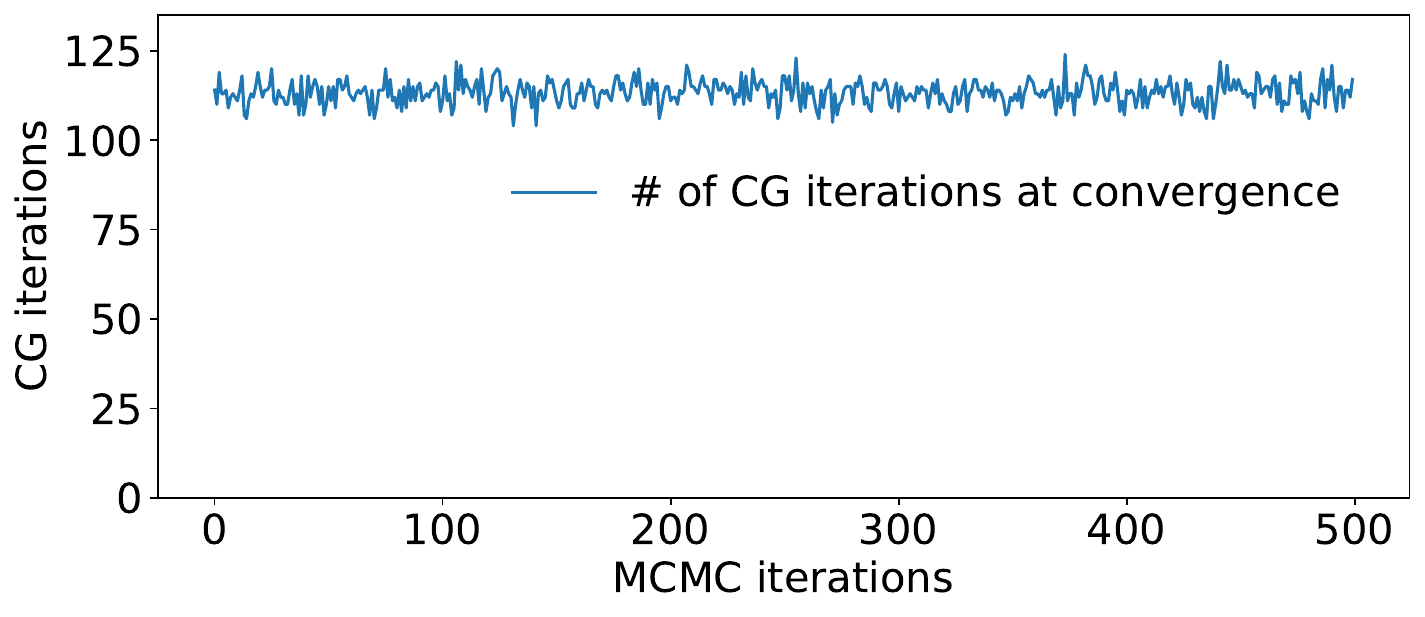}
	\caption{%
		Number of CG iterations required for the conditional updates of $\bbeta$ via the CG sampler.
		The first 10\% (500 iterations) of the post-burnin Gibbs sampler iterations is shown.
	}
	\label{fig:n_cg_iter_till_convergence}
	\end{figure}

    We have so far studied the mechanism of CG-acceleration by using the propensity score model example.
    For the treatment effect model, the Jacobi preconditioner turns out to be comparable the prior preconditioner because the precision matrix $\bPhi = \X^\transpose \bOmega \X + \tau^{-2} \bLshrink^{-2}$ is so strongly dominated by the diagonal terms.
    This phenomenon is explained by the following two facts.
    First, the posterior is extremely sparse (Section~\ref{sec:speed_up_from_cg_acceleration_on_ohdsi_data}) and correspondingly the majority of $\tau \lshrink$'s are also extremely small.
    This makes the diagonal prior shrinkage term $\tau^{-2} \bLshrink^{-2}$ far more significant than the non-diagonal likelihood term $\X^\transpose \bOmega \X$.
    Secondly, the entries of $\bOmega = \textrm{diag}(\bomega)$, which can be interpreted as the weight on or informativeness of the observation $\y$ (Equation~\ref{eq:pg_augmented_likelihood}), are typically quite small due to the low incidence rate (Section~\ref{sec:data_description}).
    This further reduces the contribution of the non-diagonal term $\X^\transpose \bOmega \X$ to the precision $\bPhi$.

	\begin{figure}
		\includegraphics[width=\linewidth]{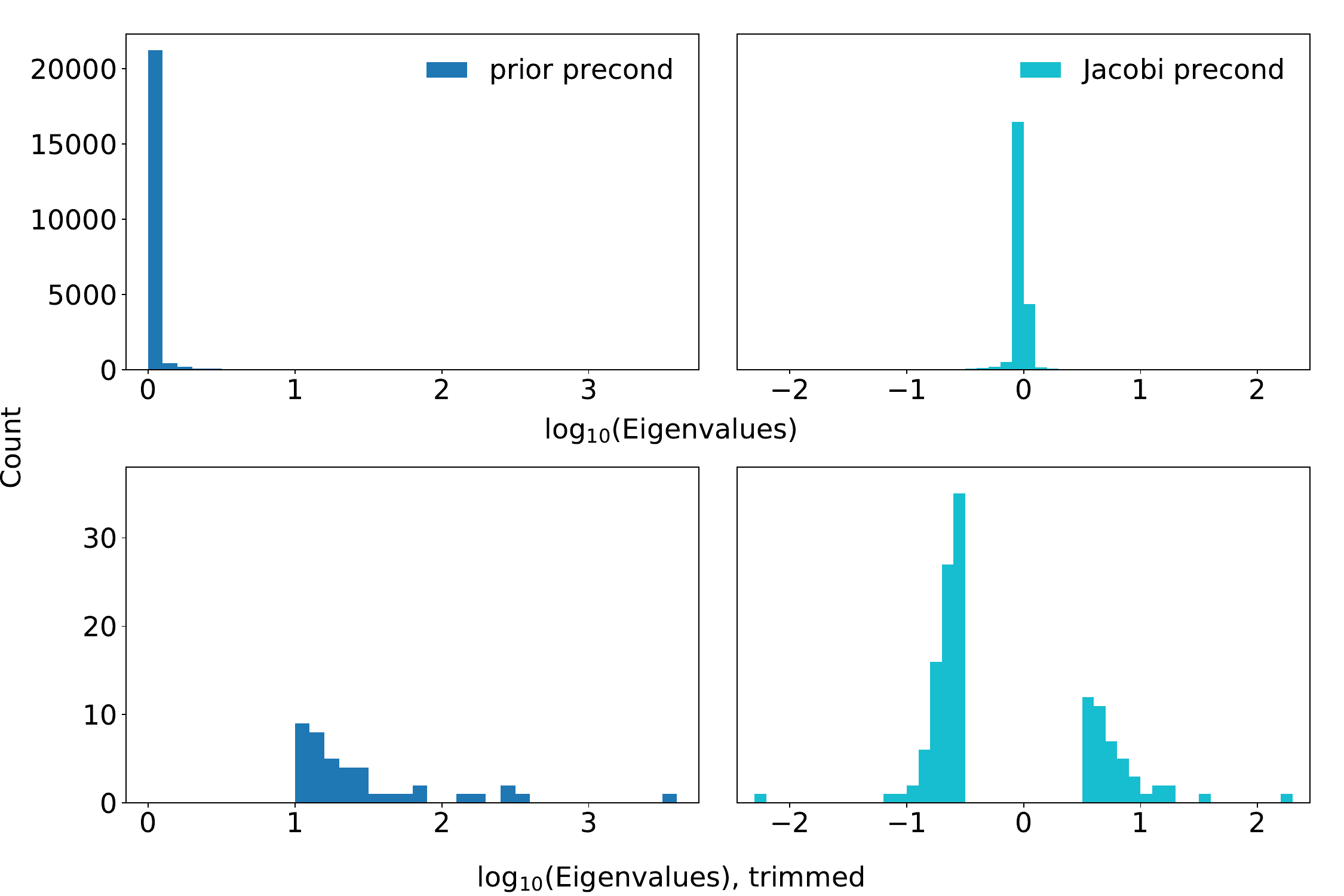}
		\caption{Histograms of the eigenvalues of the preconditioned matrices as in Figure~\ref{fig:eigval_distributions_for_data_simulated_with_10_correlated_signal}. The only differences here are that 1) the preconditioned matrices are based on a posterior sample from the treatment effect model \eqref{eq:treatment_effect_model} and 2) the trimmed version for the Jacobi preconditioner removes the eigenvalues in the range $[-0.5, 0.5]$ as this choice better demonstrates the tail behavior here.}
		\label{fig:eigval_distributions_for_outcome_model}
	\end{figure}
	
\FloatBarrier

\subsection{Choice of shrinkage prior and its effect on CG performance}
\label{supp:shrinkage_prior_choice_and_cg_performance_on_ohdsi_data}

Here we repeat the experiment of Section~\ref{supp:shrinkage_prior_choice_and_cg_performance} using the real data and assess how the CG-accelerated Gibbs sampler performance might depend on different shrinkage priors.

As in Section~\ref{supp:shrinkage_prior_choice_and_cg_performance}, we vary the bridge prior's exponent $\alpha$ among 1/2, 1/4, and 1/8.
In Section~\ref{sec:application}, we use $\alpha = 1/2$ for the propensity score model and, due to the low incident rate in the outcome of interest (Section~\ref{sec:prior_and_computation}), $\alpha = 1/4$ in the treatment effect model.
The treatment effect model additionally deploys a weakly informative Gamma prior on $\phi = \tau^{-\alpha}$ so that $\log_{10}(\tau)$ has the prior mean of $-1.5$ and standard deviation of $0.5$ (Section~\ref{sec:prior_and_computation}).
For fair comparison among different $\alpha$'s, we place the corresponding Gamma priors so that $\log_{10}(\tau)$ always have the prior mean of $-1.5$ and standard deviation of $0.5$.
As in Section~\ref{sec:application}, we run the Gibbs sampler for 5,500 iterations for the propensity score model and 11,000 iterations for the treatment effect model.

The results are summarized in Table~\ref{tab:computing_time_across_varying_shrinkage_on_ohdsi_data} and \ref{tab:sparsity_across_varying_shrinkage_on_ohdsi_data}.
Ostensibly, the CG-accelerated sampler performance here shows a trend opposite of the one observed in Section~\ref{supp:shrinkage_prior_choice_and_cg_performance}: 
the average number of required CG iterations, and hence overall computing time, increases as $\alpha$ becomes smaller (Table~\ref{tab:computing_time_across_varying_shrinkage_on_ohdsi_data}).
However, the posterior sparsity structure summarized in Table~\ref{tab:sparsity_across_varying_shrinkage_on_ohdsi_data} suggests that our main takeaway still holds here: the sparser the regression coefficient posterior, the faster the prior-preconditioned CG's convergence rate. 

The difference between the examples here and those of Section~\ref{supp:shrinkage_prior_choice_and_cg_performance} is how the posterior sparsity structure changes as $\alpha$ is varied.
In the synthetic data setting, there is a clear dichotomy between true signals $\beta_j = 1$ and non-signals $\beta_j = 0$. 
In the real-world setting, there is no such dichotomy.
In fact, as the sample size goes to infinity, all the coefficients will likely achieve non-zero values. 
It is just that, with the finite sample, some coefficients are detectable enough to be estimated away from zero while others are shrunk towards zero. 
This observation explains why we see less sparse posteriors under smaller $\alpha$'s here (Table~\ref{tab:sparsity_across_varying_shrinkage_on_ohdsi_data}) while we see an opposite behavior in Section~\ref{supp:shrinkage_prior_choice_and_cg_performance}. 
One caveat with the previous statement is that, for the propensity score model, our simple metrics do not adequately quantify the posterior sparsity structure;
for smaller $\alpha$'s, we see more coefficients with posterior mean magnitudes larger than 0.1, but fewer coefficients with posterior mean magnitudes larger than 0.01.
More generally, the relation between posterior sparsity structure and choice of $\alpha$ (or choice of other shrinkage priors) likely depends on the specific characteristics of a given data set; 
other real-world data sets may well yield posteriors that behave more like the synthetic data sets of Section~\ref{supp:shrinkage_prior_choice_and_cg_performance}.

\begin{table}
\centering
	\begin{tabular}{ccccc}
	\hline
	& \multicolumn{4}{c}{Computing time (hours)}\\
	& \multicolumn{3}{c}{CG-accelerated Gibbs} & \hspace{1em}\multirow{2}{*}[-1ex]{\shortstack{Direct Gibbs}} \hspace{1em}\\
	\cmidrule(lr){2-4}
	& $\alpha = 1/2$ & $\alpha = 1/4$ & $\alpha = 1/8$ \\
	\hline 
	Propensity score model & $11.4$ & $13.2$ & $15.6$ & $106$ \\
	Treatment effect model & $9.41$ & $11.3$ & $12.9$ & $212$ \\
	\hline
	\end{tabular}\vspace*{.1\baselineskip}
	\caption{%
		Computing times in hours for 5,500 iterations (propensity score model) and 11,000 iterations (treatment effect model) of the CG-accelerated Gibbs sampler on the blood anti-coagulant study data as the bridge prior's exponent $\alpha$ is varied.
	}
	\label{tab:computing_time_across_varying_shrinkage_on_ohdsi_data}
\end{table}

\newcommand{\lbracket}{[\hspace{.05em}}
\newcommand{\rbracket}{\hspace{.06em}]}
\begin{table}
	\centering
	\begin{tabular}{ccccc}
	\hline
	& \multicolumn{3}{c}{
		\multirow{3}{*}[-1ex]{\shortstack{Number of coefficients with\\posterior mean magnitude\\ $> 0.1$, \big[$> 0.1 / \sqrt{10}$\,\big], and ($> 0.01$)}}
	} \\ \\ \\
	\cmidrule(lr){2-4} 
	& $\alpha = 1/2$ & $\alpha = 1/4$ & $\alpha = 1/8$ \\
	\hline
	Propensity score model 
	& 82 \lbracket407\rbracket \,(3,716)  & 106 \lbracket416\rbracket \,(2,170) & 136 \lbracket432\rbracket \,(1,575) \\
	Treatment effect model
	& 0 \lbracket1\rbracket \,(18) & 2 \lbracket5\rbracket \,(85) & 3 \lbracket15\rbracket \,(156) \\
	\hline
	\end{tabular}\vspace*{.1\baselineskip}
	\caption{%
		Numbers of coefficients, among the 22,174 under shrinkage, whose estimated posterior means have magnitude above 0.1, \big[above $ 0.1 / \sqrt{10} \approx 0.0316$\,\big], and (above 0.01) as the bridge prior's exponent $\alpha$ is varied.
	}
	\label{tab:sparsity_across_varying_shrinkage_on_ohdsi_data}
\end{table}

\FloatBarrier

\section{Issue with approximating $\btilPhi$ by thresholding $\tau \lshrink_j$'s}
\label{supp:block_preconditioner}
	In Section~\ref{sec:prior_preconditioning}, we observed that the $(i, j)$-th entry of $\tau^{2} \bLshrink \X^\transpose \bOmega \X \bLshrink$ is small whenever $\tau \lshrink_i \approx 0$ or $\tau \lshrink_j \approx 0$. Given this observation, one may wonder if we can obtain a convenient low-rank approximation of the prior-preconditioned matrix $\btilPhi = \tau^{2} \bLshrink \X^\transpose \bOmega \X \bLshrink + \I_p$ by zeroing out $\tau \lshrink_j$'s at some threshold. This is \textit{not} the case in general as we will show now. Intuitively, the problem is as follows: while the ordered local scale parameter $\lshrink_{(k)}$ decays reasonably quickly as $k$ increases, there is no clear ``gap'' where $\lshrink_{(k + 1)} \ll \lshrink_{(k)}$.
	 For example, the histogram of a posterior draw of $\tau \lshrink_j$'s  in Figure~\ref{fig:prior_sd_hist_with_varying_number_of_signals} shows clearly that there is no such gap.

	We can assess the quality and utility of the thresholding approximation by using it as a preconditioner for CG in solving the system $\btilPhi \btilbeta = \bbTilde$. In other words, we consider using the thresholding approximation on top of prior-preconditioning. Let $(\tau^2 \bLshrink \X^\transpose \bOmega \X \bLshrink)_{(k)}$ denote a matrix obtained by thresholding the entries of $\tau^2 \bLshrink \X^\transpose \bOmega \X \bLshrink$ to zero except for the $k \times k$ block corresponding to the $k$ largest local scale parameters $\lshrink_{(1)}, \ldots, \lshrink_{(k)}$. Then the thresholding approximation
	\begin{equation}
	\label{eq:threshold_approximation_of_precond_matrix}
	\btilPhi_{(k)} =  (\tau^2 \bLshrink \X^\transpose \bOmega \X \bLshrink)_{(k)} + \I_p
	\end{equation}
	is the identity perturbed by the $k \times k$ block along the diagonal. As such, using it as a preconditioner requires the one time cost of computing and factorizing the $k \times k$ block, which requires $O(k^2 n + k^3)$ arithmetic operations. The quality of the approximation should improve as $k$ increases but so does the computational cost. In particular, at some point the $O(k^2 n + k^3)$ cost of preparing the thresholding preconditioner overwhelms the $O(np)$ cost of each CG iteration and becomes the computational bottleneck. When an adequate approximation requires such a large $k$, therefore, there is no benefit of using the thresholding approximation.

	We use the example of Section~\ref{sec:cg_convergence_on_simulated_data} to study the effects of preconditioning the system $\btilPhi \btilbeta = \bbTilde$ with the thresholding approximation $\btilPhi_{(k)}$. As before, the CG sampler is applied to the distribution \eqref{eq:beta_conditional} arising from the simulated data with 10 signals out of $p = 10{,}000$ predictors. Figure~\ref{fig:cg_convergence_plot_for_simulated_data_with_different_block_sizes} shows the results for $k = 1{,}000, 2{,}000, \text{and } 4{,}000$. The convergence rates of these preconditioned CG iterations are compared to that of the CG iterations applied to the prior-preconditioned system without any additional preconditioning. If a preconditioner is a good approximation of $\btilPhi$, the preconditioned CG should yield convergence in a very small number of iterations --- for example, a perfect approximation would induce the convergence after one iteration. It is clear from Figure~\ref{fig:cg_convergence_plot_for_simulated_data_with_different_block_sizes}, however, that the thresholding approximation does more harm than good in terms of the CG convergence rate, especially when $k$ is taken small relative to the size of $\btilPhi$. We can therefore conclude that the thresholding strategy yields a poor approximation except when $k$ starts to become almost as large as $p$.
	\begin{figure}
		\centering
		\includegraphics[width=.6\linewidth]{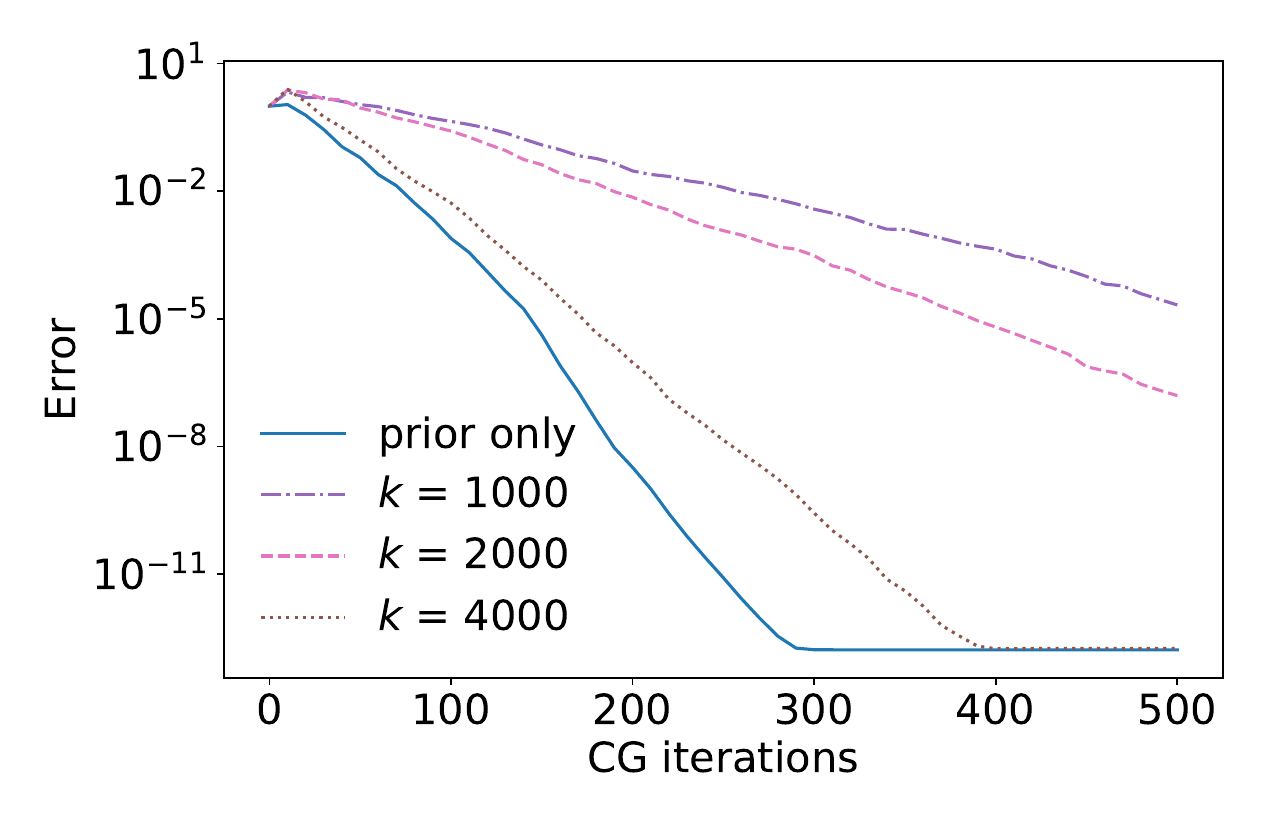}
		\vspace{-.5\baselineskip}
		\caption{Plot of the CG approximation errors (with the same error metric as used in Figure~\ref{fig:cg_convergence_plot_for_simulated_data}) as a function of the number of CG iterations. The CG sampler here is applied to the conditional distribution of $\bbeta$ arising from the simulated data example with 10 signals as described in Section~\ref{sec:cg_sampler_demonstration}. Each line corresponds to a different threshold level $k$ for the approximation \eqref{eq:threshold_approximation_of_precond_matrix}. The blue line labeled `prior only` corresponds to CG applied to the prior preconditioned system without any further preconditioner.}
		\label{fig:cg_convergence_plot_for_simulated_data_with_different_block_sizes}
	\end{figure}

	We repeated the same experiments on the thresholding approximation $\btilPhi_{(k)}$ using the other posterior distributions discussed in the manuscript. Across the experiments, no computational gain could be achieved by using the thresholding approximation as a preconditioner. More precisely, to yield a good enough approximation, the value of $k$ had to be so large that preparing the preconditioner itself became a computational bottleneck.

\FloatBarrier

\section{Proofs for Appendix~\ref{sec:theories_behind_cg}}
\label{supp:cg_theory_proofs}

\begin{proof}[Proposition~\ref{prop:cg_error_as_polynomial_sum}]
	As discussed in Section~\ref{sec:iterative_method_for_solving_linear_system}, the $k$-th CG iterate belongs to an affine space $\bbeta_0 + \mathcal{K}(\bPhi, \br_0, k)$ with $\br_0 = \bPhi \bbeta_0 - \bb =  \bPhi (\bbeta_0 - \bbeta)$. An element $\bbeta'$ of the affine space can be written as
	\begin{equation*}
		\bbeta'
		= \bbeta_0 + \sum_{\ell = 1}^k c_\ell \bPhi^{\ell - 1} \br_0
		= \bbeta + (\bbeta_0 - \bbeta) + \sum_{\ell = 1}^k c_\ell \bPhi^{\ell} (\bbeta_0 - \bbeta)
	\end{equation*}
	for some $c_1, \ldots, c_k$. In other words, for any $\bbeta'$ in the affine space we can write
	\begin{equation}
		\label{eq:polynomial_sum_representation}
		\bbeta' - \bbeta
		= Q_k(\bPhi) \left( \bbeta_0 - \bbeta \right)
	\end{equation}
	for some $Q_k \in \polyspace_k$. Together with the optimality property \eqref{eq:cg_optimality_property} of the CG iterates, the representation \eqref{eq:polynomial_sum_representation} implies
	\begin{equation*}
		\| \bbeta_k - \bbeta \|_{\bPhi}
		= \| R_k(\bPhi) (\bbeta_0 - \bbeta) \|_{\bPhi}
		= \min_{\bbeta' \in \bbeta_0 + \mathcal{K}(\bPhi, \br_0, k)} \| \bbeta' - \bbeta \|_{\bPhi}. \qedhere 
	\end{equation*}
\end{proof}

\begin{proof}[Theorem~\ref{thm:cg_bound_over_discrete_polynomial_values}]
	Let $\bv_1, \ldots, \bv_p$ be the unit eigenvectors of $\bPhi$ associated with the eigenvalues $\eigen_1, \ldots \eigen_p$. By the spectral theorem for normal matrices (Section 2.5 of \cite{horn2012matrix-analysis}), the unit eigenvectors form an orthonormal basis. In particular, we can write $\bbeta_0 - \bbeta = \sum_{j = 1}^p c_j \bv_j$ for $c_j = \left\langle \bbeta_0 - \bbeta, \bv_j \right\rangle$. Observe that, for any $Q_k \in \polyspace_k$,
	\begin{equation*}
		\begin{aligned}
			\bPhi^{1/2} Q_k(\bPhi) \left( \bbeta_0 - \bbeta \right)
			= \textstyle{\sum}_{j = 1}^p \, c_j \bPhi^{1/2} Q_k(\bPhi) \bv_j
			= \textstyle{\sum}_{j = 1}^p \, c_j \eigen_j^{1/2} Q_k(\eigen_j) \bv_j.
		\end{aligned}
	\end{equation*}
	Together with \eqref{eq:cg_error_bound_by_polynomial_sum}, the above equality yields
	\begin{equation}
		\label{eq:bound_over_eigenvalues}
		\left\| \bbeta_k - \bbeta \right\|_{\bPhi}^2
		\leq \sum_{j = 1}^p c_j^2 \eigen_j Q_k(\eigen_j)^2
		\leq \max_{j = 1, \ldots, p} Q_k(\eigen_j)^2 \Bigg( \sum_{j = 1}^p c_j^2 \eigen_j \Bigg).
	\end{equation}
	The result \eqref{eq:cg_bound_over_discrete_polynomial_values} follows from the above inequality since $\left\| \bbeta_0 - \bbeta \right\|_{\bPhi}^2 = \sum_{j = 1}^p c_j^2 \eigen_j$.

	The sharpness of the upper bound is proven by explicitly constructing an initial vector that achieves the bound; see \cite{greenbaum1979cg_splittings}.
\end{proof}

\begin{proof}[Theorem~\ref{thm:polynomial_bound_over_interval}]
	We can construct a shifted and scaled Chebyshev polynomial $P_k \in \polyspace_k$ such that $|P_k(\eigen)|$ is bounded by the right-hand side of \eqref{eq:polynomial_bound_over_interval} over the interval $[ \eigen_{\textrm{min}},  \eigen_{\textrm{max}} ]$. See \cite{saad2011methods-for-eigenvalue} for further details.
\end{proof}

\begin{proof}[Theorem~\ref{thm:cg_convergence_delay_by_large_eigenvalues}]
	Let $Q_{k}^*$ denote the minimizer of $\max_{j = r + 1, \ldots, p} |Q_{k}(\eigen_j)|$ over $\polyspace_k$ and define
	\begin{equation*}
		Q_{r + k}'(\eigen)
		= Q_k^* (\eigen) \prod_{i = 1}^{r} \left(\frac{\eigen_i - \eigen}{\eigen_i} \right).
	\end{equation*}
	Then $Q_{r + k}'$ satisfies $Q_{r + k}'(\eigen_j) = 0$ for $j = 1, \ldots, r$ and $Q_{r + k}'(\eigen_j) \leq Q_k^* (\eigen_j)$ for $j = r + 1, \ldots, p$. In particular, $Q_{r + k}'$ satisfies
	\begin{equation*}
		\max_{j = 1, \ldots, p} |Q_{r + k}'(\eigen_j)|
		\leq \max_{j = r + 1, \ldots, p} |Q^*_{k}(\eigen_j)|
		= \min_{Q_{k} \in \polyspace_k} \,  \max_{j = r + 1, \ldots, p} |Q_{k}(\eigen_j)|,
	\end{equation*}
	where the equality holds as we chose $Q_{k}^*$ to be the minimizer. From the above inequality, it follows that
	\begin{equation}
		\label{eq:bound_by_lower_order_polynomial}
		\min_{Q_{r + k} \, \in \, \polyspace_{r + k}} \,  \max_{j = 1, \ldots, p} |Q_{r + k}(\eigen_j)|
		\leq \max_{j = 1, \ldots, p} |Q_{r + k}'(\eigen_j)|
		\leq \min_{Q_{k} \in \polyspace_k} \,  \max_{j = r + 1, \ldots, p} |Q_{k}(\eigen_j)|.
	\end{equation}
	Since the maximum taken over an interval $[\eigen_p, \eigen_{r + 1}]$ is larger than that over its subset, \eqref{eq:bound_by_lower_order_polynomial} implies
	\begin{equation}
		\label{eq:bound_by_lower_order_polynomial_over_interval}
		\min_{Q_{r + k} \, \in \, \polyspace_{r + k}} \,  \max_{j = 1, \ldots, p} |Q_{r + k}(\eigen_j)|
		\leq \min_{Q_{k} \in \polyspace_k} \,  \max_{\eigen \in [\eigen_p, \eigen_{r + 1}]} |Q_{k}(\eigen)|.
	\end{equation}
	The desired inequality \eqref{eq:cg_convergence_removal_of_largest_eigenvalues} follows by bounding the right-hand side of \eqref{eq:bound_by_lower_order_polynomial_over_interval} via Theorem~\ref{thm:polynomial_bound_over_interval}.
\end{proof}

\begin{proof}[Theorem~\ref{thm:cg_convergence_extreme_eigenvalue_removal}]
	We first prove the bound \eqref{eq:cg_convergence_extreme_eigenvalue_removal} for $C_{\nRitzIter, \nlargest, \nsmallest}$ as defined in \eqref{eq_for_proof:const_definition} below. Let $R_\nRitzIter$ be the optimal CG polynomial at the $\nRitzIter$-th iteration as defined in \ref{eq:cg_error_as_polynomial_sum}. Since $R_\nRitzIter(0) = 1$, the polynomial $R_\nRitzIter(\eigen)$ can be expressed in terms of its roots $\ritz_1^{(\nRitzIter)}, \ldots, \ritz_\nRitzIter^{(\nRitzIter)}$ as
	\begin{equation*}
		R_\nRitzIter(\eigen)
		= \prod_{i = 1}^{\nRitzIter} \left(
		\frac{\ritz_i^{(\nRitzIter)} - \eigen}{\ritz_i^{(\nRitzIter)}}
		\right).
	\end{equation*}
	Now consider $Q_\nRitzIter \in \polyspace_\nRitzIter$ such that
	\begin{equation*}
		Q_\nRitzIter(\eigen)
		= \prod_{i = 1}^\nlargest \left( \frac{\eigen_{i} - \eigen}{\eigen_{i}} \right)
		\prod_{i = 0}^{\nsmallest - 1} \left( \frac{\eigen_{p - i} - \eigen}{\eigen_{p - i}} \right)
		\prod_{i = \nlargest + 1}^{\nRitzIter - \nsmallest} \left( \frac{\ritz_i^{(\nRitzIter)} - \eigen}{\ritz_i^{(\nRitzIter)}} \right)
	\end{equation*}
	and define
	\begin{equation}
		\label{eq_for_proof:const_definition}
		C_{\nRitzIter, \nlargest, \nsmallest}
		= \max_{j = \nlargest + 1, \, \ldots, \, p - \nsmallest} Q_\nRitzIter(\eigen_j) / R_\nRitzIter(\eigen_j).
	\end{equation}

	As in the proof of Theorem~\ref{thm:cg_bound_over_discrete_polynomial_values}, write $\bbeta_0 - \bbeta = \sum_{j = 1}^p c_j \bv_j$ so that $\bbeta_\nRitzIter - \bbeta = \sum_{j = 1}^p c_j R_\nRitzIter(\eigen_j) \bv_j$. Let $\bbeta_\nRitzIter'$ be a modification of $\bbeta_\nRitzIter$ such that
	\begin{equation}
		\label{eq_for_proof:modified_kth_iterate}
		\bbeta_\nRitzIter' - \bbeta
		= \sum_{j = \nlargest + 1}^{p - \nsmallest} c_j R_\nRitzIter(\eigen_j) \bv_j.
	\end{equation}
	Choose $R_\nSubseqIter' \in \polyspace_\nSubseqIter$ to be a minimizer of $\| Q_\nSubseqIter(\bPhi) (\bbeta_{\nRitzIter}' - \bbeta) \|_{\bPhi}$ over $Q_\nSubseqIter \in \polyspace_\nSubseqIter$ or, equivalently, the optimal CG polynomial \eqref{eq:cg_error_as_polynomial_sum} at the $\nSubseqIter$-th iteration  when the initial vector is taken to be $\bbeta_\nRitzIter'$. Since $\bbeta_{\nRitzIter + \nSubseqIter}$ minimizes the $\bPhi$-norm over the $(k + \ell)$-th polynomial of degree by Proposition~\ref{prop:cg_error_as_polynomial_sum},
	we have
	\begin{equation}
		\label{eq_for_proof:bound_by_restarted_cg_polynomial}
		\| \bbeta_{\nRitzIter + \nSubseqIter} - \bbeta \|_{\bPhi}
		\leq \| R_{\nSubseqIter}'(\bPhi) Q_\nRitzIter(\bPhi) (\bbeta_{0} - \bbeta) \|_{\bPhi}.
	\end{equation}

	We will now show that the right-hand side of \eqref{eq_for_proof:bound_by_restarted_cg_polynomial} is bounded above by that of \eqref{eq:cg_convergence_extreme_eigenvalue_removal}. By our definition of $\bbeta_\nRitzIter'$ and $C_{\nRitzIter, \nlargest, \nsmallest}$ in \eqref{eq_for_proof:modified_kth_iterate} and  \eqref{eq_for_proof:const_definition}, we have
	\begin{equation*}
		\begin{aligned}
			\| R_{\nSubseqIter}'(\bPhi) Q_\nRitzIter(\bPhi) (\bbeta_{0} - \bbeta) \|_{\bPhi}^2
			&= \sum_{j = \nlargest + 1}^{p - \nsmallest} \eigen_j c_j^2 R_\nSubseqIter'(\eigen_j)^2 Q_\nRitzIter(\eigen_j)^2 \\
			&\leq C_{\nRitzIter, \nlargest, \nsmallest}^2 \sum_{j = \nlargest + 1}^{p - \nsmallest} \eigen_j c_j^2 R_\nSubseqIter'(\eigen_j)^2 R_\nRitzIter(\eigen_j)^2 \\
			&= C_{\nRitzIter, \nlargest, \nsmallest}^2 \| R_{\nSubseqIter}'(\bPhi) (\bbeta_{\nRitzIter}' - \bbeta) \|_{\bPhi}^2.
		\end{aligned}
	\end{equation*}
	Noting that $\| \bbeta_\nRitzIter' - \bbeta \|_{\bPhi} \leq \| \bbeta_\nRitzIter - \bbeta \|_{\bPhi}$, we obtain
	\begin{equation}
		\label{eq_for_proof:bounding_restarted_cg_polynomial}
		\| R_{\nSubseqIter}'(\bPhi) Q_\nRitzIter(\bPhi) (\bbeta_{0} - \bbeta) \|_{\bPhi}
		\leq C_{\nRitzIter, \nlargest, \nsmallest} \frac{
			\| R_{\nSubseqIter}'(\bPhi) (\bbeta_{\nRitzIter}' - \bbeta) \|_{\bPhi}
		}{
			\| \bbeta_\nRitzIter' - \bbeta \|_{\bPhi}
		}
		\| \bbeta_{\nRitzIter} - \bbeta \|_{\bPhi}.
	\end{equation}

	By our choice of $R_{\nSubseqIter}'$, the vector $R_{\nSubseqIter}'(\bPhi) (\bbeta_{\nRitzIter}' - \bbeta)$ coincides with the residual of the $\nSubseqIter$-th CG iterate starting from the initial vector $\bbeta_{\nRitzIter}'$. Therefore, by Lemma~\ref{lem:cg_bound_over_subset_eigenvalues} combined with Theorem~\ref{thm:polynomial_bound_over_interval}, we have
	\begin{equation}
		\label{eq_for_proof:bound_over_interavl_of_restarted_cg}
		\frac{
			\| R_{\nSubseqIter}'(\bPhi) (\bbeta_{\nRitzIter}' - \bbeta) \|_{\bPhi}
		}{
			\| \bbeta_\nRitzIter' - \bbeta \|_{\bPhi}
		}
		\leq 2 \left( \frac{
			\sqrt{\eigen_{\nlargest + 1} / \eigen_{p - \nsmallest}} - 1
		}{
			\sqrt{\eigen_{\nlargest + 1} / \eigen_{p - \nsmallest}} + 1
		} \right)^\nSubseqIter.
	\end{equation}
	The claimed inequality \eqref{eq:cg_convergence_extreme_eigenvalue_removal} now follows from \eqref{eq_for_proof:bound_by_restarted_cg_polynomial}, \eqref{eq_for_proof:bounding_restarted_cg_polynomial}, and \eqref{eq_for_proof:bound_over_interavl_of_restarted_cg}.

	Now we turn to proving the claimed property of $C_{\nRitzIter, \nlargest, \nsmallest}$. Note that
	\begin{equation*}
		\frac{Q_\nRitzIter(\eigen)}{R_\nRitzIter(\eigen)}
		=
		\prod_{i = 1}^{\nlargest} \frac{\ritz^{(\nRitzIter)}_i}{\eigen_i}
		\left( \frac{\eigen_i - \eigen}{\ritz^{(\nRitzIter)}_i - \eigen} \right) \,
		\prod_{i = 0}^{\nsmallest - 1} \frac{\ritz^{(\nRitzIter)}_{\nRitzIter-i}}{\eigen_{p-i}}
		\left( \frac{\eigen_{p-i} - \eigen}{\ritz^{(\nRitzIter)}_{\nRitzIter-i} - \eigen} \right).
	\end{equation*}
	The rest of the proof focuses on the case $\nlargest = 1$ and $\nsmallest = 0$ for clarity's sake; the proof remains essentially identical in the general case except for extra notational clutters. Under this case, we have
	\begin{equation*}
		\max_{j = 2, \, \ldots, \, p} \left| \frac{Q_\nRitzIter(\eigen_j)}{R_\nRitzIter(\eigen_j)} \right|
		=
		\frac{\ritz^{(\nRitzIter)}_{1}}{\eigen_{1}}
		\max_{j = 2, \, \ldots, \, p}
		\left| \frac{\eigen_{1} - \eigen_j}{\ritz^{(\nRitzIter)}_{1} - \eigen_j} \right|
		=
		\frac{\ritz^{(\nRitzIter)}_{1}}{\eigen_{1}}
		\max_{j = 2, \, \ldots, \, p}
		\left| 1 - \frac{\ritz^{(\nRitzIter)}_{1} - \eigen_1}{\eigen_j - \eigen_{1}} \right|^{-1}.
	\end{equation*}
	Provided $|\ritz^{(\nRitzIter)}_{1} - \eigen_1| = \min_{j = 1, \ldots, p} |\ritz^{(\nRitzIter)}_{1} - \eigen_j|$, the above inequality simplifies to
	\begin{equation*}
		\max_{j = 2, \, \ldots, \, p} \left| \frac{Q_\nRitzIter(\eigen_j)}{R_\nRitzIter(\eigen_j)} \right|
		=
		\frac{\ritz^{(\nRitzIter)}_{1}}{\eigen_{1}}
		\left| 1 - \frac{\ritz^{(\nRitzIter)}_{1} - \eigen_1}{\eigen_2 - \eigen_{1}} \right|^{-1}.
	\end{equation*}
	So we have $C_{\nRitzIter, \nlargest, \nsmallest} \to 1$ as $\ritz^{(\nRitzIter)}_{1} \to \eigen_1$ in the case $\nlargest = 1$ and  $\nsmallest = 0$.
\end{proof}

\begin{lemma}
	\label{lem:cg_bound_over_subset_eigenvalues}
	Let $(\eigen_j, \bv_j)$ for $j = 1, \ldots, p$ denote the eigenvalue and eigenvector pairs of $\bPhi$. If the initial vector $\bbeta_0$ satisfies $\left\langle \bbeta_0 - \bbeta, \bv_{j} \right\rangle = 0$ for $j \in J \subset \{1, \ldots, p\}$, then the bound \eqref{eq:cg_bound_over_discrete_polynomial_values} holds over the set $\{1, \ldots, p\} \setminus J$ i.e.\
	\begin{equation*}
		\frac{
			\| \bbeta_k - \bbeta \|_{\bPhi}
		}{
			\| \bbeta_0 - \bbeta \|_{\bPhi}
		}
		\leq \min_{Q_k \, \in \, \polyspace_k} \max_{j \not \in J} | Q_k(\eigen_j) |.
	\end{equation*}
\end{lemma}

\begin{proof}
	The proof is identical to that of Theorem~\ref{thm:cg_bound_over_discrete_polynomial_values} except that we can replace the bound \eqref{eq:bound_over_eigenvalues} with
	\begin{equation*}
		\left\| \bbeta_k - \bbeta \right\|_{\bPhi}^2
		\leq \sum_{j \not \in J} c_j^2 \eigen_j Q_k(\eigen_j)^2
		\leq \max_{j \not \in J} Q_k(\eigen_j)^2 \Bigg( \sum_{j \not \in J} c_j^2 \eigen_j \Bigg)
	\end{equation*}
	since $c_j = 0$ for $j \in J$ by assumption.
\end{proof}

\FloatBarrier

\bibliographystylesupplement{agsm}
\bibliographysupplement{cg_accelerated_gibbs}{}

\end{document}